\documentclass[11pt,a4paper]{article}
\usepackage{setspace}
\usepackage[english]{babel}
\usepackage{amsmath,amssymb,amsfonts}
\usepackage[gen,right]{eurosym}
\usepackage{textcomp,longtable,epsfig,url}
\usepackage{bbm}
\usepackage{graphicx}
\usepackage[intoc]{nomencl}
\usepackage[printonlyused]{acronym}
\usepackage{verbatim}
\usepackage{tabularx}
\usepackage{url}
\usepackage{color}
\usepackage{xcolor}
\usepackage{setspace}
\usepackage{scrhack}
\usepackage{floatrow}
\usepackage{listings}
\usepackage{amsthm}
\usepackage{paralist}
\usepackage{xr}
\usepackage{graphicx}
\usepackage{listings}
\usepackage{booktabs}
\usepackage{flafter}
\usepackage{varwidth}
\usepackage[round]{natbib}
\usepackage[a4paper]{geometry}
\usepackage{enumitem}
\usepackage{titling}
\usepackage{subcaption}

\makeindex

\newtheorem{definition}{Definition}[section]
\newtheorem{assumption}{Assumption}
\newtheorem{proposition}[definition]{Proposition}
\newtheorem{lemma}[definition]{Lemma}

\newtheorem{theorem}[definition]{Theorem}

\newtheorem{remark}[definition]{Remark}
\usepackage{natbib}

\newcommand{\beqo}{\begin{eqnarray*}}
\newcommand{\eeqo}{\end{eqnarray*}\noindent}
\newcommand{\beq}{\begin{eqnarray}}
\newcommand{\eeq}{\end{eqnarray}\noindent}

\newcommand{\md}{\mathrm{d}}

\newcommand{\mE}{\mathbb{E}}
\newcommand{\mR}{\mathbb{R}}

\usepackage{tikz}
\usetikzlibrary{arrows.meta, shapes.geometric, positioning, calc}
\usetikzlibrary{fit}
\usepackage[utf8]{inputenc}

\title{Retirement Planning with Minimum Guarantees for Both Beneficiaries and Fund Managers}
\date{October 5, 2026}
\author{
Zhaojie Ren
\thanks{Department of Statistics and Data Science, The Chinese University of Hong Kong,
Shatin, N.T., Hong Kong SAR
(\url{rzj20@tsinghua.org.cn}).}
\and
Sheng Wang
\thanks{Department of Statistics and Actuarial Science, School of Computing and Data Science,
The University of Hong Kong, Pokfulam Road, Hong Kong SAR
(\url{sheng-wa15@tsinghua.org.cn}).}
\and
Tak Kwong Wong
\thanks{School of Mathematical Sciences, Shenzhen University,
518061 Shenzhen, P.\,R. China
(\url{takkwong@szu.edu.cn}).}
\and
Sheung Chi Phillip Yam
\thanks{Department of Statistics and Data Science, The Chinese University of Hong Kong,
Shatin, N.T., Hong Kong SAR
(\url{scpyam@sta.cuhk.edu.hk}).}
}

\usepackage{etoolbox}
\usepackage[colorlinks=true,linkcolor=blue,citecolor=blue]{hyperref}
\setcitestyle{numbers,square}
\numberwithin{equation}{section}
\newcommand{\keywords}[1]{\par\medskip\noindent\textbf{Keywords: }\begingroup\renewcommand{\and}{\unskip;\space}#1\endgroup\par}
\newcommand{\subclass}[1]{\par\smallskip\noindent\textbf{Mathematics Subject Classification: }\begingroup\renewcommand{\and}{\unskip;\space}#1\endgroup\par}
\AfterEndEnvironment{abstract}{\setstretch{1.15}\allowdisplaybreaks}

\begin{document}
\maketitle
\begin{abstract}
This article studies a retirement planning problem from a new perspective in which a retiree delegates an initial lump sum to a professional fund manager. The fund is managed dynamically to deliver lifelong benefits while satisfying a guarantee that, at all times, wealth remains above a prescribed solvency level and the benefit rate remains above a minimum level. We formulate this problem as a continuous-time stochastic control model, in which the controls are the benefit rate and the portfolio allocation, and the objective incorporates general utility functions over benefits, management fees, and terminal bequest. Conventional approaches such as standard viscosity solution and martingale methods are not immediately applicable here due to the simultaneous influence of (i) the control variable nature of the benefit rate; (ii) the guarantee constraints; and (iii) the running wealth-dependent utility. We develop a new approach that deals directly with the associated fully nonlinear Hamilton–Jacobi–Bellman equation and establish the unique existence of a classical solution.  One particular trick
for resolving this equation is to apply an intricate transform that diverts
the former to a semilinear auxiliary equation. A major contribution is to establish the
existence of the solution to the auxiliary equation and characterize its precise growth,
derivative, and boundary behavior. These results are crucial for establishing the verification theorems and deriving the optimal strategy. Besides, numerical experiments indicate that guarantee constraints induce more conservative policies, particularly at low wealth levels, while providing effective downside protection for both wealth and benefits. The associated well-being loss relative to the unconstrained framework remains small, suggesting that incorporating guarantees into retirement products is practically appealing.
\subclass{	49L12 \and 	91G05  \and 	91G10 }

\end{abstract}

\section{Introduction}
The transition from wealth accumulation in pension schemes to post-retirement decumulation poses a fundamental challenge for  retiring individuals, namely, how to ensure retirement income and maintain well-being in the retirement journey. 
Upon retirement, individuals typically hold a certain amount of  endowment that should be handled wisely for the rest of their uncertain lives. Moreover, since retirees often seek higher returns to sustain consumption over long retirement horizons, particularly given the rise of life expectancy, participation in risky financial assets such as equities is typically unavoidable. This exposes retirement wealth to significant market volatility along with longevity risk of the holder.
Traditional defined-benefit pension schemes transfer these risks from individuals to pension providers by providing a stable lifelong income stream; see \cite{booth2020modern}. However, such arrangements have become increasingly less prevalent in modern pension systems amid population aging; see, for example, \cite*{OECD2023Pensions, OECD2025Pensions}. With this urging issue, more individuals demand  professional management for their retirement wealth, and a range of retirement products has emerged;  see \cite*{chenmil,moenig2016revisiting, ici2025factbook}  for instance. These products  combine the flexibility of individual account management with the minimum guarantee of 
lifetime income.

A substantial amount of literature has been devoted to investigating various mechanisms for transforming accumulated retirement wealth into sustainable income streams. One prevalent approach is through the vehicle of life annuities, which provide stable lifelong income by fully insuring  longevity risk, but fail to hedge out inflation risk; see, e.g., \cite*{yaari1965uncertain,lockwood2012bequest}. 
One can imagine several possible extensions, especially those that can
 address specific retiree preferences.
For instance, variable annuities allow retirees to participate in financial market returns, but this upside potential comes with downside risk, which naturally  incorporates  guarantee features to balance risk and return; see, e.g., \cite*{chenmil,moenig2016revisiting}. 

Another  popular approach is to embed retirement income features in mutual funds, commonly referred to as retirement income funds; also see \cite*{ici2025factbook}. 
A number of such products are available in the market, for example, the {\it Vanguard Target Retirement Income Fund} \cite*{vanguard_vtinx}, the {\it Schwab Monthly Income Funds} \cite*{schwab_income_products}, and the {\it American Funds Retirement Income Portfolio} \cite*{americanfunds_fafwx}. 
These products   provide systematic withdrawals and diversified exposure to financial markets within the  common business jurisdiction. In contrast to annuity-based arrangements, such retirement income funds typically do not offer explicit benefit guarantees, leaving retirees fully exposed to both longevity and market risks. To fill this gap, we develop a baseline mathematical framework that incorporates guarantee features while maintaining market participation, providing theoretical guidance for the design of retirement income funds with both upside exposure to financial markets and downside protection.

\subsection{A Model of Retirement Planning with Minimum Guarantees}
In this article, we consider a framework in which the retiree delegates the management of the initial lump sum \(x_0\) to a professional fund manager.
In practice, the fund managers  may collect information about  clients such as their risk appetite. One possible way for obtaining such information is, in the era of big data,  through accumulated data gained by serving a large number of retirees with broadly similar backgrounds in terms of risk preferences and retirement needs. Accordingly, they are  able to put themselves in the shoes of retirees with a sufficiently accurate understanding of  utility functions of clients.

The fund is invested in a portfolio comprising a risk-free asset and a risky asset, and its value \(X(t)\) evolves randomly, subject to two constraints. First, a capital floor ensures that the fund value never falls below a fixed proportion of the initial contribution, i.e., \(X(t) \ge \alpha x_0\). This floor captures a regulatory or solvency requirement: the fund holds sufficient assets as a buffer against liquidity risk, ensuring that it can meet its obligations without relying  much on exogenous financial sources; see, e.g., \cite{di2011pension}, \cite{broeders2013pension} and \cite*{avanzi2024optimal}. Moreover, it also guarantees a lower bound on the bequest. Second, a benefit floor guarantees the retiree a lifelong  steady benefit stream \(b(t)\) that never drops below a minimum level \(b_0\), i.e., \(b(t) \ge b_0\). Together, these two constraints limit the fund manager’s risk-taking incentives, preventing the shifting of downside risks at  the expense of the retiree. In exchange for managing the fund and providing these guarantees, the manager charges proportional fees consisting of a running component \(\kappa_1 X(t)\) and a terminal component \(\kappa_2 X(\tau_d)\); see, for example, \cite*{elton2003incentive}, \cite*{khorana2009mutual}, \cite*{guasoni2015hedge} and \cite*{bensoussan2022inter}. As illustrated in Figure~\ref{fig:placeholder}, this arrangement transforms the retiree’s initial wealth into a stable lifetime benefit stream, protected by the benefit floor, while the capital floor ensures a residual amount of at least $(1-\kappa_2)\alpha x_0$, which is ultimately passed to the descendant as bequest  upon death at \(\tau_d\).

\begin{figure}[]
	\centering
	\caption{Structure of the retirement planning model with guarantees. Closely related arrangements include pension funds with minimum guarantees \cite{di2011pension}, and income-drawdown plans with minimum guarantees \cite{di2014income}}
	\label{fig:placeholder}
	
	\begin{tikzpicture}[
		node distance=3cm,
		every node/.style={align=center, font=\small},
		box/.style={rectangle, draw, thick, minimum width=2cm, minimum height=1.2cm, rounded corners},
		arrow/.style={-{Stealth[scale=1.2]}, thick}
		]
		
		\node[box, fill=blue!10] (retiree) {Retiree};
		\node[box, fill=green!10, right=of retiree] (manager) {Manager};
		\node[box, fill=yellow!10, right=6cm of manager] (fund) {Fund Value $X$};
		\node[box, fill=orange!10, below=1cm of fund, xshift=-8mm] (market) {Market\\ (Risk-free + Risky assets)};
		
		\node[box, fill=purple!10, above=1cm of retiree, xshift=3mm] (wealth) {Wealth\\Accumulation};
		
		\draw[arrow] (retiree) -- node[above] {initial wealth $x_0$} (manager);
		
		\draw[arrow] (manager) -- node[above] {control $(\pi, b)$} (fund);
		
		\draw[arrow] (fund) to[bend left=20]
		node[below] {benefit $b(t)$, bequest $(1-\kappa_2)X(\tau_d)$}
		(retiree);
		
		\draw[arrow] (fund) to[bend right=20]
		node[above] {fees $\kappa_1 X(t),\ \kappa_2 X(\tau_d)$}
		(manager);
		
		\draw[arrow] (fund.south) -- node[left] {investment $\pi$} ([xshift=8mm]market.north);
		\draw[arrow] ([xshift=8mm]market.north) -- node[right] {returns} (fund.south);
		
		\draw[arrow] ([xshift=-3mm]wealth.south) -- node[right] {lump sum at retirement} (retiree);
		
		\node (constraints) at ($(manager)!0.5!(fund) + (-0.2,-0.35cm)$)
		{constraints $X(t)\ge \alpha x_0,\ b(t)\ge b_0$};
		
	\node[
	draw=black,
    line width=1.2pt,
	rounded corners,
	inner sep=-3pt,
	fit=(constraints)
	] {};
		
	\end{tikzpicture}
\end{figure}
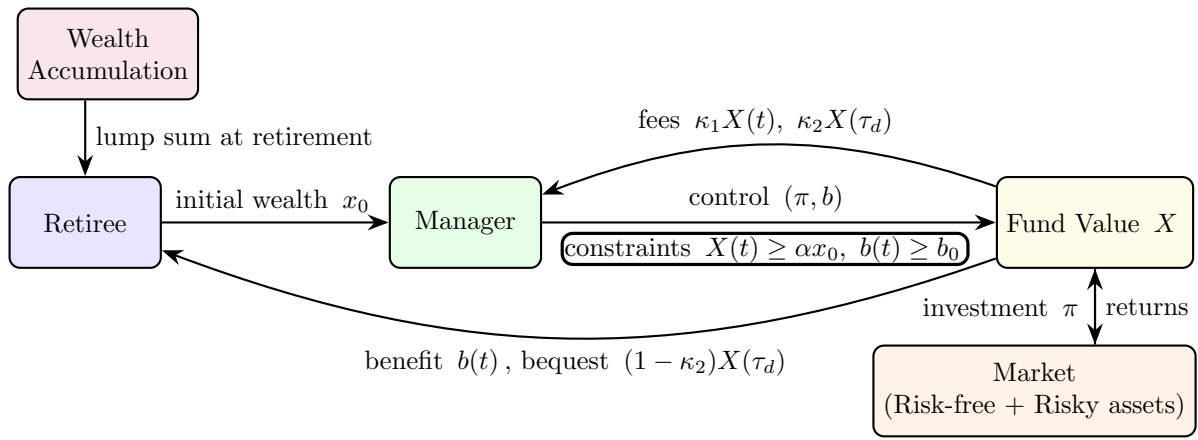

Formally, the evolution of the fund’s value $X(\cdot)$ is described by the following  stochastic differential equation (SDE):
\begin{align}\label{eq:sde:x:intro}
\left\{
\begin{aligned}
    &\mathrm{d} X(t)=[r_1+(\mu-r_1)\pi(t)]X(t)\mathrm{d}t-b(t)\mathrm{d}t-\kappa_1X(t)\md t+\sigma\pi(t)X(t)\mathrm{d}W(t),\\
    &X(0)=x_0,
    \end{aligned}
    \right.
\end{align}
where $r_1$ is the return rate of the long-term riskless asset available in the market to the manager but not to the individual clients,
$\mu$ and $\sigma$ are the expected return and
the volatility of the risky asset respectively, and $\pi(t)$ is the  proportion of wealth allocated to the risky asset at time $t$.
The manager’s objective is
\begin{align}
&\max_{(b,\pi)}
\mathbb{E}\bigg[
\int_0^{\tau_d} e^{-\rho t} U_{c,1}(b(t))\,\mathrm{d}t+e^{-\rho \tau_d}U_{c,2}((1-\kappa_2)X(\tau_d))
 \nonumber\\
&\quad\quad\quad\quad+\int_0^{\tau_d} e^{-\rho t} U_{m,1}(\kappa_1 X(t))\,\mathrm{d}t+e^{-\rho \tau_d}U_{m,2}(\kappa_2X(\tau_d))
\bigg],\label{eq:problem:max:intro}\\
&\text{subject to} \quad b(t) \geq b_0 \,\,\, \text{and} \,\,\, X(t) \geq \alpha x_0\quad \text{for a.e. } t\geq0, \text{ a.s.}.\label{eq:constraint:intro}
\end{align}
 The objective function integrates the preferences of both the retiree and the fund manager: the first and second terms correspond to the retiree’s utilities from the benefit and bequest, respectively, while the third and fourth terms capture the fund manager’s own utilities from running and terminal management fees, respectively. 

As discussed above, the fund manager has  sufficiently rich information about clients’ preferences and risk attitudes, which enables her/him to form a reasonably accurate perception of clients’ utility functions. 
In addition, taking greater account of retirees’ utilities enables the manager to gain reputation and credentials, which in turn triggers more future business opportunities.
Together, these considerations provide a rationale for modeling the fund manager's objective that internalizes both the retiree’s welfare and her/his own compensation in investment and benefit decisions on behalf of clients. Moreover, to accommodate different clients, whose utility functions can vary from one to another, we adopt a general utility specification. This framework also allows the manager to implicitly adjust the relative importance (or weights) assigned to her/his own utility and that of the retiree through the choice of utility functions $U_{c,1}, U_{c,2},U_{m,1}$ and $U_{m,2}$.

\subsection{Related Works}

The incorporation of constraints on controls and state variables has long been an interesting topic in stochastic control theory, particularly in mathematical finance and economic modeling;  more importantly, treating these two constraints together is generally more challenging.
For the classical  consumption and investment problems, a substantial amount of work focuses on constraints imposed on consumption, often through habit formation or ratcheting and drawdown requirements that confine admissible consumption paths; see \cite*{constantinides1990habit}, \cite*{dybvig1995dusenberry}, and \cite*{van2020consumption} for instance. These consumption constraints are closely related to the benefit floor \(b(t)\ge b_0\) in our retirement income setting, where the benefit rate plays the role of consumption. For completeness, we note that another line of research has focused on portfolio constraints such as no short selling; see \cite*{zariphopoulou1994consumption} and \cite*{czichowsky2012convex,czichowsky2013cone} for instance.  These portfolio constraints are not directly imposed in our model. Instead, the two minimum guarantees (benefit floor and wealth floor) naturally limit the investment behavior. As a consequence,  although no short-selling constraint is imposed \emph{a priori}, the optimal investment strategy is shown to satisfy the no short-selling property under general utility settings; see Equation \eqref{eq:b^*} in Theorem \ref{thm:verification}. This observation is particularly meaningful, as it demonstrates that the no short-selling feature emerges endogenously from the very nature of the problem.

Another important line of the literature considers constraints on the state variable, typically the wealth process in mathematical finance. These constraints can be broadly classified into two categories: terminal constraints and pathwise (or running) constraints. Many works focus on terminal constraints, requiring the terminal wealth to remain above a prescribed level; representative studies include \cite*{el2005optimal} and \cite*{avanzi2024optimal}. Problems of this type are often tackled using the martingale method for general utility functions, or via dynamic programming, where the use of specific functional forms of the utility functions allows one to employ the {\it ansatz} approach.

Closely related to our setting, but relatively less explored in the literature, are problems with running constraints on wealth, where the objective is to maximize a running wealth-dependent utility. Notable exceptions include \cite*{di2011pension} and \cite*{federico2011stochastic}, which investigate the optimal management of pension funds under minimum solvency constraints.  
Owing to the presence of such floor constraints together with wealth-dependent utility, the martingale approach is not immediately  applicable, and \cite*{di2011pension} instead relies on the theory of viscosity solutions. In particular, that study shows that the value function is a viscosity solution to the corresponding  Hamilton-Jacobi-Bellman (HJB) equation, and further establishes regularity properties of the value function. In the delay setting, \cite*{federico2011stochastic} establishes an equivalent infinite-dimensional formulation and continuity of the value function, providing a basis for subsequent HJB analysis. These works provide important insights into portfolio choice under capital guarantees and highlight the role of solvency considerations in long-term fund management. Nevertheless, the continuation of this line of research appears to be sluggish for a while,
until more recent contributions such as \cite*{federico2015utility}, \cite*{bensoussan2016unemployment} and \cite*{bensoussan2022inter}, which develop tools for handling general utility functions and thereby substantially broaden the analytical tractability of the discipline. 

Compared with  \cite*{di2011pension} and \cite*{federico2011stochastic}, where the benefit rate is assumed to be constant or determined by a prescribed rule, our model treats the benefit stream as a control variable. This feature allows for dynamic adjustment of withdrawals in response to market conditions, thereby enabling participation in favorable market returns. We emphasize that the martingale method, as well as the viscosity solution framework, are not directly applicable in our setting; see Appendix~\ref{sec:discussion} for a detailed discussion. Instead, we employ the dynamic programming approach, together with the transform techniques developed in \cite*{bensoussan2016unemployment} and \cite*{bensoussan2022inter}, to establish the existence of a classical solution to the associated HJB equation \eqref{eq:hjb}. For completeness, we also refer to related works on running constraints. Representative contributions in the deterministic framework include \cite*{federico2010hjb,federico2011hjb} and \cite*{calvia2024optimal}, while contributions in the stochastic control setting include \cite*{elie2008optimal} and \cite*{di2014income}, in which explicit solutions via dynamic programming in  specific settings are obtained.

\subsection{Contributions}
\subsubsection{Our Methodology}
To solve the problem \eqref{eq:problem:max:intro} under the wealth dynamics \eqref{eq:sde:x:intro} and subject to the constraints of \eqref{eq:constraint:intro}, we adopt a dynamic programming approach and formulate the associated stochastic control problem \eqref{eq:valuefunction:def} under \eqref{eq:sde:x} and \eqref{eq:mathcalA:def} below. The corresponding value function, denoted by $V$,  is defined on the state space $[\underline{x},\infty)$, where
 $
\underline{x}:=\max\left\{\alpha x_0,\frac{b_0}{r}\right\}$ and $r:=r_1-\kappa_1$.  \footnote{Although the solvency constraint only requires \(X(t)\ge\alpha x_0\), the benefit floor \(b(t)\ge b_0\) is self-financing only if wealth covers its capitalized value \(\frac{b_0}{r}\). The effective lower bound is thus \(\underline{x}=\max\{\alpha x_0,\frac{b_0}{r}\}\); see Lemma \ref{lma:b_0/r}. Accordingly, we consider the value function $V$ on the admissible state space $[\underline{x},\infty)$.} The HJB equation associated with this problem is given by \eqref{eq:hjb} (equivalently, \eqref{eq:hjb:simplify}).
The analysis then reduces to two main tasks: (i) proving a verification theorem to identify sufficient conditions under which a  solution to the HJB equation coincides with the value function and characterizes the optimal strategy; and (ii) establishing the existence of a solution to the HJB equation satisfying the required regularity, growth, and monotonicity properties. 
Both tasks are highly nonstandard and technically challenging, as they involve the simultaneous presence of state and control constraints, whereas handling even a single constraint can already be difficult, giving rise to intricate boundary behavior at $\underline{x}$.

To establish the verification theorem, several conditions must be imposed on a candidate value function $v$. In particular, besides satisfying the HJB equation, the candidate optimal strategy induced by $v$ must correspond to a wealth process whose associated SDE admits a unique strong solution, with coefficients satisfying appropriate linear growth conditions. Moreover, suitable boundary conditions must be imposed on $v$ and its derivatives at the boundary $\underline{x}$, which depend on the control admissible set corresponding to the initial wealth   $\underline{x}$ and, in turn, are determined by the relative magnitude of $\alpha x_0$ and $\frac{b_0}{r}$, leading to two distinct cases. In the first case, when $\alpha x_0 > \frac{b_0}{r}$, the set of  admissible strategies when initial wealth equals  $\underline{x}$ contains infinitely many elements, so the boundary is non-degenerate. Using the shorthand $U_1:=U_{c,1}$ introduced in Section~\ref{sec:dpp:hjb},
the candidate value function $v$ is then expected to satisfy $U_{1}'(r\underline{x})<v'(\underline{x})<\infty$; see Theorem \ref{thm:verification}.  This implies that the optimal benefit rate at the boundary $\underline{x}$ must be strictly below the interest income $r\underline{x}$. This condition ensures that $\underline{x}$ is not an absorbing boundary under the optimal strategy.
In contrast, in the second case, when $\alpha x_0 \le \frac{b_0}{r}$, the  admissible strategy when the initial wealth equals  $\underline{x}$ is unique, and the value function $V(\underline{x})$ can be computed explicitly. The key requirement is that the candidate value function satisfies $v(\underline{x}) = V(\underline{x})$.

To analyze the HJB equation  \eqref{eq:hjb:simplify}, we employ a transform method introduced in \cite*{bensoussan2016unemployment} and further developed in \cite*{bensoussan2022inter} and \cite*{bensoussan2022risk}. This technique converts the fully nonlinear HJB equation \eqref{eq:hjb:simplify} into the auxiliary equation \eqref{eq:H:first}. Solving \eqref{eq:H:first} poses fundamentally different difficulties in the two boundary cases: 
(i) when $\alpha x_0 > \frac{b_0}{r}$, the verification theorem requires the candidate value function to satisfy $U_{1}'(r\underline{x}) < v'(\underline{x}) < \infty$, with $v'(\underline{x})$ undetermined in advance, and the main challenge is to determine this value. By treating $v'(\underline{x})$ as a free parameter, we establish the global unique existence  of the solution to the semilinear equation \eqref{eq:extension:u_22} (a slight modification of \eqref{eq:H:first}) for each feasible parameter value, together with monotonicity and growth estimates of the solution. 
By analyzing the asymptotic behavior of the resulting family of global solutions, we select a suitable parameter value for which the corresponding solution is strictly increasing and satisfies the precise growth and derivative bounds; see Theorem~\ref{thm:H}. Inverting the transform then yields a solution to the original HJB equation \eqref{eq:hjb:simplify} satisfying the hypotheses of the verification theorem. (ii)
when $\alpha x_0 \le \frac{b_0}{r}$, equation  \eqref{eq:H:first} is posed on $(0,\infty)$ and the point $0$ becomes a degenerate boundary. 
To overcome this difficulty, we apply the Schauder fixed point theorem to establish the existence of a solution to equation  \eqref{eq:H:first}.
We then determine the precise asymptotic behavior of the derivative at the degenerate boundary and establish monotonicity and growth properties analogous to those in Case (i); see Theorem~\ref{thm:Hnew:case2}. We then invert the transform and apply the verification theorem to identify the resulting solution as the value function.

\subsubsection{Numerical Experiments}
Our numerical experiments reveal several eye-catching observations regarding the impact of minimum guarantee constraints on optimal strategies. First, both the optimal investment and benefit policies under minimum guarantees are more conservative than those under the unconstrained framework, particularly when the running wealth level is low, reflecting the prudent influence of the guarantee constraints. Moreover, tighter guarantees, either through a higher wealth floor $\alpha x_0$ or a larger minimum benefit requirement $b_0$, lead to lower investment proportions and, away from the lower wealth bound, lower benefit-to-wealth ratios.  On the other hand, when the running wealth level increases, the influence of these constraints gradually diminishes, and all constrained strategies approach those under the unconstrained framework. This suggests that guarantee mechanisms primarily affect decisions in low-wealth regions, while becoming negligible when the fund is sufficiently well-capitalized.

From a welfare perspective, numerical experiments show that the minimum guarantee constraints entail only a small well-being loss relative to that under the unconstrained framework. This loss is borne primarily by the retiree, while the manager’s utility is only marginally affected. At the same time, the guarantees provide effective downside protection by preventing both wealth and benefit payments from falling below prescribed levels under worst-performing scenarios.
In particular, our analysis of the 5th and 20th percentiles shows that wealth levels and benefit rates under the constrained policies remain at or above their respective guarantees, whereas their unconstrained counterparts quickly breach the lowest guarantee levels at the 5th percentile and also breach them, though more gradually, at the 20th.
Overall, these findings suggest that guarantee mechanisms can enhance retirement security only at the expense of a limited amount of sacrifice, supporting their use as a practical design feature in retirement decumulation products.

\subsection{Outline}
The remainder of the article is organized as follows. Section \ref{sec:formu} presents the mathematical formulation of the model. Section \ref{sec:dpp:hjb} applies the dynamic programming approach to derive the HJB equation and establishes verification theorems for two cases determined by the relative magnitude of $\alpha x_0$ and $\frac{b_0}{r}$. Section \ref{sec:well:hjb} establishes the existence and properties of the HJB solution in the nondegenerate boundary case, while Section \ref{sec:degenerate} treats the degenerate boundary case.
Section \ref{sec:numerical} presents numerical experiments. 
Section \ref{sec:conclusion} concludes the article.
Some discussions, extensions and proofs are given in the Appendix.

\section{Problem Formulation}\label{sec:formu}
Let \((\Omega, \mathcal{F}, \mathbb{F}, \mathbb{P})\) be a complete filtered probability space satisfying the usual conditions. 
The space supports a standard Brownian motion \(W=\{W(t):t\geq 0\}\), and 
  a random variable $\tau_d$ representing the death time of the retiree.  We assume that
  $W$ and $\tau_d$ are independent.
  The filtration
 $\mathbb{F}:=\{\mathcal{F}_t\}_{t \ge 0}$ is defined as the smallest filtration that contains $W$ and renders $\tau_d$ a stopping time.

As mentioned in the introduction, we consider a retiree who, upon retirement, receives a lump sum from a pension scheme and entrusts it to a professional fund manager, who makes decisions on behalf of the retiree.
The fund manager is responsible for controlling the investment proportion $\pi(\cdot)$ allocated to the risky asset, as well as offering the benefit rate $b(\cdot)$ distributed to the retiree. The evolution of the fund's value $X(\cdot)$ is described by  \eqref{eq:sde:x:intro}.
There,
$r_1>0$ is the return rate of the riskless asset offered by the fund,
$\mu>r_1$ and $\sigma>0$ are the expected return and
the volatility of the risky asset respectively, and
$0<\kappa_1<r_1$ is the running management fee rate;  see \cite*{elton2003incentive}, \cite*{khorana2009mutual}, \cite*{guasoni2015hedge} and \cite*{bensoussan2022inter}.
The market price of risk is $\theta := \frac{\mu - r_1}{\sigma}$.

It is typically the case that $r_1 > r_0$, where $r_0$ denotes the standard risk-free rate generally
available in the market to the retiree. This reflects the fact that the fund may have access to better
investment opportunities available only to companies or can offer a higher credited rate to attract
retirees' funds; see, e.g., \cite*{korajczyk2008pricing}, \cite*{gerhold2014transaction},  and \cite{le2020you}. In particular, owing to its large scale and long investment
horizon, the fund can allocate part of its assets to long-maturity bonds and
other buy-and-hold instruments that are effectively illiquid. These assets
typically command a liquidity premium, which is not directly accessible to
individual retirees who mainly invest in liquid, short- to medium-term bonds
through the market;  see, e.g., \cite*{korajczyk2008pricing}, \cite*{gerhold2014transaction},  and \cite{le2020you}. 
Financial intermediaries facilitate the transformation of illiquid investments into liquid claims for customers;  also see \cite*{booth2020modern}.
As a result, the fund is able to generate a higher
effective return,
and $r_1$ can be interpreted as an incentive rate that makes
the fund more attractive relative to outside market alternatives.

Meanwhile, the fund manager charges a running management fee at rate $\kappa_1$ ($\kappa_1<r_1$), reducing the net return experienced by the retiree; see \cite*{elton2003incentive}, \cite*{khorana2009mutual}, \cite*{guasoni2015hedge} and \cite*{bensoussan2022inter}. The net interest rate for the retiree within the fund is thus $r := r_1 - \kappa_1>0$.  
Despite this fee deduction, the net return typically still exceeds the market benchmark;  in practice, one often observes $r \geq r_0$.
For instance, based on the average rates observed in July 2026, the short risk-free rate in the U.S. market is $r_0=3.7\%$, while the fund may achieve a higher return $r_1=5.3\%$ by investing in long-term illiquid bonds that generate a liquidity premium; see \cite*{us_treasury_daily_bill_rates_2026} and \cite*{fed_nominal_yield_curve_2026}.
With a management fee of $\kappa_1=0.4\%$, as documented in \cite*{ici_expense_ratios}, the resulting net rate is $r := r_1 - \kappa_1 = 4.9\%$, which exceeds $r_0$.
Similar patterns can be observed in other markets, including Australia, the United Kingdom and Germany.  
From a mathematical perspective, however, the analysis in this article only requires $r>0$; the comparison with $r_0$ plays no role in the theoretical results.

 The dynamics of $X(\cdot)$ in \eqref{eq:sde:x:intro} can then be rewritten as:
\begin{align}
\left\{
\begin{aligned}
    &\mathrm{d} X(t)=[r+\sigma \theta\pi(t)]X(t)\mathrm{d}t-b(t)\mathrm{d}t+\sigma\pi(t)X(t)\mathrm{d}W(t),\\
    &X(0)=x_0.
    \end{aligned}
    \right.\label{eq:wealth:statics}
\end{align}
 The manager's objective is given by \eqref{eq:problem:max:intro} under the wealth dynamics \eqref{eq:wealth:statics} and subject to the constraints in \eqref{eq:constraint:intro}.
The parameter $\rho> 0$ is the discount rate, used to evaluate the present value of future benefits and management fees. It is taken to be arbitrary and may, in particular, coincide with either $r_0$ or $r$, or not.

In the objective function \eqref{eq:problem:max:intro}, the manager seeks to maximize the expected discounted utility from two distinct sources:
\begin{itemize}
    \item \textbf{Utility from benefit payments and bequest}: \(\int_0^{\tau_d} e^{-\rho t} U_{c,1}(b(t))\,\mathrm{d}t\) represents the utility derived from the continuous stream of benefits \(b(t)\) paid to the retiree over her/his lifetime; \(e^{-\rho \tau_d} U_{c,2}((1-\kappa_2)X(\tau_d))\) accounts for the utility from the net terminal wealth remaining at $\tau_d$ which, after deduction of the terminal management fee, constitutes the bequest to the descendant. As mentioned in the introduction, the manager is fully aware of the retiree’s satisfaction and preferences; in particular, the manager knows the utility functions $U_{c,1}$ and $U_{c,2}$ specified below.
    \item \textbf{Utility from management fees}:  $
\int_0^{\tau_d} e^{-\rho t} U_{m,1}(\kappa_1 X(t))\,\mathrm{d}t
+ e^{-\rho \tau_d}U_{m,2}(\kappa_2X(\tau_d))$
captures the utility obtained from management fees. The first term corresponds to running fees that are proportional to the fund wealth \(X(t)\) with rate \(\kappa_1\), while the second term represents a terminal fee charged at the death time \(\tau_d\) of the retiree, proportional to the terminal wealth with rate \(\kappa_2\in[0,1)\).

\end{itemize}

The constraints in \eqref{eq:constraint:intro} impose minimum guarantees on the fund's operation:
\begin{itemize}
    \item The condition 
    \begin{align*}
    b(t) \geq b_0,
    \end{align*}
    with \(b_0 \in (0, r x_0]\), ensures that beneficiaries receive at least a guaranteed minimum benefit rate at all times. We  require $b_0 \leq rx_0$, since
maintaining a constant minimum benefit level must be financially feasible
given the initial wealth $x_0$ and the risk-free return $r$. If \(b_0 > r x_0\), there exists no admissible strategy that can satisfy \eqref{eq:constraint:intro}.  This feasibility restriction is
consistent with the results in Lemma \ref{lma:b_0/r} and Lemma \ref{lma:stop}.
    \item The condition 
    \begin{align*}
    X(t) \geq \alpha x_0,
    \end{align*}
    where \(\alpha \in (0,1]\), sets a lower bound on the fund's wealth. Here, $\alpha$ is a fixed coefficient, and the product $\alpha x_0$
 represents a prescribed proportion of the initial wealth $x_0$.
    This constraint has the following two implications:
        \begin{enumerate}
            \item [(a)] It ensures that the fund wealth never falls below a fixed proportion of the initial capital, thereby guaranteeing that the net terminal wealth, after deduction of the terminal management fee, remains at least $(1-\kappa_2)\alpha x_0$ and thus provides a minimum bequest level for beneficiaries. In addition, it provides a minimum guarantee for the fund manager, since the running management fee satisfies \(\kappa_1 X(t) \geq \kappa_1 \alpha x_0\) and the terminal fee satisfies \(\kappa_2 X(\tau_d) \geq \kappa_2 \alpha x_0\).
            \item [(b)] More broadly, it serves as a solvency requirement aligned with regulatory or supervisory standards, preventing the manager from adopting overly risky strategies that could jeopardize the fund's ability to meet future obligations and thus averting bankruptcy; see, e.g., \cite*{di2011pension}, \cite{broeders2013pension} and \cite*{avanzi2024optimal}.
        \end{enumerate}
\end{itemize}
It is worth emphasizing that these constraints, while introducing minimum guarantees, do not lead to a significant loss in well-being measured by value function for either
the retiree or the fund manager. In fact, as demonstrated in Section \ref{sec:numerical}, the relative
utility cost of the above guarantees is at most $2.22\%$ for either the retiree or the fund manager
across the three baseline scenarios; see Table~\ref{tab:utility-loss-decomposition}.

 In Appendix \ref{sec:extend}, we further show that, under Constant Relative Risk Aversion (CRRA)  utility, both constraints can be extended to deterministic exponential growth processes. Specifically, we consider a setting in which the guaranteed benefit and the wealth floor grow at a common rate $0< \nu<r$, thereby leading to time-dependent constraints.

In this article, we assume that there exists a constant $\lambda>0$ such that
\begin{align*}
    \mathbb{P}(\tau_d\geq t)=\mathrm{e}^{-\lambda t}, \quad t\in[0,\infty).
\end{align*}
Under this assumption, the objective \eqref{eq:problem:max:intro} can be rewritten in an equivalent form that incorporates the constant mortality rate \(\lambda\); see \eqref{eq:valuefunction:def} below. 
The constant-force specification is a standard tractable benchmark in continuous-time retirement and financial models.  It appears in lifetime consumption, investment, and retirement problems \cite*{dybvig2010lifetime}, in consumption models with independent exponential termination \cite*{davey2021duality}, and in insurance valuation through an equivalent adjustment to the discount rate \cite*{chiarolla2022analytical}. Related insurance-control studies also use exponential random horizons or analyze the constant-hazard case; see \cite*{cao2025continuous,cao2025optimalreporting}. This benchmark keeps the control problem time-homogeneous and allows us to isolate the model's main mechanisms; even in this setting, the analysis remains technically challenging.

In future research, more realistic mortality models could be incorporated. For instance, one could adopt the Makeham mortality law; see, e.g., \cite*{turnbull2017history}. Under such a general mortality specification, the mortality rate \(\lambda\) becomes time-dependent, and the corresponding HJB equation would transform into a partial differential equation (PDE). Solving such problems poses significant challenges and remains an interesting direction for further investigation.

The following lemma establishes a sharper lower bound for the wealth process. 
This lower bound allows us to reformulate the problem as a standard stochastic control problem on a fixed state domain, 
which is essential for applying the dynamic programming approach in the subsequent analysis.
\begin{lemma}\label{lma:b_0/r}
   For any $(b,\pi)$ such that \eqref{eq:wealth:statics} and \eqref{eq:constraint:intro} hold,  we have
    \begin{align*}
        X(s)\geq \frac{b_0}{r}\quad \text{for all } s\geq0, \text{ a.s.}.
    \end{align*}
\end{lemma}
The proof of Lemma \ref{lma:b_0/r} is put in Appendix \ref{appen:lma:b_0/r}. 
The intuition is that $\frac{b_0}{r}$ represents the minimum capital required to sustain the guaranteed benefit level $b_0$ using risk-free returns. If $X(s)<\frac{b_0}{r}$ at some time $s$, then even allocating the entire wealth to the risk-free asset would generate interest income $rX(s)<b_0$, making the minimum benefit payment unsustainable. Consequently, the wealth process would necessarily decline and eventually violate the solvency constraint $X(t)\ge \alpha x_0$.

\section{Dynamic Programming Method and Verification Theorem}\label{sec:dpp:hjb}
Fix $x_0>0$. For simplicity of notation,  we define  $\underline{x}:=\max\left\{\alpha x_0,\frac{b_0}{r}\right\}$, $\delta:=\rho+\lambda$, $U_1(\cdot):=U_{c,1}(\cdot)$ and $$U_2(x):=U_{m,1}(\kappa_1 x)+\lambda U_{m,2}(\kappa_2x)+\lambda U_{c,2}((1-\kappa_2)x), \quad x\in[\underline{x},\infty).$$ We employ the dynamic programming approach in \cite*{fleming2006controlled} and \cite*{pham2009continuous} to solve problem \eqref{eq:problem:max:intro}  under the wealth dynamics \eqref{eq:wealth:statics} and  subject to the constraints  in \eqref{eq:constraint:intro}.
Specifically, we consider the following control problem
\begin{align}\label{eq:valuefunction:def}
    V(x):=\sup_{(b,\pi)\in\mathcal{A}(x)}\mathbb{E}\bigg[\int_0^{\infty}\mathrm{e}^{-\delta s}U_1(b(s))\md s+\int_0^{\infty}\mathrm{e}^{-\delta s}U_2(X^{x;u}(s))\md s\bigg],
\end{align}
where $V$ is defined on $\left[\underline{x},\infty\right)$ and the state process satisfies
\begin{align}\label{eq:sde:x}
\begin{cases}
      &\mathrm{d} X^{x;u}(s)=rX^{x;u}(s)\md s+\sigma \theta\pi(s)X^{x;u}(s)\mathrm{d}s-b(s)\mathrm{d}s+\sigma\pi(s)X^{x;u}(s)\mathrm{d}W(s),\\
       &X^{x;u}(0)=x.
       \end{cases}
\end{align}
In light of Lemma \ref{lma:b_0/r}, we define the admissible   set as follows:
\begin{align}
    \mathcal{A}(x):&=\bigg\{u=(b,\pi):\int_0^T(|b(s)|+|\pi(s)|^2)\md s<\infty\quad \text{for all } T\geq0, \text{ a.s.}; \nonumber\\ &\quad\quad \quad\quad\quad\quad\quad \quad b(s)\geq b_0\,\, \text{and}\,\, {X^{x;u}(s)}\geq \underline{x} \quad \text{for a.e. } s\geq0, \text{ a.s.} \bigg\}.\label{eq:mathcalA:def}
\end{align}

Before applying dynamic programming, we first show that when $\alpha x_0\leq \frac{b_0}{r}$, the boundary point $\underline{x}$ is absorbing under any admissible strategy.
\begin{lemma}\label{lma:stop}
   Suppose that $\alpha x_0\leq \frac{b_0}{r}$. Let $x\geq \underline{x}$ and $u\in \mathcal{A}(x)$. If for some stopping time $\tau$ it holds that $\mathbb{P}(\{X^{x;u}(\tau)=\underline{x}\}\cap\{\tau<\infty\})>0$, then on the event $\{X^{x;u}(\tau)=\underline{x}\}\cap\{\tau<\infty\}$, we have  $X^{x;u}(\tau+s)=\underline{x}$ and  $b(\tau+s)=b_0$ \text{for a.e.} $s\geq0$, \text{a.s.}.
\end{lemma}
The proof of Lemma \ref{lma:stop} is put in Appendix \ref{appen:lma:stop}.
The intuition is as follows. When the fund wealth reaches the lower bound $\underline{x}=\frac{b_0}{r}$, the interest income generated by investing the entire wealth in the risk-free asset is exactly sufficient to finance the minimum benefit payment $b_0$. Any attempt to take investment risk or increase the benefit payment would immediately violate the wealth constraint. Therefore, the only admissible strategy is to invest the entire wealth in the risk-free asset and continuously pay the minimum benefit $b_0$, which keeps the wealth process at the boundary level $\underline{x}$ thereafter.

Based on Lemma~\ref{lma:stop}, we make the following remark concerning the structure of $\mathcal{A}(\underline{x})$ and the expected behavior of $V'(\underline{x})$ in the two parameter cases.
\begin{remark}\phantomsection\label{rmk:cases}
\begin{enumerate}
    \item[(i)]
 If $\alpha x_0 > \frac{b_0}{r}$, then $\underline{x} = \alpha x_0$.
Consider the constant strategy $u \equiv (b_0,0)$. The corresponding state process satisfies
    \begin{align*}
\mathrm{d} X^{\underline{x};u}(s)=(rX^{\underline{x};u}(s)-b_0)\md s,\quad X^{\underline{x};u}(0)=\underline{x},
    \end{align*}
    and thus 
    \begin{align*}
        X^{\underline{x};u}(s)=\left(\underline{x}-\frac{b_0}{r}\right)\mathrm{e}^{rs}+\frac{b_0}{r}\geq \underline{x}\quad \text{for all } s\geq 0.
    \end{align*}
    Therefore, $u\in\mathcal{A}(\underline{x})$. Moreover, it is clear that  $(r\underline{x},0)\in \mathcal{A}(\underline{x})$. In this case, we expect that $V'(\underline{x})\in(0,\infty)$; similar representative examples can be found in \cite*{di2011pension}; our model additionally treats the benefit rate as a control variable subject to a minimum guarantee in \eqref{eq:mathcalA:def}.
    \item[(ii)] If $\alpha x_0\leq \frac{b_0}{r}$, then $\underline{x}= \frac{b_0}{r}$. By Lemma~\ref{lma:stop}, the admissible set reduces to $\mathcal{A}\left(\underline{x}\right)=\left\{(b_0,0)\right\}$. Thus,

    \begin{align}\label{eq:V:underline:x}
    V\left(\underline{x}\right)=\frac{1}{\delta}U_1(b_0)+\frac{1}{\delta}U_2\left(\underline{x}\right).
    \end{align}
    In this case, we expect that $V'\left(\underline{x}\right)=\infty$; similar representative and instructive examples can be found in \cite{merton1992continuous}, \cite{herdegen2021elementary}, which do not consider state or control constraints.
\end{enumerate} 
\end{remark}

Using the dynamic programming principle, we first formally derive the HJB equation \eqref{eq:hjb} satisfied by the value function. We then establish verification theorems in Section \ref{sec:verify}, distinguishing two cases according to whether $\alpha x_0>b_0/r$ or $\alpha x_0\leq b_0/r$. These theorems show that any sufficiently regular solution of the HJB equation satisfying appropriate growth and boundary conditions coincides with the value function, and that the associated optimal strategy can be constructed accordingly. Furthermore, the existence of such solutions will be investigated separately in Sections \ref{sec:well:hjb} and \ref{sec:degenerate} for the nondegenerate and degenerate boundary cases, respectively. To this end, we impose the following growth and regularity assumptions on the utility functions.

\begin{assumption}\label{assump:u:1}
     \item[(i)] Assume that $U_1\in C^2([b_0,\infty))$. Moreover, $U_1'(b)>0$ and $U_1''(b)<0$ for all $b\in[b_0,\infty)$, and $U_1'(\infty)=0$. There exist  constants $C_{1,0}>0$  and $\gamma_1\in(0,1)$ such that \footnote{Hereafter, $C_{1,i}$, $i=0,1,2,\dots$, denote constants associated with the utility function $U_1$, in a similar way as $C_{2,i}$ correspond to $U_2$. We will also use $C_{H,i}$, $C_{v,i}$, etc., to denote constants related to functions $H$, $v$, and so on.}
     \begin{align*}
         & 0<U_1(b)\leq C_{1,0} b^{\gamma_1}\quad \text{for all } b\in [b_0,\infty).
     \end{align*}
    \item[(ii)] Assume that  $U_2\in C^1([\underline{x},\infty))$. Moreover, there exist  constants $C_{2,0}>0$ and $\gamma_2\in(0,1)$  such that
     \begin{align*}
        0<U_2(x)\leq C_{2,0} x^{\gamma_2}\quad \text{for all } x\in [\underline{x},\infty).
     \end{align*}
\end{assumption}
We emphasize that Assumption \ref{assump:u:1}, together with the parameter condition specified in Assumption \ref{assump:parameter}, suffices for the formal derivation of the HJB equation and for the verification argument. However, the well-posedness of the HJB equation, in particular the existence and regularity of classical solutions, requires stronger structural assumptions on the utility functions. These additional assumptions will be introduced later in Assumptions \ref{assmup:U12}, \ref{assump:u_2'}  and \ref{assump:holder:u2}.
\begin{remark}\phantomsection\label{rmk:assump:1}
\begin{enumerate}
    \item[(i)]
The positivity of $U_1$ and $U_2$ is assumed without loss of generality. Indeed, adding or subtracting a constant from the utility functions does not affect the optimal strategies; it only shifts the value function by an additive constant.

\item[(ii)] Assumption \ref{assump:u:1} is satisfied by a broad class of utility functions commonly used in the literature, including standard CRRA utility functions. In particular, $U_1$ and $U_2$ may correspond to different  parameters $\gamma$ related to risk aversion, and each of them may also be represented as a linear combination of CRRA utilities with different $\gamma$. We do not impose any specific functional form; rather, our framework accommodates any utility function that satisfies Assumption \ref{assump:u:1}.

\item[(iii)] Define $\gamma := \max\{\gamma_1,\gamma_2\}$. Since $b_0>0$ and $\underline{x}>0$, the growth conditions in Assumption~\ref{assump:u:1}(i) and (ii) for $U_1$ and $U_2$ can equivalently be written using the single exponent $\gamma$ by suitably adjusting the constants $C_{1,0}$ and $C_{2,0}$.
\end{enumerate}
\end{remark}

We also  impose the following standard assumption on the model parameters.
\begin{assumption}\label{assump:parameter}
   We shall assume throughout that
    \begin{align*}
        \delta>r \gamma +\frac{\gamma}{2(1-\gamma)}\theta^2.
    \end{align*}
\end{assumption}
\begin{remark}\label{rmk:unconstraint}
When $U_1$ and $U_2$ are both of CRRA type with the same exponent $\gamma\in(0,1)$, namely $U_1(b)=b^\gamma/\gamma$ and $U_2(x)=\eta x^\gamma/\gamma$, and in the absence of constraints (i.e., $b_0=0$ and $\alpha=0$), Assumption~\ref{assump:parameter} is necessary and sufficient for the finiteness of the value function; see, for example, \cite*{merton1992continuous} and \cite*{herdegen2021elementary}.

In this case, the value function admits the explicit form $V_{un}(x)=\bar{b}\frac{x^{\gamma}}{\gamma}$, where $\bar{b}\in(0,\infty)$ is the unique positive solution to the algebraic equation in $b>0$
\begin{align*}
    \eta+(1-\gamma) b^{\frac{\gamma}{\gamma-1}}-b\left(\delta-r \gamma-\frac{\gamma}{2 (1-\gamma)}\theta^2\right)=0.
\end{align*}
Moreover, the optimal benefit rate and investment proportion are given by
\begin{align}\label{eq:unconstraint:bpi}
b_0^*(x)=\bar{b}^{\frac{1}{\gamma-1}} x, \quad
\pi_0^*(x)=\frac{\theta}{(1-\gamma) \sigma}.
\end{align}
 To demonstrate that Assumption~\ref{assump:parameter} holds for a wide range of economically reasonable parameter values, consider the parameter set $\gamma=0.3$, $\theta=0.135$, $r=0.049$ and $\delta=0.082$ used in Section~\ref{sec:numerical}. A straightforward calculation gives
\begin{align*}
\delta-r\gamma-\frac{\gamma}{2(1-\gamma)}\theta^2=0.0634  >0.
\end{align*}
Therefore, Assumption~\ref{assump:parameter} is satisfied under these parameter choices.

\end{remark}

Formally applying the dynamic programming principle, the HJB equation for the value function reads, see, e.g., \cite*{soner1986optimal} and \cite*{ishii2002class}, 
\begin{align}\label{eq:hjb}
 	-\delta v(x)&+\sup _{b\geq b_0} \{U_1(b)-bv'(x)\}+\sup_{\pi\in\mR}\{\sigma \theta\pi x v'(x)+\frac{\sigma^2}{2}\pi^2x^2v''(x)\}\nonumber\\&\quad\quad\quad\quad\quad\quad\quad\quad\quad\quad\quad\quad\quad\quad\quad+ U_2(x)+rxv'(x)=0,\quad x>\underline{x}.
 \end{align}
Define $I:(0,\infty)\to [b_0,\infty)$ 
by
\begin{align}\label{eq:I}
    I(p):=\left\{\begin{aligned}
       &(U_1')^{-1}(p), && p\in(0,U_1'(b_0)];\\
        &b_0,&&p\in(U_1'(b_0),\infty),
    \end{aligned}\right.
\end{align}
and $D:(0,\infty)\to \mR$ by
\begin{align}\label{eq:D}
    D(p):=\sup _{b\geq b_0} \{U_1(b)-bp\}=U_1(I(p))-pI(p).
\end{align}
Then, $D\in C^1((0,\infty))$ and $D'=-I$.
If, in addition, $v'>0$ and $v''<0$ on $(\underline{x},\infty)$, then \eqref{eq:hjb} is equivalent to
 \begin{align}
 	\delta v(x)=D(v'(x))+ U_2(x)
 	-\frac{1}{2}\theta^2\frac{\left(v'(x)\right)^2}{v''(x)}+rx v'(x), \quad  x>\underline{x}.\label{eq:hjb:simplify}
 \end{align}

\subsection{Verification Theorem}\label{sec:verify}
Before analyzing \eqref{eq:hjb:simplify} in detail, we establish a verification theorem.
Since the behavior of $V'(\underline{x})$ may be either finite or infinite (see Remark \ref{rmk:cases}), we consider two cases, namely $\alpha x_0> \frac{b_0}{r}$ and $\alpha x_0\le \frac{b_0}{r}$.
These cases are addressed separately in Theorems~\ref{thm:verification} and~\ref{thm:verification2}, as they give rise to different boundary conditions and require different verification arguments.

Our verification approach is inspired by \cite*{di2011pension}, 
but substantial modifications are needed because, in our model, 
the benefit rate is part of the control rather than being fixed exogenously. 
This feature leads to a different boundary structure and requires additional arguments in the verification step. 
For completeness and to keep the main exposition focused on the analysis of the HJB equation, 
the detailed proofs of the two verification theorems are deferred to Appendices \ref{appendix:sec:veri1} and \ref{appendix:sec:veri2}, respectively.

\begin{theorem}\label{thm:verification} (Verification theorem (Case I: $\alpha x_0>\frac{b_0}{r}$))
Under Assumptions \ref{assump:u:1} and \ref{assump:parameter},
    suppose that $v$ satisfies the following conditions:
    \begin{enumerate}
         \item[(i)] $v\in C^1([\underline{x},\infty))\cap C^2((\underline{x},\infty))$, $v'>0$ on $[\underline{x},\infty)$ and $v''<0$ in $(\underline{x},\infty)$. Moreover, $v$ satisfies the HJB equation \eqref{eq:hjb} (equivalently, \eqref{eq:hjb:simplify}).
 \item[(ii)] $v'(\underline{x})>U_1'(r\underline{x})$, $\lim\limits_{x\downarrow \underline{x}}v''(x)=-\infty$, and there exists some $C_{v,0}>0$ such that $0\leq v(x)\leq C_{v,0}x^{\gamma}$ for any $x\geq \underline{x}$.
    \item[(iii)] Define
\begin{align}\label{eq:b^*}
b^*(x)
:=
\begin{cases}
I\bigl(v'(x)\bigr),   & x \ge \underline{x},\\
I\bigl(v'(\underline{x})\bigr), & x \le \underline{x},
\end{cases}
\quad
\pi^*(x)
:=
\begin{cases}
-\dfrac{\theta v'(x)}{\sigma x\, v''(x)}, & x > \underline{x},\\
0,                                              & x \le \underline{x}.
\end{cases}
\end{align}
For any $x\geq \underline{x}$, suppose the following SDE admits a unique strong solution satisfying $X^{x;*}\geq \underline{x}$:
\begin{align}\left\{\begin{aligned}\label{eq:verify:wealth}
	&\md X^{x;*}(s)=rX^{x;*}(s)\md s+\sigma \theta\pi^*(X^{x;*}(s))X^{x;*}(s)\md s-b^*(X^{x;*}(s))\md s\\&\quad\quad\quad\quad\quad\quad\quad\quad\quad\quad+\sigma\pi^*(X^{x;*}(s))X^{x;*}(s)\md W(s),\\
	&X^{x;*}(0)=x.
	\end{aligned}\right.
\end{align}
Suppose further that  $I(v')$ and $\frac{v'}{v''}$ are of linear growth.
    \end{enumerate}
Then, $v=V$ and $(b^*(X^{x;*}),\pi^*(X^{x;*}))$ constitutes the optimal strategy for an initial wealth of $x$.
\end{theorem}
\begin{remark}

Theorem \ref{thm:verification} relies essentially on three sets of conditions, among which the most distinctive from classical results is the boundary behavior at $\underline{x}$, in particular the requirement that $
v'(\underline{x}) > U_1'(r\underline{x})$. We emphasize that all these conditions, such as $
v'(\underline{x}) > U_1'(r\underline{x})$ and the linear growth conditions on $I(v')$ and $\frac{v'}{v''}$, will be proved later in Theorem \ref{thm:v}.

First, the candidate function $v$ must possess sufficient regularity, namely 
$v \in C^1([\underline{x},\infty)) \cap C^2((\underline{x},\infty))$, and it must satisfy the HJB equation \eqref{eq:hjb}. 
These smoothness requirements ensure that Itô's formula can be applied and that the HJB equation holds in the classical sense.

Second, appropriate boundary behavior at $\underline{x}$ is required. 
The condition $
\lim\limits_{x \downarrow \underline{x}} v''(x) = -\infty $
is standard in optimal investment and consumption problems  and characterizes the singular behavior of the value function at the boundary $\underline{x}$; see \cite{pham2009continuous} and \cite{touzi2012optimal} for example. It ensures that the optimal investment in the risky asset vanishes as the wealth approaches the boundary, i.e., $\lim\limits_{x \downarrow \underline{x}}\pi^*(x)=0$.
Moreover, the inequality $
v'(\underline{x}) > U_1'(r\underline{x})$
implies that 
the optimal benefit rate at the boundary satisfies 
$b^*(\underline{x})=I(v'(\underline{x})) < r\underline{x}$ 
(see Lemma \ref{lma:limit:v''}), 
that is, at $\underline{x}$ the optimal payout must be strictly below the interest income $r\underline{x}$. This condition ensures that $\underline{x}$ is not an absorbing boundary corresponding to the optimal strategy. It further yields the asymptotic behavior $v''(x)\asymp -(x-\underline{x})^{-\frac12}$ as $x\downarrow \underline{x}$ (see Lemma \ref{lma:limit:v''}). This rate is independent of $\gamma$ and is crucial for the proof of the verification theorem; see \eqref{eq:conver:cross} and \eqref{eq:for:something}. 
In addition, $v$ must satisfy a suitable growth condition, which is consistent with the properties of $U_1$ and $U_2$ imposed in Assumption~\ref{assump:u:1}.

Third, the state equation associated with the candidate optimal strategy must admit a unique strong solution, 
which ensures that the candidate optimal strategy is admissible and that the verification procedure is well-posed. 
Moreover, the linear growth conditions on $I(v')$ and $\frac{v'}{v''}$ guarantee the required integrability estimates such as \eqref{eq:moment:} in this procedure.
\end{remark}

\begin{theorem}\label{thm:verification2} (Verification theorem (Case II: $\alpha x_0\leq\frac{b_0}{r}$))
Under Assumptions \ref{assump:u:1} and \ref{assump:parameter}, 
suppose that $v$ satisfies the following conditions:
    \begin{enumerate}
         \item[(i)] $v\in C([\underline{x},\infty))\cap C^2((\underline{x},\infty))$,  with $v'>0$ and $v''<0$ in $(\underline{x},\infty)$. Moreover,  $v$ satisfies the  HJB equation \eqref{eq:hjb} (equivalently, \eqref{eq:hjb:simplify}).
 \item[(ii)] $v(\underline{x})= V(\underline{x})$,  and there exists some $C_{v,1}>0$ such that $0\leq v(x)\leq C_{v,1}x^{\gamma}$ for any $x\geq \underline{x}$.
    \item[(iii)] Let $(b^*,\pi^*)$ be defined as in \eqref{eq:b^*}.  
For any $x\geq \underline{x}$, the SDE \eqref{eq:verify:wealth} admits a unique strong solution satisfying $X^{x;*}\geq \underline{x}$.

Suppose further that  $I(v')$ and $\frac{v'}{v''}$ are of linear growth.
    \end{enumerate}
Then, $v=V$ and $(b^*(X^{x;*}),\pi^*(X^{x;*}))$ constitutes the optimal strategy for an initial wealth of $x$.
\end{theorem}
\begin{remark}
Theorem \ref{thm:verification2} also relies on three sets of conditions, to be proved later in Theorem \ref{thm:vnew:case2}, 
that  are largely parallel to those in Theorem \ref{thm:verification}. We emphasize that all these conditions, such as the linear growth conditions on $I(v')$ and $\frac{v'}{v''}$, will be proved later in Theorem \ref{thm:vnew:case2}.

Regarding the first and third sets of conditions, the regularity of $v$, the validity of the HJB equation, 
the existence of a unique strong solution to the state equation under the candidate optimal strategy $(b^*,\pi^*)$, 
as well as the linear growth conditions on $I(v')$ and $\frac{v'}{v''}$ used to get \eqref{eq:moment:}, 
play exactly the same technical role as in Theorem \ref{thm:verification}. 
They ensure that Itô's formula can be applied, that the candidate optimal strategy is admissible, 
and that the required integrability estimates such as \eqref{eq:moment:} hold in the verification argument.

Second, the boundary condition is different. 
In the present case, recall from Remark \ref{rmk:cases} that the admissible
 set $\mathcal{A} (\underline{x})$ reduces to a singleton, so that the value $V
(\underline{x})$ is explicitly given by  \eqref{eq:V:underline:x}.
Therefore, instead of imposing a structural boundary condition such as 
$\lim\limits_{x\downarrow \underline{x}} v''(x)=-\infty$, 
we require directly that 
$v(\underline{x})=V(\underline{x})$,
which guarantees consistency between the candidate function and the original control problem at the boundary. For completeness, we note that $\lim\limits_{x\downarrow \underline{x}} v''(x)=-\infty$
does hold for the candidate solution as shown in Theorem \ref{thm:vnew:case2}; however, this property is not needed for the verification argument.
\end{remark}
 
 \section{Existence and Properties of the Solution to the HJB Equation}
 \label{sec:well:hjb}
The HJB equation \eqref{eq:hjb:simplify} is highly nonlinear due to the presence of the terms 
 $\frac{(v')^2}{v''}$ and $D(v')$.
 This fully nonlinear structure makes a direct analysis difficult. 
 To overcome this difficulty, we employ a transform method introduced in 
\cite*{bensoussan2016unemployment} and further developed in 
\cite*{bensoussan2022inter} and \cite*{bensoussan2022risk}. 
This approach allows us to convert the fully nonlinear HJB equation into a semi-linear differential equation.

We first formally derive the  semilinear equation under the assumption that the solution $v$ to \eqref{eq:hjb:simplify} is sufficiently smooth. We then rigorously establish the existence and properties of its solutions. Finally, by inverting the transform, we obtain the existence and properties of solutions to the original HJB equation.

\subsection{Formal Derivation of the Auxiliary Equation}\label{sec:transform}
To derive the semilinear equation, we first differentiate \eqref{eq:hjb:simplify} with respect to $x$, obtaining
\begin{align*}
		\left(\delta-r+ \theta^2\right)v'(x)= 
		  U_2'(x) +\frac{1}{2} \theta^2 \frac{\left(v'(x)\right)^2 v'''(x)}{\left(v''(x)\right)^2}  +\left[r x-I\left(v'(x)\right)\right] v''(x),\quad  x>\underline{x}.
	\end{align*}
Let $A:=\lim\limits_{x\downarrow\underline{x}}v'(x)$ and
    $$f(x):=v'(x)\in(0,A), \quad x\in(\underline{x},\infty),$$
    where (i) $0<A<\infty$ when $\alpha x_0>\frac{b_0}{r}$; or (ii) $A=\infty$ when $\alpha x_0\leq\frac{b_0}{r}$.
Then $f$ satisfies
	\begin{align}\label{eq:f}
    \left\{
    \begin{aligned}
		&\left(\delta-r+ \theta^2\right)f(x)= 
		  U_2'(x) +\frac{1}{2} \theta^2 \frac{\left(f(x)\right)^2 f''(x)}{\left(f'(x)\right)^2}   +\left[r x-I\left(f(x)\right)\right] f'(x),\quad x>\underline{x}, \\
        &\lim\limits_{x\downarrow\underline{x}}f(x)=A,\quad \lim\limits_{x\downarrow \underline{x}}f'(x)=-\infty.
        \end{aligned}
        \right.
	\end{align}
To eliminate the nonlinear quotient term, we introduce the inverse function $F:(0,A)\to (\underline{x},\infty)$ defined by
 $F(f(x))=x$ for all $x\in(\underline{x},\infty)$. A direct computation yields
	\begin{align*}
		f'(x) =\frac{1}{F'(f(x))},\quad  f''(x)=-\frac{F''(f(x))}{(F'(f(x)))^3},\quad x\in(\underline{x},\infty).
	\end{align*}
Moreover, the boundary conditions become: $\lim\limits_{z\uparrow A}F(z)=\lim\limits_{x\downarrow\underline{x}}F(f(x))=\underline{x}$ and $\lim\limits_{z\uparrow A}F'(z)=\lim\limits_{x\downarrow\underline{x}}F'(f(x))=\lim\limits_{x\downarrow\underline{x}}\frac{1}{f'(x)}=0$.
	We then get
	\begin{align}\label{eq:F:noboud}
		\!-\!\frac{1}{2} \theta^2 (f(x))^2\frac{F''(f(x))}{F'(f(x))} \!+\![r F(f(x))\!-\!I( f(x))]\frac{1}{F'(f(x))}\!+\!(r\!-\!\theta^2\!-\!\delta) f(x)\!+\!  U_2'(F(f(x)))\!=\!0.
	\end{align}	
By performing the change of variable \(z = f(x)\) and substituting into \eqref{eq:F:noboud}, we obtain the following ordinary differential equation (ODE)  for \(F\):
	\begin{align}\label{eq:F}
    \left\{
    \begin{aligned}
		&\!-\!\frac{1}{2} \theta^2 z^2F''(z)+\left(r\!-\!\theta^2\!-\!\delta\right) zF'(z)  +  U_2'(F(z))F'(z)+r F(z)-I\left( z\right)=0,~ z\in(0,A),\\
        &\lim\limits_{z\uparrow A}F(z)=\underline{x},\quad \lim\limits_{z\uparrow A}F'(z)=0.
        \end{aligned}
        \right.
	\end{align}
By introducing the change of variable \(h = 1/z\), we reformulate \eqref{eq:F} as an initial-value problem starting from \(h = 1/A\). This forward formulation ensures that the solution is monotone increasing, in accordance with the monotonicity of the value function \(V\) with respect to \(x\). Specifically, we define the auxiliary function $H$ by
	\begin{align*}
		H(h):=F\left(\frac 1 h\right), \quad h\in\left(\frac {1}{A}, \infty\right).
	\end{align*}
	Then $H$ satisfies
	\begin{align*}
		& H'(h)=-\frac{1}{h^2}F'\left(\frac 1 h\right),\quad  H''(h)=2\frac{1}{h^3}F'\left(\frac 1 h\right)+\frac{1}{h^4}F''\left(\frac 1 h\right),\quad h\in\left(\frac {1}{A}, \infty\right).
	\end{align*}
	Consequently, the auxiliary function $H$ solves the auxiliary equation \eqref{eq:H:first}:
	\begin{align}\label{eq:H:first}
    \left\{
    \begin{aligned}
		&\!-\!\frac{1}{2} \theta^2 h^2H''(h) \!-\!(r\!-\!\delta) hH'(h)\!-\!  U_2'(H(h))h^2H'(h) \!+\!r H(h) \!-\!I\left( \frac1h\right)\!=\!0,~ h\in\left(\frac {1}{A}, \infty\right),\\
        &\lim\limits_{h\downarrow \frac{1}{A}}H\left(h\right)=\underline{x}, \quad\lim\limits_{h\downarrow \frac{1}{A}}h^2H'\left(h\right)=0.
        \end{aligned}
        \right.
        \end{align}

\subsection{Assumptions and the Main Result for the Auxiliary Equation}\label{sec:case1}
In the remainder of this section, we focus on the nondegenerate boundary case \(\alpha x_0 > b_0/r\).   The degenerate boundary
case $\alpha x_0\leq\frac{b_0}{r}$ is treated separately in
Section~\ref{sec:degenerate}. We start by examining the ODE \eqref{eq:H:first}. Its solution depends on the initial point \(y = 1/A\), where the boundary slope \(A = v'(\underline{x})\) is not known in advance. For each fixed \(y\), we prove that \eqref{eq:H:first}, after a suitable
extension, admits a unique global solution; see
Proposition~\ref{lma:global}. Hence, varying \(y\) over the admissible range yields a family of
solutions.
Whenever the inverse transform is well-defined, the corresponding member of this family gives rise to a candidate solution of the original HJB equation. The remaining challenge is to determine the appropriate value of $y$, or
equivalently the actual boundary slope $A=v'(\underline{x})=1/y$ of the value
function. This actual boundary slope cannot be determined directly.
Guided by
the verification theorem (Theorem~\ref{thm:verification}), we instead seek a value of $y$ for
which the candidate value function reconstructed from $H(\cdot;y)$ satisfies
the requirements of verification.
Once such a value of $y$ is found, the verification theorem ensures that the
candidate coincides with the value function, thereby also identifying
$y$ as precisely the reciprocal of the actual boundary slope of the value
function.
Theorem~\ref{thm:H} below shows that there exists an initial point $y_0$
such that the associated solution $H(\cdot;y_0)$ has the monotonicity and growth
properties needed for the subsequent reconstruction and verification of the
value function.

        In order to establish existence and growth estimates for the solution to \eqref{eq:H:first}, and building on Assumption \ref{assump:u:1}, we impose the following stronger assumptions on the utility functions $U_1$ and $U_2$; see  \cite*{bensoussan2022inter} for a similar assumption.
\begin{assumption}\label{assmup:U12}
    \item[(i)]  Assume that $U_1\in C^3([b_0,\infty))$. Moreover,
    $U_1(b)>0$ and $U_1''(b)<0$ for all $b\in[b_0,\infty)$, with
    $U_1'(\infty)=0$ and $U_1(\infty)=\infty$.
    There exist constants
    $C_{1,1}>0$ and $\gamma\in(0,1)$ such that
     \begin{align*}
          0<U_1'(b)\leq C_{1,1} b^{\gamma-1}\quad \text{for all } b\in[b_0,\infty).
     \end{align*}
    \item[(ii)] Assume that $U_2\in C^1([\underline{x},\infty))$.
    Moreover, $U_2'(x)>0$ for all $x\in[\underline{x},\infty)$.
    There exists a constant $C_{2,1}>0$ such that
     \begin{align*}
        0<U_2(x)\leq C_{2,1} x^{\gamma}\quad \text{for all }x\in[\underline{x},\infty).
     \end{align*}
\end{assumption}
Assumption \ref{assmup:U12} covers a broad class of utility functions commonly used in the literature (see Remark \ref{rmk:assump:1}(ii)).  Similar to Remark \ref{rmk:assump:1}(iii), the growth conditions in Assumption~\ref{assmup:U12}(i) and (ii) could equivalently use separate exponents $\gamma_1$ and $\gamma_2$ for $U_1$ and $U_2$, with the constants $C_{1,1}$ and $C_{2,1}$ adjusted accordingly. For convenience in the subsequent analysis, we adopt a single exponent $\gamma$ without loss of generality. In addition,
it does not require $U_2$ to be concave.

For the nondegenerate boundary case, we impose the following mild regularity condition on \(U_2'\) so that the local well-posedness of \eqref{eq:H:first} follows from standard ODE theory. The degenerate boundary case, treated in Section~\ref{sec:degenerate}, instead relies on a Hölder continuity condition on $U_2$.
\begin{assumption}\label{assump:u_2'}
    $U_2'$ is locally Lipschitz continuous.
\end{assumption}

The following theorem is the main result of this section, establishing the existence of a suitable initial point $y_0$ and a corresponding global solution to \eqref{eq:H:first} with the required monotonicity and growth properties.

         \begin{theorem}\label{thm:H}
  Under Assumptions \ref{assump:parameter}, \ref{assmup:U12} and \ref{assump:u_2'},  suppose further that $\alpha x_0>\frac{b_0}{r}$. Then there is a  $y_0 \in \left(0,\frac{1}{U_1'(r\underline{x})}\right)$ such that the following ODE
\begin{align}\label{eq:extension:u_222}
\left\{
\begin{aligned}
		&\!-\!\frac{1}{2} \theta^2 h^2H''(h) \!-\!(r\!-\!\delta) hH'(h)\!-\!  U_2'(H(h))h^2H'(h)\! +\!r H(h) \!-\!I\left( \frac1h\right)\!=\!0, ~h\in(y_0, \infty),\\
        &H\left(y_0\right)=\underline{x}, \quad H'\left(y_0\right)=0
        \end{aligned}
        \right.
\end{align}
 admits a  solution $H\in C^2([y_0,\infty))$, and $H$ satisfies 
 \begin{align}
     &(i)\quad\quad\quad\quad\quad\quad{\frac{2}{\theta^2p_1(p_1+p_2)}I\left( \frac1h\right)\leq H(h)\leq C_{H,0}h^{\frac{1}{1-\gamma}}}\quad \text{for all } h\in[y_0,\infty),\label{eq:H:property}\\
     &(ii)\quad\quad\quad\quad\quad\quad0<\frac{hH'(h)}{H(h)}\le p_1\quad \text{for all }h\in(y_0,\infty)\label{eq:H':property}
 \end{align}
 for some constant $C_{H,0}>0$, where $p_1$, $q_1=-p_2$ are the two roots of the following quadratic equation:
\begin{align}\label{eq:p}
    -\frac{1}{2} \theta^2 p(p-1) - (r-\delta) p + r = 0,
\end{align}
such that $p_1>\frac{1}{1-\gamma}$ and $p_2>0$.
\end{theorem}

\begin{figure}[htbp]
    \centering
    \includegraphics[width=0.5\textwidth]{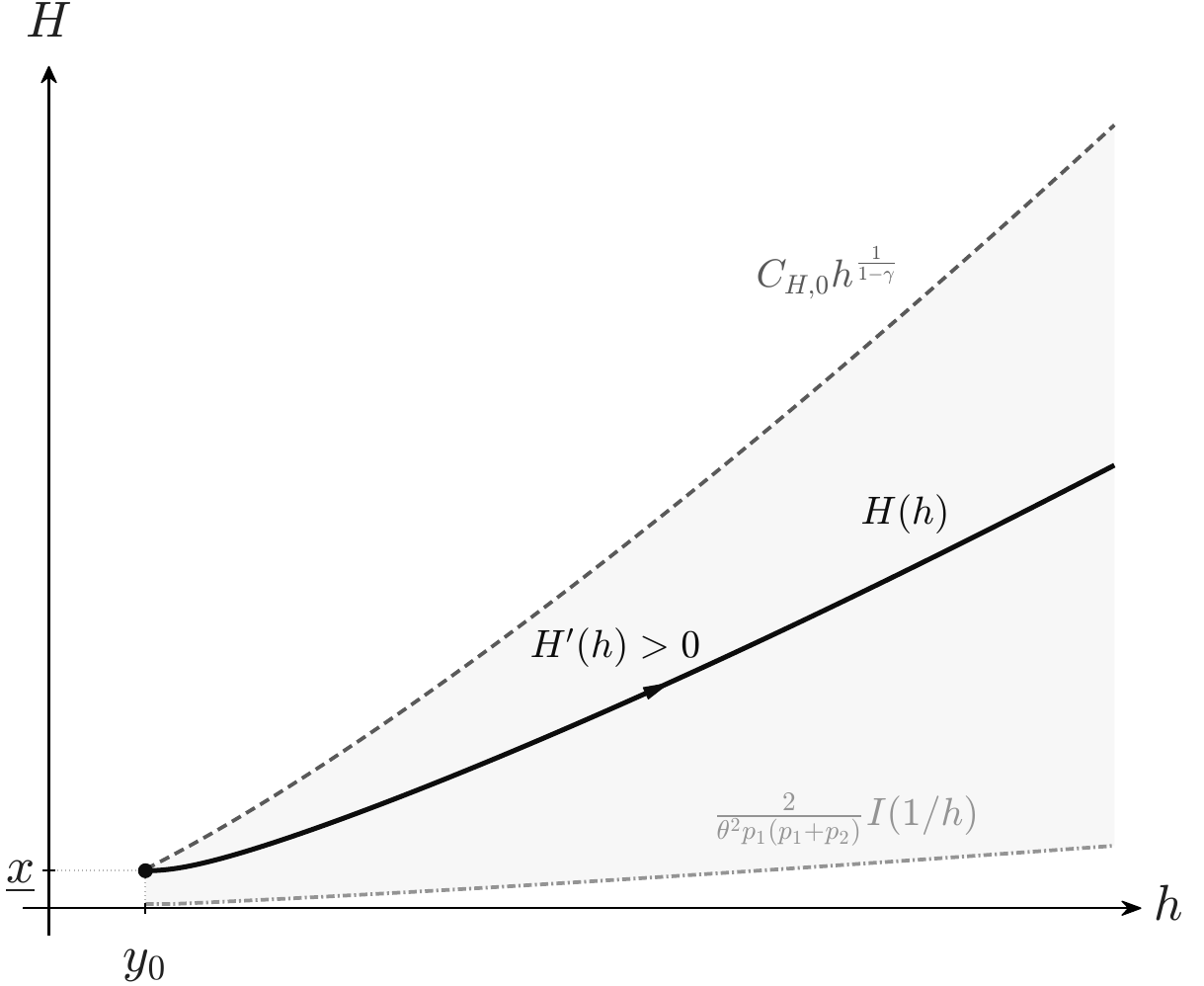}
    \caption{Schematic illustration of the selected
    solution $H$ in Theorem~\ref{thm:H}, together with the lower and
    upper bounds in \eqref{eq:H:property}. The arrow highlights $H'>0$;
    $H$ is not necessarily convex or concave.}
    \label{fig:thm3-growth}
\end{figure}

 Figure~\ref{fig:thm3-growth} schematically depicts the selected solution \(H\), its strict increase, and the two-sided bounds. The strict increase of \(H\) translates into the strict concavity \(v''<0\) of the value function ((i) of Theorem \ref{thm:verification}), reflecting the diminishing marginal utility of wealth. The upper growth bound, in turn, ensures that \(v(x)\) grows at most like \(x^\gamma\) with \(\gamma<1\) ((ii) of Theorem \ref{thm:verification}), so that total welfare grows sublinearly with wealth. The remaining two estimates in \eqref{eq:H:property}--\eqref{eq:H':property} further imply that both the optimal benefit rate and the amount invested in the risky asset grow at most linearly with wealth ((iii) of Theorem \ref{thm:verification}). These bounds prevent excessive payouts and extreme leverage. Together, these properties provide the analytical estimates needed to reconstruct and subsequently verify the value function and its associated optimal policies, while also ensuring that the resulting strategy is economically sensible, namely, diminishing marginal utility of wealth discourages excessive risk-taking, and linearly bounded policies rule out extreme leverage and unsustainable payouts.

 In the following, we always assume that Assumptions \ref{assump:parameter}, \ref{assmup:U12} and \ref{assump:u_2'} hold.

\subsection{ Global Existence for a Fixed Initial Point}
\label{sec:case1:existence}

To prove Theorem~\ref{thm:H}, we first address a technical issue in solving
\eqref{eq:extension:u_222}. Since $U_2$ is originally defined only on
$[\underline{x},\infty)$, the Picard--Lindel\"of local existence theorem
cannot be directly applied: a local solution might fall below the level
$\underline{x}$ within its construction steps, which would make the term
$U_2'(H)$ undefined. We therefore
begin by extending $U_2$ to the entire real line.
Under Assumptions~\ref{assmup:U12} and \ref{assump:u_2'}, there exists an extension $\tilde{U}_2\in C^1(\mathbb{R})$ whose derivative is locally Lipschitz continuous and which satisfies
\begin{align}
&
\tilde{U}_2=U_2 \ \text{on }[\underline{x},\infty),\qquad
\tilde{U}_2\geq0,\quad \tilde{U}_2'\geq0 \ \text{on }\mathbb{R},\qquad
\tilde{U}_2'=0 \ \text{on }(-\infty,0].
\label{eq:tildeU2:properties}
\end{align}
 The subsequent verification
$H\geq\underline{x}$ ensures that the final result is independent of the
particular choice of the extension. For an initial point $y>0$, consider the following
extended initial-value problem. For notational simplicity,
we use the same symbol $H$ for the unknown function of the extended equation \eqref{eq:extension:u_22}
and write
$H(\cdot;y)$ when emphasizing its dependence on $y$. Note that any solution of this extended equation agrees with that of \eqref{eq:H:first} provided \(H\ge \underline{x}\), a property that will be established for the initial point \(y_0\) selected in the sequel.
\begin{align}\label{eq:extension:u_22}
\left\{
\begin{aligned}
		&-\frac{1}{2} \theta^2 h^2H''(h) -\left(r-\delta\right) hH'(h)- \tilde{U}_2'(H(h))h^2H'(h) +r H(h) -\tilde{I}(h)=0,\quad h\in(y,\infty),\\
        &H\left(y\right)=\underline{x}, \quad H'\left(y\right)=0,
        \end{aligned}
        \right.
\end{align}
where the benefit rule expressed in terms of $z$ is defined by
$\tilde{I}(z):=I(1/z)$ for $z>0$, with $I$ defined in \eqref{eq:I}. That is,
\begin{align}\label{eq:tildeI}
    \tilde{I}(z)=
    \left\{
    \begin{aligned}
       &b_0, && z\in(0,z_0),\\
       &(U_1')^{-1}\left(\frac{1}{z}\right),
       && z\in[z_0,\infty),
    \end{aligned}
    \right.
\end{align}
where
\begin{equation*}
z_0:=\frac{1}{U_1'(b_0)}.
\end{equation*}
The threshold $z_0$ marks the point beyond which the guaranteed
minimum-benefit constraint ceases to bind.

We next introduce a second threshold to specify the admissible range
of the initial point, namely
\begin{equation*}
z_1:=\frac{1}{U_1'(r\underline{x})}.
\end{equation*}
Since $r\underline{x}>b_0$ and $U_1'$ is strictly decreasing,
$0<z_0<z_1<\infty$. At the threshold $z_1$, the benefit rule gives
$\tilde I(z_1)=r\underline{x}$, the net risk-free return at the
minimum wealth level. We therefore expect the initial point $y$ to lie in 
$(0,z_1)$, ensuring that the minimum-wealth boundary is not
absorbing. In particular, this is precisely the range implied by the boundary requirement
in (ii) of Theorem~\ref{thm:verification}. Accordingly, we henceforth regard \eqref{eq:extension:u_22} as a
family of extended initial-value problems  parametrized by the
free parameter $y\in(0,z_1)$. Figure~\ref{fig:H-family-equation32}
illustrates the resulting family: different choices of $y$ give rise to
solutions $H(\cdot;y)$ with distinct qualitative behaviors.
\begin{figure}[htbp]
    \centering
    \includegraphics[width=0.5\textwidth]{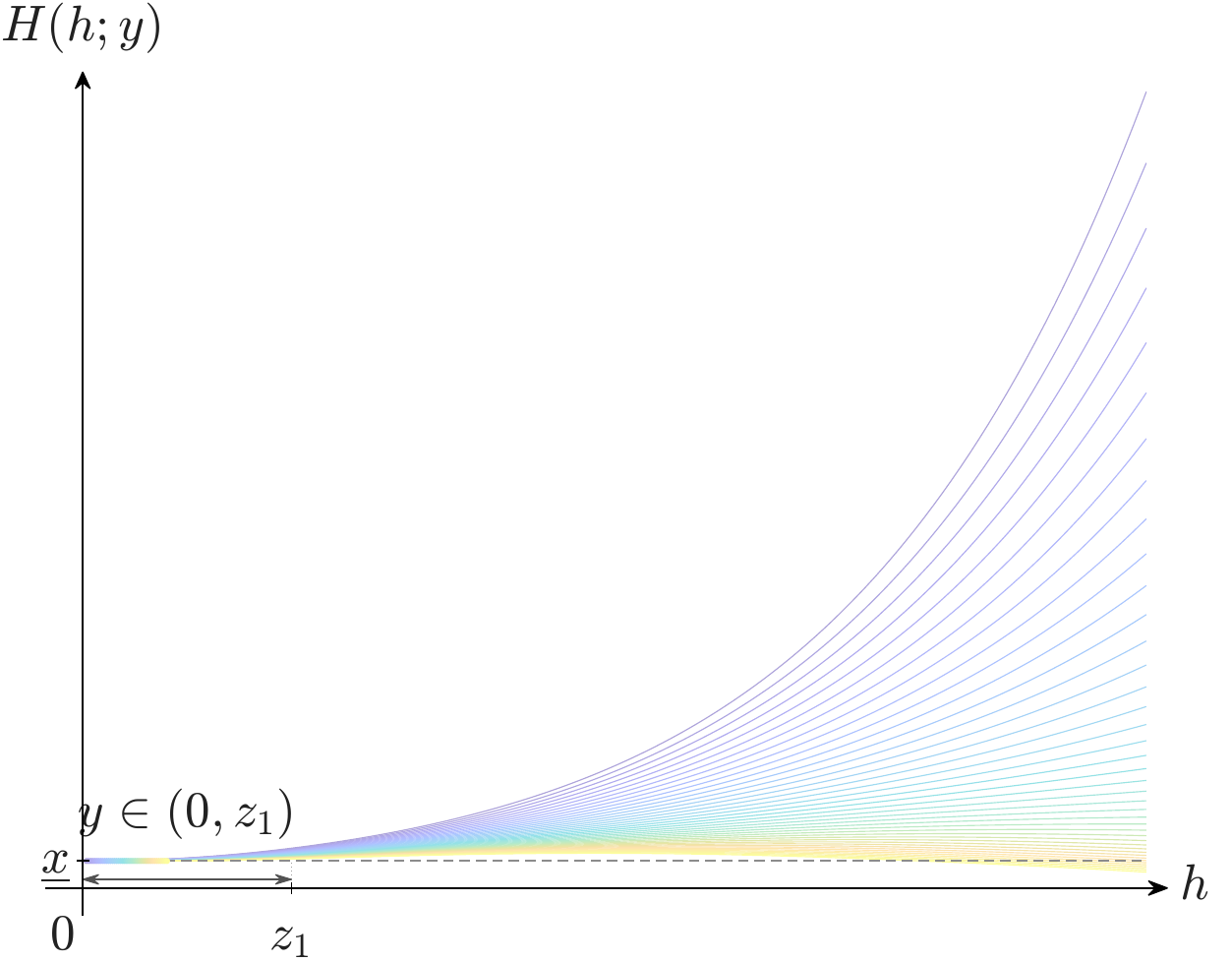}
    \caption{ Schematic illustration of the solution family $H(\cdot;y)$
    to \eqref{eq:extension:u_22} for representative initial points
    $y\in(0,z_1)$.}
    \label{fig:H-family-equation32}
\end{figure}

For each fixed $y\in(0,z_1)$, our goal in this subsection is to
establish global well-posedness of \eqref{eq:extension:u_22}. We first
record the standard local existence and uniqueness result in
Lemma~\ref{lma:local:tildeh}. The key step is the one-turning-point
property established in Proposition~\ref{lma:phi}, which imposes a
global shape restriction on the solution. Combined with the kernel
representation, this structural property
yields an {\it a priori} upper bound, see Lemma~\ref{lma:represent}. The resulting bound allows the
local solution to be extended uniquely to $[y,\infty)$, as stated in
Proposition~\ref{lma:global}.

The proof of the following local well-posedness result is
included in Appendix~\ref{sec:proof:local:tildeh} for completeness.
\begin{lemma}\label{lma:local:tildeh}
    For any $y \in \left(0,z_1\right)$, \eqref{eq:extension:u_22} admits a unique solution $H(\cdot;y)$ in $C^2([y,y+\epsilon(y)))$ for some $\epsilon(y)>0$.
\end{lemma}

The next proposition establishes a one-turning-point
property for the nonlinear ODE \eqref{eq:extension:u_22}. Specifically,
each solution is either strictly increasing throughout its interval of
existence or is strictly increasing up to a single turning point and
strictly decreasing thereafter. Thus, despite the nonlinearity of the
ODE, its solution family obeys a simple global shape restriction and
cannot exhibit oscillatory behavior. This structural regularity is key
to the tractability of the analysis that follows. The two possible
qualitative behaviors established in Proposition~\ref{lma:phi} are
illustrated in Figure~\ref{fig:prop1-turning}. By the definition of the first turning point $\phi(y)$ in \eqref{eq:underliney}, the core assertion is that once the solution passes $\phi(y)$, the sign
of $H'$ never changes: it remains strictly negative for all $h\in(\phi(y),M_y)$,
ruling out any subsequent return to zero or sign reversal.

Establishing this global shape restriction directly from
\eqref{eq:extension:u_22} is difficult because the nonlinear term
$\tilde U_2'(H(h))h^2H'(h)$ couples $H$ and $H'$. We instead exploit
a linear viewpoint along a fixed solution. We treat the equation as a first-order
linear equation for $H'$; the coupled sign structure of $H'$ and
$rH-\tilde I$ then rules out more than one sign change of $H'$. Figure~\ref{fig:prop1-sign-pattern} provides a schematic illustration
of this coupled sign argument.

\begin{proposition}\label{lma:phi}
For any $y \in (0,z_1)$, suppose that $H(\cdot;y)$ is a local (or global)
solution to \eqref{eq:extension:u_22} defined on $[y,M_y)$ for some
$M_y\in(y,\infty]$.
Define the first turning point by
\begin{align}\label{eq:underliney}
\phi(y):=\inf\{h\in(y,M_y):H'(h;y)\le 0\}.
\end{align}
In other words, if the set in \eqref{eq:underliney} is empty, we set $\phi(y):=M_y$.
Then the following conclusions hold:
\begin{enumerate}
    \item[(i)] $\phi(y)>y$. \label{it:phi_gt_y}
    \item[(ii)] The following inequalities hold:
    \begin{enumerate}
        \item $H'(h;y)>0 \quad \text{for } h\in(y,\phi(y))$, \label{it:Hprime:pos}
        \item $H'(h;y)<0 \quad \text{for } h\in(\phi(y),M_y)$. \label{it:Hprime:neg}
    \end{enumerate}
\end{enumerate}
\end{proposition}

\begin{figure}[htbp]
    \centering
    \begin{subfigure}[t]{0.4\textwidth}
        \centering
        \includegraphics[width=\textwidth]{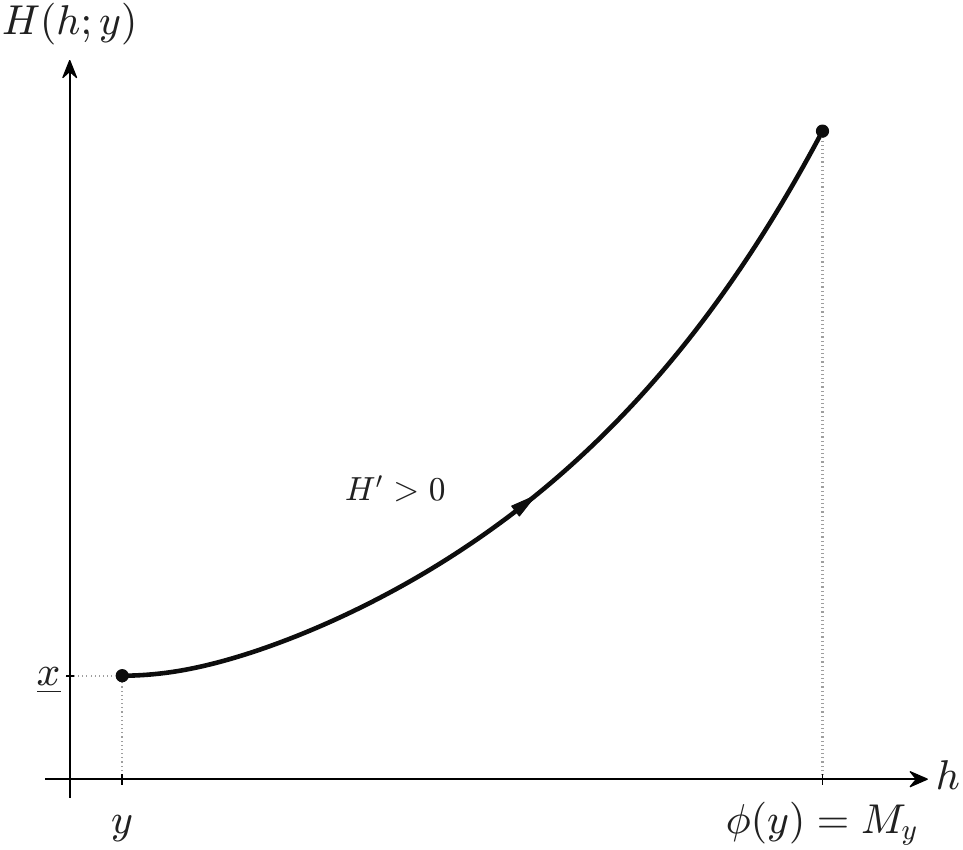}
        \caption{$\phi(y)=M_y$}
    \end{subfigure}
    \hspace{0.5cm}
    \begin{subfigure}[t]{0.4\textwidth}
        \centering
        \includegraphics[width=\textwidth]{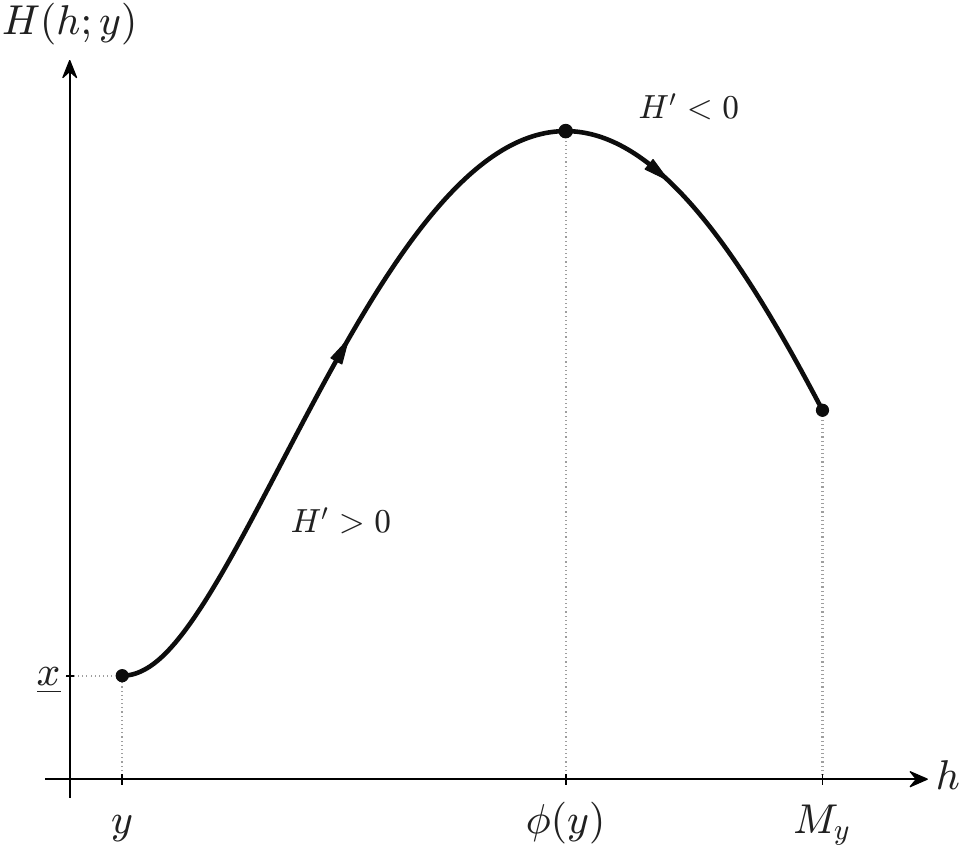}
        \caption{$\phi(y)<M_y$}
    \end{subfigure}
    \caption{Schematic illustration of the two possible
    behaviors of $H(\cdot;y)$ in Proposition~\ref{lma:phi}: the solution is
    strictly increasing in panel~(a), and is strictly increasing up to a single
    turning point and strictly decreasing thereafter in panel~(b).}
    \label{fig:prop1-turning}
\end{figure}

\begin{proof}
 Fix $y\in(0,z_1)$ and suppress the dependence on $y$. We denote $H(\cdot;y)$ by $H(\cdot)$ and $\phi(y)$ by $\phi$.
 Set
\begin{align*}
\Delta(h):=rH(h)-\tilde I(h),
\end{align*}
where $\tilde{I}$ is defined in \eqref{eq:tildeI}. 
Equation \eqref{eq:extension:u_22} can then be written as
\begin{align}\label{eq:Z:first-order}
(H'(h))'+\left(\frac{2(r-\delta)}{\theta^2h}
+\frac{2}{\theta^2}\tilde U_2'(H(h))\right)H'(h)
&=\frac{2}{\theta^2h^2}\Delta(h).
\end{align}
Since $H'(y)=0$ and, by $y<z_1$,
$\Delta(y)=r\underline{x}-\tilde I(y)>0$, letting $h\downarrow y$ in
\eqref{eq:Z:first-order} gives $(H'(y))'>0$. Thus $H'>0$ immediately to
the right of $y$, and hence $\phi>y$.

\begin{figure}[htbp]
    \centering
    \includegraphics[width=0.5\textwidth]{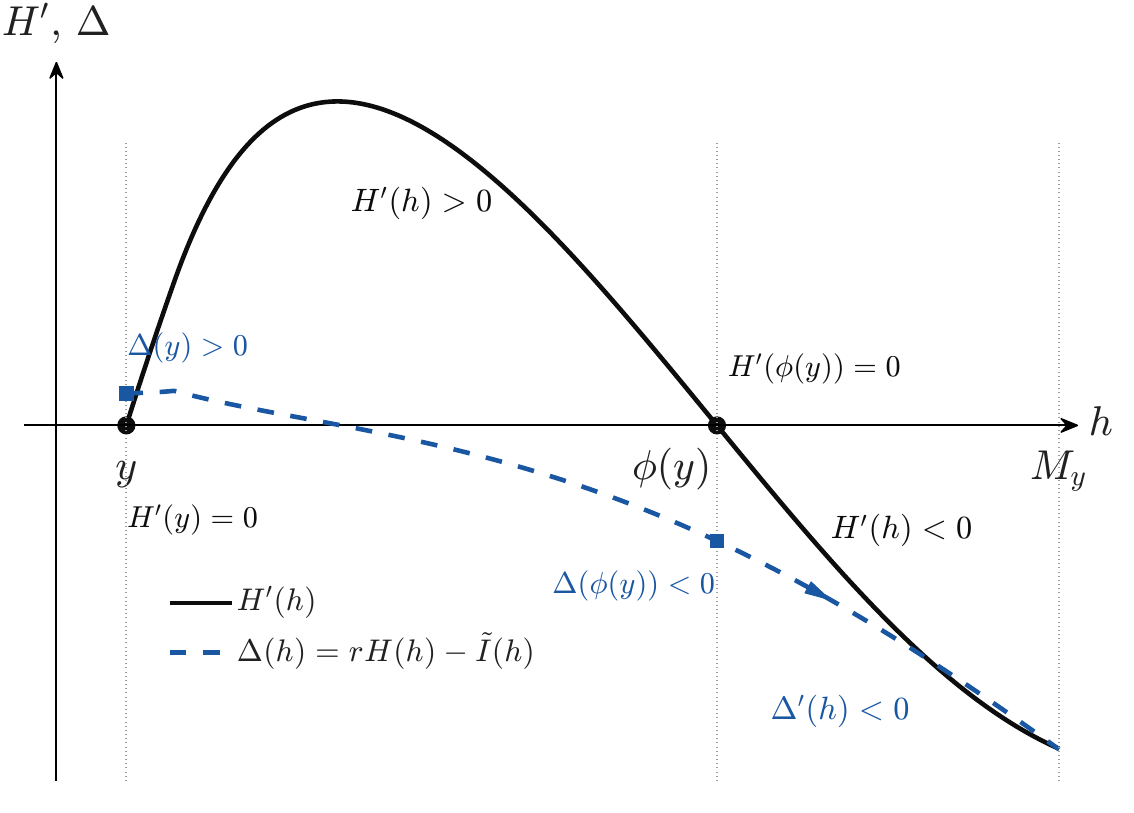}
    \caption{Schematic illustration of the sign pattern of $H'$ and
    $\Delta=rH-\tilde I$.}
    \label{fig:prop1-sign-pattern}
\end{figure}

If $\phi=M_y$, then the set in \eqref{eq:underliney} is empty,
and hence $H'>0$ on $(y,M_y)$. The result therefore follows. Suppose instead that
$\phi<M_y$. By the definition of $\phi$ and the continuity of
$H'$,
\[
H'>0\ \text{on }(y,\phi),\qquad H'(\phi)=0,
\qquad (H')'(\phi)\leq0.
\]
Evaluating \eqref{eq:Z:first-order} at $\phi$ yields
\begin{align}\label{eq:Q:phi}
\Delta(\phi)=\frac{\theta^2\phi^2}{2}(H')'(\phi)\leq0.
\end{align}
In particular, $\phi>z_0$. Otherwise, $\tilde I(\phi)=b_0$ and
$H(\phi)>H(y)=\underline{x}$, so
$\Delta(\phi)>r\underline{x}-b_0>0$, a contradiction.

We claim that the inequality in \eqref{eq:Q:phi} is strict. If $\Delta(\phi)=0$,
then $\phi>z_0$ and $\tilde I'>0$ on $(z_0,\infty)$ give
\[
\Delta'(\phi)=rH'(\phi)-\tilde I'(\phi)=-\tilde I'(\phi)<0.
\]
Hence $\Delta>0$ immediately to the left of $\phi$. Solving
\eqref{eq:Z:first-order} backward from $H'(\phi)=0$ gives
\begin{align*}
H'(h)=&\exp\!\left[-\int_\phi^h \left(\frac{2(r-\delta)}{\theta^2s}
+\frac{2}{\theta^2}\tilde U_2'(H(s))\right)\,\md s\right]
\\&\quad\quad\quad\quad\quad\quad\int_\phi^h\frac{2\Delta(u)}{\theta^2u^2}
\exp\!\left[\int_\phi^u \left(\frac{2(r-\delta)}{\theta^2s}
+\frac{2}{\theta^2}\tilde U_2'(H(s))\right)\,\md s\right]\md u.
\end{align*}
For $h<\phi$ sufficiently close to $\phi$, the right-hand side is
negative, contradicting $H'>0$ on $(y,\phi)$. Therefore
$\Delta(\phi)<0$, and \eqref{eq:Z:first-order} gives $(H')'(\phi)<0$.
Consequently, $H'<0$ immediately to the right of $\phi$.

It remains to exclude a return to zero. Suppose, to the contrary, that
such a return occurs, and let
\[
\psi:=\inf\{h\in(\phi,M_y):H'(h)\geq0\}.
\]
Then, $\phi<\psi\le M_y$.
Moreover, $H'<0$ on $(\phi,\psi)$ and $H'(\psi)=0$. Since
$\phi>z_0$,
\[
\Delta'(h)=rH'(h)-\tilde I'(h)<0,
\qquad h\in(\phi,\psi),
\]
and therefore $\Delta(\psi)<\Delta(\phi)<0$. Because $H'$ reaches zero from
below at $\psi$, we must have $(H')'(\psi)\geq0$. On the other hand,
since $\Delta(\psi)<0$ and $H'(\psi)=0$,
evaluating \eqref{eq:Z:first-order} at $h=\psi$ gives
\[
(H')'(\psi)=\frac{2}{\theta^2\psi^2}\Delta(\psi)<0,
\]
a contradiction. Hence $H'<0$ on $(\phi,M_y)$, completing the proof.
\end{proof}

By regarding 
\(
\tilde I(h)+\tilde U_2'(H(h))h^2H'(h)
\)
as a prescribed forcing term, we reduce the semilinear equation \eqref{eq:extension:u_22}
to a second-order linear Cauchy--Euler equation (\cite{teschl2012ordinary}) for $H$.
The solution of this linear equation gives the kernel representation 
in part~\textup{(i)} of Lemma~\ref{lma:represent}. This representation
remains implicit and self-consistent, since the forcing term still
depends on $H$ and $H'$. 
Combining this representation with Proposition~\ref{lma:phi}
yields the upper estimate for $H$ in part~\textup{(ii)}.

\begin{lemma}\label{lma:represent}
For any $y \in (0,z_1)$, suppose that $H(\cdot;y)$ is a local (or global)
solution to \eqref{eq:extension:u_22} defined on $[y,M_y)$ for some
$M_y\in(y,\infty]$. Then
\medskip

\noindent\textup{(i)} for any $h\in[y,M_y)$,
  \begin{align}
        &H(h;y)\nonumber\\
        =\ &\frac{2}{\theta^2p_1(p_1+p_2)}y^{-p_1}\left\{r\underline{x}-p_1\int_{y}^{h}\left[\tilde{I}(z)+\tilde{U}_2'(H(z;y))z^2H'(z;y)\right]z^{-p_1-1}\md z \,y^{p_1}\right\}h^{p_1}\nonumber\\
    &+\frac{2}{\theta^2p_2(p_1+p_2)}y^{p_2}\left\{r\underline{x}+p_2
    \int_{y}^{h}\left[\tilde{I}(z)+\tilde{U}_2'(H(z;y))z^2H'(z;y)\right]z^{p_2-1}\md z\,y^{-p_2}\right\} h^{-p_2}
    \label{eq:tildeH},
  \end{align}
where $p_1>0$ and $-p_2<0$ are the two roots of
\eqref{eq:p}.

\noindent\textup{(ii)}  define
\[
\bar H(h;y):=\underline{x}h^{p_1}y^{-p_1}\geq\underline{x},
\qquad h\geq y.
\]
Then
\[
H(h;y)\leq \bar H(h;y),\qquad h\in[y,M_y).
\]
\end{lemma}
The proof of Lemma~\ref{lma:represent} is provided in the
Appendix~\ref{sec:proof:lma:represent}.

With the upper bound in Lemma~\ref{lma:represent}(ii), standard
ODE extension theory yields the following global well-posedness
result. Specifically, the upper bound on $H$ ensures that $\tilde U_2'(H)$
is bounded on any finite interval. Then, the Gr\"onwall estimate rules out
finite-time blow-up, thereby allowing global extension to the entire half-line
$[y,\infty)$.

 \begin{proposition}\label{lma:global}
  For any $y \in (0,z_1)$, the ODE \eqref{eq:extension:u_22} admits a unique solution $H(\cdot;y)\in C^2([y,\infty))$.
\end{proposition}
The proof of Proposition \ref{lma:global}  is provided in Appendix \ref{sec:proof:lma:global}.

\subsection{Selection of a Suitable Global Solution}
\label{sec:case1:selection}
\label{sec:case1:properties}

We have shown in Proposition~\ref{lma:global} that, for any $y \in (0,z_1)$, the ODE \eqref{eq:extension:u_22} admits a unique solution $H(\cdot;y)\in C^2([y,\infty))$. However, we still need to identify a suitable $y_0$ such that the
corresponding solution $H(\cdot;y_0)$ has the properties required to
apply the verification theorem, Theorem~\ref{thm:verification}.

Recall that $p_1>0$ and $-p_2<0$ are the two roots of
\eqref{eq:p}. Let $\chi(\lambda)$ denote the characteristic polynomial associated
with \eqref{eq:p}. Under Assumption~\ref{assump:parameter}, a direct calculation gives
\[
\chi\left(\frac{1}{1-\gamma}\right)
=
\frac{1}{1-\gamma}
\left[
\delta-\frac{\gamma}{2(1-\gamma)}\theta^2-r\gamma
\right]>0.
\]
Since $\chi$ has roots $-p_2<0<p_1$ and a negative leading coefficient,
it follows that
\begin{align}\label{eq:p1:gamma}
p_1>\frac{1}{1-\gamma}.
\end{align}
Assumption~\ref{assmup:U12}(i) and the definition of
$\tilde I$ in \eqref{eq:tildeI} imply that, with
$M_{1,1}:=C_{1,1}^{\frac{1}{1-\gamma}}$,
\begin{equation}\label{eq:tildeI:growth}
\tilde I(z)\leq M_{1,1}z^{\frac{1}{1-\gamma}}
\quad\text{for }z\geq z_0,
\qquad
\tilde I(z)=b_0
\quad\text{for }0<z<z_0.
\end{equation}
With these preliminary estimates in hand, we first introduce the
function $G$ and establish its basic properties in the following lemma.

\begin{lemma}\phantomsection\label{lma:G}
\begin{enumerate}
\item[(i)] For $y\in(0,z_1]$, define
\[
\bar G(y):=r\underline{x}
-p_1y^{p_1}\int_y^\infty\tilde I(z)z^{-p_1-1}\md z.
\]
There exists a unique $\bar y\in(0,z_1)$ such that $\bar G(\bar y)=0$. 

\item[(ii)]  For
$y\in(0,\bar y]$, define
\begin{equation}\label{eq:G:def}
G(y):=r\underline{x}
-p_1y^{p_1}\int_y^\infty
\left[\tilde I(z)+\tilde U_2'(H(z;y))z^2H'(z;y)\right]
z^{-p_1-1}\md z.
\end{equation}
Then $G$ is well-defined and continuous on $(0,\bar y]$, and
$\liminf\limits_{y\downarrow0}G(y)>0$. Consequently, the point
\[
y_0:=\sup\{0<y\leq\bar y:G(y)>0\}\in(0,\bar y]
\]
is well-defined and satisfies $G(y_0)\geq0$.
\end{enumerate}
\end{lemma}
The proof of Lemma~\ref{lma:G} is provided in the Appendix;
see Section~\ref{sec:proof:lma:G}.

We next explain why $G$ provides the appropriate criterion for selecting
the initial point.
Substituting \eqref{eq:G:def} into the kernel representation
\eqref{eq:tildeH} yields, for every $y\in(0,\bar y]$ and $h\geq y$,
\begin{align}\label{eq:H:G:preIBP}
&H(h;y)=
\frac{2G(y)}{\theta^2p_1(p_1+p_2)y^{p_1}}h^{p_1}\nonumber\\
&+\frac{2}{\theta^2(p_1+p_2)}
\Bigg\{\left[
\int_h^\infty\tilde I(z)z^{-p_1-1}\md z
+\int_h^\infty
\tilde U_2'(H(z;y))H'(z;y)z^{1-p_1}\md z
\right]h^{p_1}\nonumber\\
&+\left[
\frac{r\underline{x}}{p_2}y^{p_2}
+\int_y^h\tilde I(z)z^{p_2-1}\md z
+\int_y^h
\tilde U_2'(H(z;y))H'(z;y)z^{p_2+1}\md z
\right]h^{-p_2}\Bigg\}.
\end{align}
Integrating by parts in the terms involving
$\tilde U_2'(H(z;y))H'(z;y)$, as in the proof of
Lemma~\ref{lma:G}(ii), then gives
\begin{align}\label{eq:H:G:representation}
&H(h;y)-
\frac{2G(y)}{\theta^2p_1(p_1+p_2)y^{p_1}}h^{p_1}\nonumber\\
=&\frac{2}{\theta^2(p_1+p_2)}
\Bigg\{\left[
\int_h^\infty\tilde I(z)z^{-p_1-1}\md z
+(p_1-1)\int_h^\infty
\tilde U_2(H(z;y))z^{-p_1}\md z
\right]h^{p_1}\nonumber\\
&+\Bigg[
\frac{r\underline{x}}{p_2}y^{p_2}
+\int_y^h\tilde I(z)z^{p_2-1}\md z
-\tilde U_2(\underline{x})y^{p_2+1}
-(p_2+1)\int_y^h
\tilde U_2(H(z;y))z^{p_2}\md z
\Bigg]h^{-p_2}\Bigg\}.
\end{align}
Recall from \eqref{eq:p1:gamma} that
$p_1>\frac{1}{1-\gamma}$. We claim that, for every
$y\in(0,\bar y]$,
\begin{equation}\label{eq:H:G:asymptotic}
H(h;y)-
\frac{2G(y)}{\theta^2p_1(p_1+p_2)y^{p_1}}h^{p_1}
=o\!\left(h^{p_1}\right),
\qquad h\to\infty,
\end{equation}
Indeed, by \eqref{eq:tildeI:growth}, direct integration gives,
for every $h\geq y$,
\begin{equation}\label{eq:H:G:I-bound}
\begin{split}
0&\le h^{p_1}\int_h^\infty\tilde I(z)z^{-p_1-1}\md z
+h^{-p_2}\int_y^h\tilde I(z)z^{p_2-1}\md z\\
&\le
 \max\left\{M_{1,1},b_0y^{-\frac{1}{1-\gamma}}\right\}
\left(
\frac{1}{p_1-\frac{1}{1-\gamma}}
+\frac{1}{p_2+\frac{1}{1-\gamma}}
\right)h^{\frac{1}{1-\gamma}}.
\end{split}
\end{equation}
Similarly, the properties of $\tilde U_2$ and
Lemma~\ref{lma:represent}(ii) yield
$\tilde U_2(H(z;y))
\leq C_{2,1}\underline{x}^{\gamma}y^{-\gamma p_1}
z^{\gamma p_1}$ for $z\geq y$.
Direct integration then gives, for every $h\geq y$,
\begin{equation}\label{eq:H:G:U-bound}
\begin{split}
0&\le(p_1-1)h^{p_1}\int_h^\infty
\tilde U_2(H(z;y))z^{-p_1}\md z +(p_2+1)h^{-p_2}\int_y^h
\tilde U_2(H(z;y))z^{p_2}\md z\\
&\leq
C_{2,1}\underline{x}^{\gamma}y^{-\gamma p_1}\left[
\frac{p_1-1}{(1-\gamma)p_1-1}
+\frac{p_2+1}{p_2+\gamma p_1+1}
\right]h^{1+\gamma p_1}.
\end{split}
\end{equation}
Both exponents $\frac{1}{1-\gamma}$ and
$1+\gamma p_1$ are strictly smaller than $p_1$.
Together with $p_2>0$, \eqref{eq:H:G:I-bound} and
\eqref{eq:H:G:U-bound} therefore show that
the right hand side of \eqref{eq:H:G:representation} is
$o(h^{p_1})$, proving \eqref{eq:H:G:asymptotic}.

Consequently, if $G(y)>0$, then
\eqref{eq:H:G:asymptotic} shows that $H(h;y)$ grows at the rate
$h^{p_1}$. Since $p_1>\frac{1}{1-\gamma}$, this violates the upper
growth bound $H(h)\leq C_{H,0}h^{\frac{1}{1-\gamma}}$ in
\eqref{eq:H:property}. On the other hand, if $G(y)<0$, then
$H(h;y)\to-\infty$, which contradicts the requirement in
\eqref{eq:H':property} that $H$ be strictly increasing. Thus, neither
sign is admissible, and a
suitable initial point must satisfy $G(y_0)=0$.

Lemma~\ref{lma:G} is devoted to constructing such an initial point
$y_0$. Since $G$ depends on $y$ implicitly through $H(\cdot;y)$,
locating a root of $G$ directly is difficult. We instead remove the
term involving $H(\cdot;y)$ from $G$ to obtain a new function
$\bar G$, whose root determines a suitable search interval, and then
select $y_0$ by a supremum argument. With $G(y_0)\geq0$
established, the following proposition first proves that
$H(\cdot;y_0)$ is strictly increasing and then shows that $G(y_0)=0$.

\begin{proposition}\label{lma:H'}
Let $y_0$ be the point defined in Lemma~\ref{lma:G}(ii), and let
$H(\cdot;y_0)$ denote the unique global solution of
\eqref{eq:extension:u_22} corresponding to the initial point $y_0$. Then:
\begin{enumerate}
\item[(i)] The function $H(\cdot;y_0)$ satisfies
\[
H'(h;y_0)>0,\qquad h\in(y_0,\infty).
\]
\item[(ii)] The initial point satisfies
\[
y_0\in(0,\bar y),\qquad G(y_0)=0.
\]
\end{enumerate}
\end{proposition}
\begin{proof}
\textit{Part~\textup{(i)}.}
By Proposition~\ref{lma:phi}, it suffices to show that
$\phi(y_0)=\infty$. If $\phi(y_0)<\infty$,
Proposition~\ref{lma:phi} gives $H'(h;y_0)<0$ for all
$h>\phi(y_0)$. Consequently,
\[
L:=\lim_{h\to\infty}H(h;y_0)
\]
exists in $[-\infty,H(\phi(y_0);y_0)]$.
By the definition of $\tilde I$ in \eqref{eq:tildeI}, for
$b>r\underline{x}$,
$\tilde I(z)\geq b$ if and only if $z\geq1/U_1'(b)$. Moreover,
$1/U_1'(b)\geq1/U_1'(r\underline{x})=z_1$.
Hence, by the monotone convergence theorem,
\begin{align}
\int_{z_1}^{\infty}\tilde I(z)z^{-2}\md z\!\geq\!\int_{z_1}^{\infty}\int_{r\underline{x}}^\infty
\mathbf 1_{\{\tilde I(z)\geq b\}}\md b\,z^{-2}\md z\!=\!\int_{r\underline{x}}^\infty
\int_{1/U_1'(b)}^\infty z^{-2}\md z\,\md b\!=\!\int_{r\underline{x}}^\infty U_1'(b)\md b
\!=\!\infty,\label{eq:I:tail:divergence}
\end{align}
where the last equality follows from $U_1(\infty)=\infty$.
On the other hand, dividing ODE~\eqref{eq:extension:u_22} by $h^2$ yields
\begin{align}
\frac{\tilde I(h)}{h^2}
=\frac{\md}{\md h}
\left[-\frac{\theta^2}{2}H'(h;y_0)
-(r-\delta)\frac{H(h;y_0)}{h}-\tilde U_2(H(h;y_0))\right]
+\delta\frac{H(h;y_0)}{h^2}.\label{eq:H:tail:identity}
\end{align}
Suppose first that $L$ is finite. Then $H(\cdot;y_0)$ is bounded on
$[z_1,\infty)$, $\tilde U_2(H(h;y_0))$ converges, and
$\int_{z_1}^{\infty}|H(h;y_0)|h^{-2}\md h<\infty$. By integrating
\eqref{eq:H:tail:identity} over $[z_1,h]$ and using
\eqref{eq:I:tail:divergence}, we obtain
\[
-\frac{\theta^2}{2}H'(h;y_0)
-(r-\delta)\frac{H(h;y_0)}{h}
-\tilde U_2(H(h;y_0))\longrightarrow\infty
\quad\text{as }h\to\infty.
\]
It follows that $H'(h;y_0)\to-\infty$, contradicting the finiteness of $L$.

Suppose next that $L=-\infty$. Then there
exists $h_1>\phi(y_0)$ such that
$H(h;y_0)\leq0$ for all $h\geq h_1$. Since $\tilde U_2'=0$ on
$(-\infty,0]$, $\tilde U_2(H(h;y_0))=\tilde U_2(0)$ for $h\geq h_1$.
The two integral terms involving $\tilde U_2$ in
\eqref{eq:H:G:representation} therefore satisfy, for $h\geq h_1$,
\begin{align*}
&(p_1-1)h^{p_1}\int_h^\infty
\tilde U_2(H(z;y_0))z^{-p_1}\md z-(p_2+1)h^{-p_2}
\int_{y_0}^h\tilde U_2(H(z;y_0))z^{p_2}\md z\\
=&\left[\tilde U_2(0)h_1^{p_2+1}
-(p_2+1)\int_{y_0}^{h_1}\tilde U_2(H(z;y_0))z^{p_2}\md z\right]h^{-p_2}
\longrightarrow0
\quad\text{as }h\to\infty.
\end{align*}
Moreover, $\tilde I$ is nondecreasing and, since $U_1'(\infty)=0$,
$\tilde I(h)=(U_1')^{-1}(1/h)\to\infty$ as $h\to\infty$.
Therefore, the first integral term involving $\tilde I$ in
\eqref{eq:H:G:representation} satisfies
\begin{align*}
h^{p_1}\int_h^\infty\tilde I(z)z^{-p_1-1}\md z
\geq \tilde I(h)h^{p_1}\int_h^\infty z^{-p_1-1}\md z
=\frac{1}{p_1}\tilde I(h)\longrightarrow\infty.
\end{align*}
Substituting the preceding two estimates into
\eqref{eq:H:G:representation} and using $G(y_0)\geq0$ from
Lemma~\ref{lma:G}(ii), together with $\tilde I\geq0$ and $p_2>0$, we
obtain $H(h;y_0)\to\infty$, contradicting $L=-\infty$.
Thus, part~\textup{(i)} follows.
\medskip

\noindent\textit{Part~\textup{(ii)}.}
Suppose, for contradiction, that $y_0=\bar y$. By the definitions of
$G$ and $\bar G$ in
Lemma~\ref{lma:G}(i)--(ii), and since
$U_2'>0$ and $H'(z;y_0)>0$ for all $z>y_0$,
\[
G(y_0)-\bar G(y_0)
=-p_1y_0^{p_1}\int_{y_0}^{\infty}
U_2'(H(z;y_0))H'(z;y_0)z^{1-p_1}\md z<0.
\]
Since $\bar G(\bar y)=0$, this contradicts $G(y_0)\geq0$; hence $y_0<\bar y$. The continuity of $G$ and
the definition of $y_0$ then yield $G(y_0)=0$.
\end{proof}

With $y_0$ now determined, we henceforth write
$H=H(\cdot;y_0)$ for the corresponding global solution.
Proposition~\ref{lma:H'}(i) shows that $H$ is strictly increasing.
Since $H(y_0)=\underline{x}$, it follows that
$H(h)\geq\underline{x}$ for every $h\geq y_0$. Hence,
\eqref{eq:extension:u_22} coincides along this solution with the
original problem \eqref{eq:extension:u_222}, because
$\tilde U_2=U_2$ on $[\underline{x},\infty)$ and
$\tilde I(h)=I(1/h)$. The strict monotonicity of $H$ also yields the
lower inequality in \eqref{eq:H':property}. Moreover, Proposition~\ref{lma:H'}(ii) gives \(G(y_0)=0\). Substituting this identity into \eqref{eq:H:G:representation}  yields a simplified representation of \(H(h;y_0)\), which serves as the starting point for the remaining growth and derivative estimates established in the next proposition.

\begin{proposition}\label{lem:limsup}
The auxiliary function $H=H(\cdot;y_0)$ satisfies
\begin{align*}
\textup{(i)}\quad&H(h)\leq C_{H,0}h^{\frac{1}{1-\gamma}},
\qquad h\in[y_0,\infty),\\
\textup{(ii)}\quad&H(h)\geq
\frac{2}{\theta^2p_1(p_1+p_2)}I\left(\frac1h\right),
\qquad h\in[y_0,\infty),\\
\textup{(iii)}\quad&\frac{hH'(h)}{H(h)}\leq p_1,
\qquad h\in(y_0,\infty),
\end{align*}
where $C_{H,0}>0$ is a constant.
\end{proposition}

\begin{proof}
\textit{Part~\textup{(i)}.}
Setting \(y=y_0\) in \eqref{eq:H:G:representation}, using \(G(y_0)=0\) and $\tilde{U}_2=U_2\geq 0$, we obtain
    \begin{align}
    {H}(h)
    \le&\frac{2}{\theta^2(p_1+p_2)}\left\{\left[\int_h^\infty \tilde{I}(z)z^{-p_1-1}\md z+(p_1-1)\int_h^\infty {U}_2({H}(z))z^{-p_1}\md z\right]h^{p_1}\right.\nonumber\\
    &\quad\left.+\left[\frac{1}{p_2}r\underline{x}y_0^{p_2}+\int_{y_0}^{h}\tilde{I}(z)z^{p_2-1}\md z\right]h^{-p_2}\right\}.\label{eq:H:growth-upper}
    \end{align}
By \eqref{eq:H:G:I-bound} 
evaluated
at $y=y_0$ and $p_2>0$,
there is some constant $C_0>0$ such that
\begin{equation}\label{eq:H:bootstrap}
H(h)\leq
C_0h^{\frac{1}{1-\gamma}}
+\frac{2(p_1-1)}{\theta^2(p_1+p_2)}
h^{p_1}\int_h^\infty U_2(H(z))z^{-p_1}\md z,
\qquad h\geq y_0.
\end{equation}
 Then, the idea is to repeatedly substitute the resulting bound for \(H\) into the right-hand side of \eqref{eq:H:bootstrap}. Each iteration improves the growth estimate until the desired bound is obtained.
 Specifically,  define $a_0$ and $c_0$ by
\[
a_0:=p_1,
\qquad
c_0:=\underline{x}y_0^{-p_1},
\]
and, for $n\geq0$, recursively set
\begin{align*}
a_{n+1}:=1+\gamma a_n,\quad
c_{n+1}:=
C_0+
\frac{2(p_1-1)C_{2,1}}
{\theta^2(p_1+p_2)(p_1-\gamma p_1-1)}
c_n^\gamma.
\end{align*}
The first recursion gives
\[
a_n=
\frac{1}{1-\gamma}
+\gamma^n\left(p_1-\frac{1}{1-\gamma}\right)
\in\left[\frac{1}{1-\gamma},p_1\right]
\quad\text{and}\quad
a_n\overset{n\to\infty}{\longrightarrow}\frac{1}{1-\gamma}.
\]
We prove by induction that
\begin{equation}\label{eq:H:recursive-bound}
H(h)\leq c_n\left(h^{\frac{1}{1-\gamma}}+h^{a_n}\right),
\qquad h\geq y_0,\quad n\geq0.
\end{equation}
For $n=0$, Lemma~\ref{lma:represent}(ii) gives
$H(h)\leq c_0h^{a_0}$, so \eqref{eq:H:recursive-bound} holds.
Fix $k\geq0$ and suppose that \eqref{eq:H:recursive-bound} holds for
$n=k$.
By the induction hypothesis,
Assumption~\ref{assmup:U12}(ii) and $0<\gamma<1$, we have, for
$z\geq y_0$,
\[
U_2(H(z))
\leq C_{2,1}c_k^\gamma
\left(z^{\frac{1}{1-\gamma}}+z^{a_k}\right)^\gamma
\leq C_{2,1}c_k^\gamma
\left(z^{\frac{\gamma}{1-\gamma}}+z^{\gamma a_k}\right).
\]
Since
$p_1-\frac{1}{1-\gamma}>p_1-\gamma p_1-1>0$
and
$p_1-\gamma a_k-1\geq p_1-\gamma p_1-1>0$,
direct integration yields
\begin{align*}
&h^{p_1}\int_h^\infty U_2(H(z))z^{-p_1}\md z
\leq C_{2,1}c_k^\gamma h^{p_1}\int_h^\infty
\left(z^{\frac{\gamma}{1-\gamma}-p_1}
+z^{\gamma a_k-p_1}\right)\md z\\
=\ &C_{2,1}c_k^\gamma
\left[
\frac{h^{\frac{1}{1-\gamma}}}
{p_1-\frac{1}{1-\gamma}}
+
\frac{h^{1+\gamma a_k}}
{p_1-\gamma a_k-1}
\right]\leq
\frac{C_{2,1}c_k^\gamma}
{p_1-\gamma p_1-1}
\left(h^{\frac{1}{1-\gamma}}+h^{a_{k+1}}\right).
\end{align*}
Combining this estimate with \eqref{eq:H:bootstrap} and the definition
of $c_{k+1}$ gives
\begin{align*}
H(h)
\leq
c_{k+1}
\left(h^{\frac{1}{1-\gamma}}+h^{a_{k+1}}\right),
\qquad h\geq y_0.
\end{align*}
This proves \eqref{eq:H:recursive-bound} for $n=k+1$ and completes
the induction.

Since $\gamma\in(0,1)$, there exists $\bar c\geq c_0$ such that
\[
C_0+
\frac{2(p_1-1)C_{2,1}}
{\theta^2(p_1+p_2)(p_1-\gamma p_1-1)}
\bar c^\gamma
\leq\bar c.
\]
The recursion and the preceding inequality imply inductively that
$c_n\leq\bar c$ for every $n$.
Hence, by \eqref{eq:H:recursive-bound}, for each $h\geq y_0$,
letting $n\to\infty$ yields
\[
H(h)\leq 2\bar c\,h^{\frac{1}{1-\gamma}}.
\]
Thus, the assertion holds with
\(
C_{H,0}:=2\bar c.
\)
\medskip

\noindent\textit{Part~\textup{(ii)}.}
Applying \eqref{eq:H:G:preIBP} with $y=y_0$ and $G(y_0)=0$ gives
\begin{align}\label{eq:H:G:y0}
H(h)
=\frac{2}{\theta^2(p_1+p_2)}
\Bigg\{\left[
\int_h^\infty\tilde I(z)z^{-p_1-1}\md z
+\int_h^\infty
U_2'(H(z))H'(z)z^{1-p_1}\md z
\right]h^{p_1}\nonumber\\
+\left[
\frac{r\underline{x}}{p_2}y_0^{p_2}
+\int_{y_0}^h\tilde I(z)z^{p_2-1}\md z
+\int_{y_0}^h
U_2'(H(z))H'(z)z^{p_2+1}\md z
\right]h^{-p_2}\Bigg\},
 ~h\geq y_0.
\end{align}
Using $U_2'\geq0$ and $p_2>0$, and the fact that
$\tilde I(z)\geq\tilde I(h)$ for $z\geq h$, we obtain
\begin{align*}
H(h)
\geq \frac{2}{\theta^2(p_1+p_2)}\tilde I(h)h^{p_1}
\int_h^\infty z^{-p_1-1}\md z
=\frac{2}{\theta^2p_1(p_1+p_2)}\tilde I(h).
\end{align*}
\medskip

\noindent\textit{Part~\textup{(iii)}.}
Differentiating \eqref{eq:H:G:y0} gives
\begin{align*}
    H'(h)=\frac{2p_1}{\theta^2(p_1+p_2)}\left[\int_h^\infty \tilde{I}(z)z^{-p_1-1}\md z+\int_h^\infty U_2'(H(z))H'(z)z^{1-p_1}\md z\right]h^{p_1-1}\\
    -\frac{2p_2}{\theta^2(p_1+p_2)}\left[\frac{1}{p_2}r\underline{x}y_0^{p_2}+\int_{y_0}^{h}\tilde{I}(z)z^{p_2-1}\md z+\int_{y_0}^hU_2'(H(z))H'(z)z^{1+p_2}\md z\right]h^{-p_2-1}.
\end{align*}
Using $U_2'\geq 0$ and $p_2>0$, we obtain
\begin{align*}
H'(h)
&\leq \frac{2p_1}{\theta^2(p_1+p_2)}\left[\int_h^\infty \tilde{I}(z)z^{-p_1-1}\md z+\int_h^\infty U_2'(H(z))H'(z)z^{1-p_1}\md z\right]h^{p_1-1}\\
&\leq\frac{p_1H(h)}{h},
\end{align*}
where the second inequality uses the non-negativity of the second line of \eqref{eq:H:G:y0}.
\end{proof}

Overall, Propositions \ref{lma:global}, \ref{lma:H'}, and
\ref{lem:limsup} together establish that the resulting function $H$
satisfies \eqref{eq:extension:u_222}, \eqref{eq:H:property}, and
\eqref{eq:H':property}. Therefore, the proof of Theorem \ref{thm:H} is
complete.

\subsection{Reconstruction and Verification of the Value Function}
We now invert the transform introduced in
Section~\ref{sec:transform} to construct a candidate solution $v$ of the
original HJB equation \eqref{eq:hjb:simplify}. The theorem below establishes the regularity, boundary
behavior, and growth properties of $v$ needed to apply
Theorem~\ref{thm:verification}.

\begin{theorem}\label{thm:v}
    Under Assumptions \ref{assump:parameter}, \ref{assmup:U12}, and \ref{assump:u_2'}, suppose further that $\alpha x_0>\frac{b_0}{r}$. Then there is an $A\in\left(U_1'(r\underline{x}),\infty\right)$ such that the following ODE admits a solution in $C^1([\underline{x},\infty))\cap C^3((\underline{x},\infty))$:
    \begin{align}\label{eq:ode:final}
    \left\{\begin{aligned}
		&\delta \cdot v(x)=D(v'(x))+ U_2(x)
		-\frac{1}{2}\theta^2\frac{\left(v'(x)\right)^2}{v''(x)}+rxv'(x),  \quad x\in(\underline{x},\infty),\\
        &v'(\underline{x})=A,\quad \lim\limits_{x\downarrow \underline{x}}v''(x)=-\infty.
        \end{aligned}
        \right.
	\end{align}
 Moreover, there exist constants $C_{v,2}>0$,
    $C_{v,3}>0$, $C_{v,4}>0$, and $C_{v,5}>0$ such that
    \begin{equation}\label{eq:property:v}
    \begin{aligned}
    &(i) \quad v(x)\leq C_{v,2}x^\gamma,\quad
    0<v'(x)\leq C_{v,3}x^{\gamma-1},\quad
    I(v'(x))\leq C_{v,4}x,
    \qquad x\geq\underline{x},\\
    &(ii)\quad v''(x)<0,\quad
    -\frac{v'(x)}{xv''(x)}\leq C_{v,5},
    \qquad x>\underline{x}.
    \end{aligned}
    \end{equation}
\end{theorem}
The proof of Theorem~\ref{thm:v} is provided in Appendix~\ref{appendix:sec:case1:2}.

\begin{remark}\label{rmk:case1:properties}
To better understand \eqref{eq:property:v}, it is useful to
compare it with the unconstrained CRRA benchmark in Remark~\ref{rmk:unconstraint}, where $U_1(b)=b^\gamma/\gamma$ and
$U_2(x)=\eta x^\gamma/\gamma$. In that case, the value function admits the explicit form $V_{un}(x)=\bar{b}\frac{x^\gamma}{\gamma}$. The upper bounds on $v$ and $v'$ in \eqref{eq:property:v} indicate that, despite the lack of an explicit representation for $v$ due to constraints and general functions, both $v(x)$ and $v'(x)$
 grow at most at the same orders as $V_{un}(x)$ and $V_{un}'(x)$ for large values of $x$. Moreover, under the unconstrained CRRA benchmark, the optimal investment proportion is constant, and the optimal benefit rate is proportional to wealth (so that the benefit-to-wealth ratio is constant). In contrast, the bounds on $I(v'(x))$ and $-\frac{v'(x)}{xv''(x)}$ in \eqref{eq:property:v} indicate that, under constraints, the optimal investment proportion and the benefit-to-wealth ratio are no longer necessarily constant; nevertheless, both quantities remain uniformly bounded. These observations are consistent with the numerical experiments in Section~\ref{sec:numerical}; see Figure \ref{fig:investment_benefit}  for the illustration.
\end{remark}

The sign properties of $v'$ and $v''$ established in
Theorem~\ref{thm:v} ensure that the candidate feedback functions
$(b^*,\pi^*)$ defined in \eqref{eq:b^*} are well-defined. We next establish global existence and pathwise uniqueness for the state
equation \eqref{eq:verify:wealth} induced by these feedback functions. The main difficulty is the lack of standard Lipschitz regularity of the feedback coefficients at $\underline{x}$. We choose a suitable change of measure under which the drift is locally Lipschitz and prove that the diffusion coefficient is locally H\"older continuous of order $\frac{1}{2}$. A truncation argument then gives a unique strong solution remaining in $[\underline{x},\infty)$, while the linear growth bounds in \eqref{eq:property:v} rule out finite-time explosion and yield global well-posedness.

\begin{theorem}\label{thm:strong:case1}
 Under Assumptions \ref{assump:parameter}, \ref{assmup:U12} and \ref{assump:u_2'},  suppose further that $\alpha x_0>\frac{b_0}{r}$, let $v$ be the solution of \eqref{eq:ode:final} obtained in Theorem \ref{thm:v}. Then, for all $x\geq \underline{x}$, SDE \eqref{eq:verify:wealth} admits a unique strong solution satisfying $X^{x;*}\geq \underline{x}$.
\end{theorem}
The proof of Theorem \ref{thm:strong:case1}  is put in Appendix \ref{appendix:sec:case1:3}.

Combining Theorems~\ref{thm:v} and
\ref{thm:strong:case1} shows that all the conditions of the verification
theorem, Theorem~\ref{thm:verification}, are satisfied. 
More precisely, Assumption~\ref{assmup:U12} implies
Assumption~\ref{assump:u:1}. The regularity and sign properties of
\(v\), together with the HJB identity in \eqref{eq:ode:final}, verify
Theorem~\ref{thm:verification}(i). The choice
\(A>U_1'(r\underline{x})\), the boundary condition
\(\lim_{x\downarrow\underline{x}}v''(x)=-\infty\), and the growth bound
for \(v\) in \eqref{eq:property:v} verify
Theorem~\ref{thm:verification}(ii). Here, \(v\geq0\) follows from
\(v'>0\) and the positive boundary value used in its construction.
Moreover, the bounds on \(I(v')\) and
\(-v'/(xv'')\) in \eqref{eq:property:v} give the two linear-growth
conditions in Theorem~\ref{thm:verification}(iii), while
Theorem~\ref{thm:strong:case1} provides the required unique strong
solution satisfying \(X^{x;*}\geq\underline{x}\). Thus, all conditions
of Theorem~\ref{thm:verification} are satisfied.

Consequently, when \(\alpha x_0>b_0/r\), \(v=V\), and the feedback
functions \((b^*,\pi^*)\) defined in \eqref{eq:b^*} constitute an
optimal strategy. In particular, \(V\) is a classical solution of
\eqref{eq:hjb:simplify} and possesses all the properties stated in
Theorem~\ref{thm:v}.

The verification argument also yields uniqueness: any classical solution
of \eqref{eq:hjb:simplify} satisfying the conditions of
Theorem~\ref{thm:verification} must coincide with $V$. Since the
construction links the initial point to the boundary slope through
$y_0=[V'(\underline{x})]^{-1}$, the initial point giving rise to a
classical solution satisfying the conditions of
Theorem~\ref{thm:verification} is unique as well. This completes the characterization of the
value function and the optimal strategy in the nondegenerate boundary
case $\alpha x_0>b_0/r$. This case corresponds to Scenarios~1 and~2 in the numerical experiments
presented in Section~\ref{sec:numerical}.

\section{The Degenerate Boundary Case: $\alpha x_0\leq\frac{b_0}{r}$}\label{sec:degenerate}
After we completed the analysis of the nondegenerate boundary case $\alpha x_0>\frac{b_0}{r}$ in Section~\ref{sec:well:hjb}, we next turn to the degenerate boundary case $\alpha x_0\leq\frac{b_0}{r}$. As anticipated in Remark~\ref{rmk:cases}(ii), we take the boundary tangent to be infinite: $A:=\lim_{x\downarrow\underline{x}}v'(x)=\infty$. Hence, unlike in the nondegenerate case, no boundary parameter needs to be determined, and the auxiliary equation \eqref{eq:H:first} is posed on $(0,\infty)$ with a degenerate boundary at $h=0$. Specifically, the coefficient $-\frac12 \theta^2h^2$ of $H''$ vanishes at $h=0$, so writing \eqref{eq:H:first} in standard form produces singular coefficients and prevents a direct application of standard local ODE theory at the boundary.

\subsection{Assumptions and the Main Result in the Degenerate Case}\label{sec:case2}
 To address the degeneracy of \eqref{eq:H:first} at $h=0$, we impose the following additional condition on $U_2$.
\begin{assumption}\label{assump:holder:u2}
There exist constants
   $C_{2,2}>0$ and $\gamma_0\in(0,\gamma]$ such that
    \begin{align}\label{eq:holder:u2}
        |U_2(x)-U_2(y)|\leq C_{2,2}|x-y|^{\gamma_0} \quad\text{for all } \,x,y\in[\underline{x},\infty).
    \end{align}
\end{assumption}
Note that a sufficient condition for~\eqref{eq:holder:u2} is
\[
0\leq U_2'(x)\leq \gamma_0 C_{2,2}x^{\gamma_0-1}
\quad\text{for all }x\in[\underline{x},\infty).
\]

The following theorem gives the main result for the degenerate boundary case and its auxiliary function $H$.  
  \begin{theorem}\label{thm:Hnew:case2}
  Under Assumptions \ref{assump:parameter}, \ref{assmup:U12}  and \ref{assump:holder:u2},  suppose further that $\alpha x_0\le \frac{b_0}{r}$. Then the following ODE:
\begin{align}\label{eq:extension:u_222new}
\left\{
\begin{aligned}
		&\!-\!\frac{1}{2} \theta^2 h^2H''(h) \!-\!(r\!-\!\delta) hH'(h)\!-\!{U}_2'(H(h))h^2H'(h)\!+\!r H(h)\! -\!I\left( \frac1h\right)\!=\!0,~ h\in(0,\infty),\\
        &\limsup_{h\downarrow 0}|H(h)| \,\,\text{and}\,\, \limsup_{h\to \infty}\frac{|H\left(h\right)|}{h^{\frac{1}{1-\gamma}}} \text{ are finite},
        \end{aligned}
        \right.
\end{align}
 admits a  solution $H\in C^1([0,\infty))\cap C^2((0,\infty))$. Moreover, $H$ satisfies
 \begin{align}\label{eq:Hnew:boundary-asymptotics}
 H(0)=\frac{b_0}{r}\quad \text{and}\quad
 \lim_{h\downarrow0}\frac{H'(h)}{h^{p_1-1}}=l_H\in(0,\infty).
 \end{align}
 In addition,
 \begin{align}
     &(i)\quad H(h)\geq \frac{2}{\theta^2p_1(p_1+p_2)}I\left( \frac1h\right)
     \quad \text{for all } h\in(0,\infty);\label{eq:Hnew:lower-bound}\\
     &(ii)\quad 0<\frac{hH'(h)}{H(h)}\le p_1
     \quad \text{for all }h\in(0,\infty).\label{eq:Hnew:derivative-bound}
 \end{align}
 where $p_1$, $q_1=-p_2$ are the two roots of the quadratic equation \eqref{eq:p} such that $p_1 >\frac{1}{1-\gamma}$ and $p_2>0$.
\end{theorem}

Compared with Theorem~\ref{thm:H}, Theorem~\ref{thm:Hnew:case2} now
concerns a global problem with the degenerate point at $h=0$, where
standard initial-value arguments are unavailable. To address this
difficulty, we impose the endpoint growth conditions
in~\eqref{eq:extension:u_222new}, construct a global solution through
the associated kernel representation in
Proposition~\ref{prop:Hnew:existence}, and then recover its precise
boundary behavior in~\eqref{eq:Hnew:boundary-asymptotics} through
Proposition~\ref{prop:Hnew:boundary}. In particular, the boundary
behavior of $H'$ in Theorem~\ref{thm:H} follows directly by evaluating
the ODE at the nondegenerate initial point $y_0>0$, whereas at $h=0$,
the vanishing coefficient of $H''$ prevents the same argument from
being applied. The second relation
in~\eqref{eq:Hnew:boundary-asymptotics} must therefore be established
through a refined asymptotic analysis. This relation is crucial for
establishing the global monotonicity in
Proposition~\ref{prop:Hnew:global} and reconstructing the value
function in Theorem~\ref{thm:vnew:case2}.

In the rest of this section, we always assume that Assumptions \ref{assump:parameter}, \ref{assmup:U12} and \ref{assump:holder:u2} hold.

\subsection{Global Existence for the Auxiliary Equation}

As at the beginning of Section \ref{sec:case1:existence}, we first extend $U_2$ to the entire real line.
Under Assumptions~\ref{assmup:U12} and
\ref{assump:holder:u2}, there exists an extension
$\tilde U_2\in C^1(\mathbb R)$ satisfying
\begin{align}
&(i)\quad 
\tilde U_2=U_2 \ \text{on }[\underline{x},\infty),\qquad
\tilde U_2\geq0,\quad \tilde U_2'\geq0 \ \text{on }\mathbb R,
\qquad \tilde U_2'=0 \ \text{on }(-\infty,0];\nonumber\\
&(ii)\quad
|\tilde U_2(x)-\tilde U_2(y)|
\leq C_{2,2}|x-y|^{\gamma_0},\qquad x,y\in\mathbb R.
\label{eq:tildeU2:properties:case2}
\end{align}
Here, $C_{2,2}$ is enlarged if necessary. Unlike in
Section~\ref{sec:well:hjb}, we do not require $\tilde U_2'$ to be locally
Lipschitz continuous; instead, we require $\tilde U_2$ to be globally
H\"older continuous of order $\gamma_0$. 

The following proposition constructs a global solution of the extended equation without  imposing the range condition $H\geq\underline{x}=b_0/r$ yet, and provides the kernel representation in terms of the fundamental solutions and boundary values needed below. Its proof uses the integral operator suggested by applying variation of parameters to the Cauchy--Euler part: $-\frac{1}{2}\theta^2h^2H''(h)-(r-\delta)hH'(h)+rH(h)$ of the ODE \eqref{eq:Hnew:case2:extended}  and applies the Schauder fixed point theorem to a weighted space defined in \eqref{eq:norm:weight} of continuous functions.

\begin{proposition}\label{prop:Hnew:existence}
The following extended problem associated with ODE~\eqref{eq:extension:u_222new} admits a solution
$H\in C^1([0,\infty))\cap C^2((0,\infty))$:
\begin{align}\label{eq:Hnew:case2:extended}
\left\{
\begin{aligned}
&-\frac{1}{2}\theta^2h^2H''(h)-(r-\delta)hH'(h)
-\tilde{U}_2'(H(h))h^2H'(h)+rH(h)-\tilde{I}(h)=0,\quad h>0,\\
&\limsup_{h\downarrow0}|H(h)|\,\,\text{and}\,\,
\limsup_{h\to\infty}\frac{|H(h)|}{h^{\frac{1}{1-\gamma}}} \text{ are finite}.
\end{aligned}
\right.
\end{align}
Moreover, $H$ has the following integral representation:
\begin{align}\label{eq:Hdef}
H(h) =\ & \frac{2}{\theta^2(p_1+p_2)}\bigg\{\int_h^\infty\left[\tilde{I}(z)+(p_1-1)\tilde{U}_2(H(z))z\right]z^{-p_1-1}\mathrm{d}z \, h^{p_1}\nonumber\\
& \quad +\int_0^h\left[\tilde{I}(z)-(p_2+1)\tilde{U}_2(H(z))z\right]z^{p_2-1}\mathrm{d}z \, h^{-p_2}\bigg\},\quad h>0,
\end{align}
which satisfies $
H(0)=\frac{b_0}{r}$ and $ H'(0)=0$.
\end{proposition}
Note that compared with \eqref{eq:extension:u_222new}, \eqref{eq:Hnew:case2:extended} replaces $U_2$ by its extension $\tilde U_2$; see Section \ref{sec:case1:existence} for the same treatment for the nondegenerate boundary case $\alpha x_0>b_0/r$.
 The proof of Proposition~\ref{prop:Hnew:existence} is provided in Appendix~\ref{appendix:case2:schauder}.

We henceforth fix the auxiliary function $H$ of the extended problem provided by
Proposition~\ref{prop:Hnew:existence}, for which $H(0)=b_0/r$ and
$H'(0)=0$. To recover a solution of the original ODE and prepare for the
subsequent global analysis, the following proposition shows that
$H(h)\geq H(0)=b_0/r$ for all $h\geq0$ and determines the precise
asymptotic behavior of $H'$ stated in the second relation
of~\eqref{eq:Hnew:boundary-asymptotics}.

\begin{proposition}\label{prop:Hnew:boundary}
The auxiliary function $H$ has the following properties:
\begin{enumerate}
\item[(i)] $H(h)\ge H(0)=b_0/r$ for all $h\ge0$;
\item[(ii)] There exists $l_H\in(0,\infty)$ such that
\(
\lim_{h\downarrow0}\frac{H'(h)}{h^{p_1-1}}=l_H.
\)
\end{enumerate}
\end{proposition}

\begin{proof}
\textit{Part~\textup{(i)}.}
 Suppose, to the contrary, that there exists $h_0\in(0,\infty)$ such that $H(h_0)<b_0/r$. We first claim that
\[
H'(h)\leq0,\qquad h>h_0.
\]
If not, there exists $h_1>h_0$ such that $H'(h_1)>0$. The minimum of $H$ on $[0,h_1]$ cannot be attained at either endpoint and must therefore be attained at some $h_m\in(0,h_1)$. Since $h_m$ is a minimum point, $H(h_m)\leq H(h_0)<b_0/r$.
On the other hand, $H'(h_m)=0$ and $H''(h_m)\geq0$. Evaluating ODE~\eqref{eq:Hnew:case2:extended} at $h_m$ gives
\[
rH(h_m)=\tilde I(h_m)+\frac{1}{2}\theta^2h_m^2H''(h_m)
\geq\tilde I(h_m)\geq b_0,
\]
where the last inequality uses the fact $\tilde I\geq b_0$ by the definition of $\tilde{I}$ in \eqref{eq:tildeI}.
Hence, $H(h_m)\geq b_0/r$, contradicting the immediately  strict inequality up above. This proves the claim. The tail argument in the proof of Proposition~\ref{lma:H'}(i), applied here with \eqref{eq:Hdef} in place of \eqref{eq:H:G:representation}, rules out this claim. Therefore, $H(h)\geq b_0/r$ for all $h\geq0$.

\medskip
\noindent\textit{Part~\textup{(ii)}.}
 By part~\textup{(i)}, $H(h)\geq H(0)=b_0/r$, and hence
$\tilde U_2(H(h))=U_2(H(h))$ for all $h\geq0$. Differentiating
\eqref{eq:Hdef} and telescoping with the term $U_2(H(0))$, we obtain
  \begin{align}\label{eq:h':new}
        H'(h)&=\frac{2}{\theta^2(p_1+p_2)}\bigg\{p_1\int_h^\infty\left[\tilde{I}(z)+(p_1-1)\left({U}_2(H(z))-U_2(H(0))\right)z\right]z^{-p_1-1}\mathrm{d}z \, h^{p_1-1}\nonumber\\
   &\quad\quad\quad\quad\quad\quad\quad -p_2\int_0^h\left[\tilde{I}(z)-(p_2+1)\left({U}_2(H(z))-U_2(H(0))\right)z\right]z^{p_2-1}\mathrm{d}z \, h^{-p_2-1}\nonumber\\
   &\quad\quad\quad\quad\quad\quad\quad-(p_1+p_2)\left({U}_2(H(h))-{U}_2(H(0))\right)\bigg\}.
\end{align}
By combining the two $\tilde I$-terms in \eqref{eq:h':new} and using $\tilde I(z)=b_0$ for $z\in(0,z_0)$, a direct calculation yields
\begin{align}\label{eq:lim:i}
    \lim_{h\downarrow 0}\frac{p_1\int_h^\infty\tilde{I}(z)z^{-p_1-1}\mathrm{d}z \, h^{p_1-1} \!-\!p_2\int_0^h\tilde{I}(z)z^{p_2-1}\mathrm{d}z \, h^{-p_2-1}}{h^{p_1-1}}\!=\!p_1\int_{z_0}^\infty \tilde{I}(z) z^{-p_1-1}\mathrm{d}z\!-\!b_0z_0^{-p_1}.
\end{align}
For the remaining terms, let $k\geq1$ be the unique integer such that $k<p_1\leq k+1$. We first proceed by induction to show that
\begin{align}\label{eq:induction:l}
    \lim_{h\downarrow0}\frac{H'(h)}{h^{m}}=0 ,
\end{align}
for any $ m=0,\cdots, k-1$.
First, because $H'(0)=0$, \eqref{eq:induction:l} holds for $m=0$. Suppose that \eqref{eq:induction:l} holds for some $0\leq m=l\le k-2$. We show that \eqref{eq:induction:l}  also holds for   $m=l+1$.  
Applying l'H\^opital's rule and using
$\lim\limits_{h\downarrow0}\frac{H'(h)}{h^{l}}=0$, we obtain
\begin{align}\label{eq:limit:part1}
    &\lim_{h\downarrow 0} \frac{{U}_2(H(h))-{U}_2(H(0))}{h^{l+1}}=\lim_{h\downarrow 0}\frac{U_2'(H(h))H'(h)}{(l+1)h^{l}}=0,\\
    &\lim_{h\downarrow 0} \frac{\int_0^h\left({U}_2(H(z))-U_2(H(0))\right)z^{p_2}\mathrm{d}z \, h^{-p_2-1}}{h^{l+1}}
    =\lim_{h\downarrow 0} \frac{\int_0^h \left({U}_2(H(z))-U_2(H(0))\right)z^{p_2}\mathrm{d}z}{h^{l+p_2+2}}\nonumber\\
    =&\lim_{h\downarrow 0} \frac{\left({U}_2(H(h))-U_2(H(0))\right)h^{p_2}}{(l+p_2+2)h^{l+p_2+1}}
    =\lim_{h\downarrow 0} \frac{{U}_2(H(h))-U_2(H(0))}{(l+p_2+2)h^{l+1}}
    =0,\label{eq:limit:part2}\\
     &\lim_{h\downarrow 0} \frac{\int_h^\infty\left({U}_2(H(z))-U_2(H(0))\right)z^{-p_1}\mathrm{d}z \, h^{p_1-1}}{h^{l+1}}
    =\lim_{h\downarrow 0} \frac{\int_h^\infty \left({U}_2(H(z))-U_2(H(0))\right)z^{-p_1}\mathrm{d}z}{h^{l+2-p_1}}\nonumber\\
    =&\lim_{h\downarrow 0} \frac{\left({U}_2(H(h))-U_2(H(0))\right)h^{-p_1}}{(p_1-l-2)h^{l+1-p_1}}
    =\lim_{h\downarrow 0} \frac{{U}_2(H(h))-U_2(H(0))}{(p_1-l-2)h^{l+1}}
    =0,\label{eq:limit:part3}
\end{align}
where \eqref{eq:limit:part3} follows from the fact that $l+2-p_1\leq k-p_1<0$.
Substituting \eqref{eq:lim:i},
\eqref{eq:limit:part1}--\eqref{eq:limit:part3} into \eqref{eq:h':new}, and using
$l+1<p_1-1$, we obtain
\[
\lim_{h\downarrow0}\frac{H'(h)}{h^{l+1}}=0.
\]
 Therefore, \eqref{eq:induction:l} holds for $m=l+1$. By induction, we conclude that  $\lim\limits_{h\downarrow0}\frac{H'(h)}{h^{k-1}}=0$. 

 We next extend \eqref{eq:induction:l} beyond the integer exponents.
Fix $\varepsilon\in(0,1)$. Since
$p_1-3+\varepsilon<p_1-2\leq k-1$, the calculations leading to
\eqref{eq:limit:part1}--\eqref{eq:limit:part3} remain valid with
$l=p_1-2$ for the first two limits and
$l=p_1-3+\varepsilon$ for the third. Hence,
\begin{align*}
&\lim_{h\downarrow0}
\frac{U_2(H(h))-U_2(H(0))}{h^{p_1-1}}=0,\quad \lim_{h\downarrow0}
\frac{\displaystyle
\int_0^h\bigl(U_2(H(z))-U_2(H(0))\bigr)z^{p_2}\mathrm{d}z\,
h^{-p_2-1}}
{h^{p_1-1}}=0,\\
&\lim_{h\downarrow0}
\frac{\displaystyle
\int_h^\infty\bigl(U_2(H(z))-U_2(H(0))\bigr)z^{-p_1}\mathrm{d}z\,
h^{p_1-1}}
{h^{p_1-2+\varepsilon}}=0.
\end{align*}
 Since $\varepsilon<1$, substituting these three estimates into \eqref{eq:h':new} yields
$\lim_{h\downarrow0}\frac{H'(h)}{h^{p_1-2+\varepsilon}}=0$.
Applying \eqref{eq:limit:part1} with $l=p_1-2+\varepsilon$ then gives
\begin{align}\label{eq:limit:u2:epsilon}
\lim_{h\downarrow0}
\frac{U_2(H(h))-U_2(H(0))}{h^{p_1-1+\varepsilon}}=0.
\end{align}
Consequently, \eqref{eq:h':new} yields
\begin{align*} 
    \lim_{h\downarrow 0} \frac{H'(h)}{h^{p_1-1}}=~&\frac{2}{\theta^2(p_1+p_2)}\left[p_1\int_{z_0}^\infty \tilde{I}(z) z^{-p_1-1}\mathrm{d}z-b_0z_0^{-p_1}\right.\\
   &\left.+p_1(p_1-1)\int_0^\infty(U_2(H(z))-U_2(H(0)))z^{-p_1}\md z\right]=:l_H\in(0,\infty).
\end{align*}
Here, the finiteness of $l_H$ follows from the fact that 
 $\bigl(U_2(H(z))-U_2(H(0))\bigr)z^{-p_1}
=o(z^{-1+\varepsilon})\quad\text{as }z\downarrow0$ by \eqref{eq:limit:u2:epsilon}, while
$l_H>0$ follows from $\tilde I(z)>b_0$ for all $z>z_0$ and
$H(z)\geq b_0/r=H(0)$, together with the monotonicity of $U_2$.
\end{proof}

 By Proposition~\ref{prop:Hnew:boundary}(i),
$H(h)\geq b_0/r$ for all $h\geq0$. Hence,
$\tilde U_2'(H(h))=U_2'(H(h))$, and the solution $H$ of the extended
problem~\eqref{eq:Hnew:case2:extended} is therefore a solution of the
original ODE~\eqref{eq:extension:u_222new}. It remains only to establish
the global estimates~\eqref{eq:Hnew:lower-bound} and
\eqref{eq:Hnew:derivative-bound}, which is the purpose of the following
proposition.

\begin{proposition}\label{prop:Hnew:global}
The auxiliary function $H$ satisfies
\begin{align*}
\textup{(i)}\quad&H'(h)>0;
\qquad h\in(0,\infty);\\
\textup{(ii)}\quad&H(h)\geq
\frac{2}{\theta^2p_1(p_1+p_2)}I\left(\frac1h\right),
\qquad h\in(0,\infty);\\
\textup{(iii)}\quad&\frac{hH'(h)}{H(h)}\leq p_1,
\qquad h\in(0,\infty).
\end{align*}
\end{proposition}

\begin{proof}
\textit{Part~\textup{(i)}.} By
Proposition~\ref{prop:Hnew:boundary}(ii), $H'(h)>0$ for all
sufficiently small $h>0$. Since $\tilde I(h)=b_0$ near $h=0$ and
$H(0)=b_0/r$, we also have $rH(h)>\tilde I(h)$ there. The sign-change
argument in the proof of Proposition~\ref{lma:phi}, applied here to
\eqref{eq:Hnew:case2:extended},  then shows that
either $H'(h)>0$ for all $h>0$, or $H'$ becomes negative and remains
negative thereafter. The latter possibility is ruled out by the tail
argument in the proof of Proposition~\ref{lma:H'}(i), applied with
\eqref{eq:Hdef} in place of \eqref{eq:H:G:representation}. Therefore,
$H'(h)>0$ for all $h>0$. 

\medskip
\noindent\textit{Parts~\textup{(ii)}--\textup{(iii)}.}
The arguments used to prove
Proposition~\ref{lem:limsup}(ii)--(iii) then yield the remaining
estimates directly from \eqref{eq:Hdef}.

\end{proof}

Overall, Propositions~\ref{prop:Hnew:existence}, \ref{prop:Hnew:boundary}, and~\ref{prop:Hnew:global} together establish that the resulting auxiliary function $H$ solves \eqref{eq:extension:u_222new} and satisfies all the properties stated in Theorem~\ref{thm:Hnew:case2}. Therefore, the proof of Theorem~\ref{thm:Hnew:case2} is complete.

\subsection{Reconstruction and Verification of the Value Function}
We now invert the transform (Eqs.~\eqref{eq:f}--\eqref{eq:H:first}) introduced in Section~\ref{sec:transform}. The following theorem establishes the regularity, boundary behavior, and growth properties required in Theorem \ref{thm:verification2}.

\begin{theorem}\label{thm:vnew:case2}
   Under Assumptions \ref{assump:parameter}, \ref{assmup:U12}, and \ref{assump:holder:u2}, suppose further that $\alpha x_0\leq \frac{b_0}{r}$. Then the following ODE admits a solution in $C([\underline{x},\infty))\cap C^3((\underline{x},\infty))$:
    \begin{align}\label{eq:ode:final:case2}
    \left\{\begin{aligned}
		&\delta \cdot v(x)=D(v'(x))+ U_2(x)
		-\frac{1}{2}\theta^2\frac{\left(v'(x)\right)^2}{v''(x)}+rxv'(x),  \quad x\in(\underline{x},\infty),\\
        &v(\underline{x})=V(\underline{x}),\quad\lim\limits_{x\downarrow \underline{x}}v'(x)=+\infty,\quad \lim\limits_{x\downarrow \underline{x}}v''(x)=-\infty.
        \end{aligned}
        \right.
	\end{align}
     Moreover, there exist constants $C_{v,6}>0$, $C_{v,7}>0$,
    and $C_{v,8}>0$ such that
    \begin{equation}\label{eq:property:v:case2}
    \begin{aligned}
    &(i)\quad\lim_{x\downarrow \underline{x}}
    \frac{v'(x)}{(x-\underline{x})^{-\frac{1}{p_1}}}
    =\left(\frac{l_H}{p_1}\right)^{\frac{1}{p_1}},\qquad
    \limsup_{x\to\infty}\frac{v'(x)}{x^{\gamma-1}}<\infty;\\
    &(ii)\quad v(x)\leq C_{v,6}x^\gamma
    \quad\text{for all }\,\, x\geq\underline{x};\\
    &(iii)\quad v'(x)>0\quad\text{and}\quad
    I(v'(x))\leq C_{v,7}x
    \quad \text{for all }\,\, x>\underline{x};\\
    &(iv)\quad v''(x)<0\quad\text{and}\quad 
    -\frac{v'(x)}{xv''(x)}\leq C_{v,8} 
    \quad\text{for all } \,\,x>\underline{x}.
    \end{aligned}
    \end{equation}
\end{theorem}
 The proof of Theorem~\ref{thm:vnew:case2} is provided in Appendix~\ref{appendix:sec:case2:2}.

 \begin{remark}\label{rmk:case2:properties}
The properties in \eqref{eq:property:v:case2} are consistent with those in \eqref{eq:property:v}, except for the first relation, which gives the boundary asymptotic
$v'(x)\sim (l_H/p_1)^{1/p_1}(x-\underline{x})^{-1/p_1}$ as $x\downarrow\underline{x}$. Thus, unlike the finite boundary slope $v'(\underline{x})=A$ in the nondegenerate case: $\alpha x_0>\frac{b_0}{r}$, $v'$ diverges at the lower boundary in the degenerate case. Moreover, in the
unconstrained CRRA benchmark of Remark~\ref{rmk:unconstraint},
$V_{un}'(x)=\bar{b}x^{\gamma-1}$ and hence is of order
$x^{\gamma-1}$ as $x\downarrow0$. Since
$-1/p_1>\gamma-1$, the boundary singularity of $v'$ at
$\underline{x}$ is weaker than that of the CRRA benchmark at zero.
\end{remark}

The growth and regularity properties established in Theorem \ref{thm:vnew:case2} ensure that the candidate feedback functions $(b^*,\pi^*)$ in \eqref{eq:b^*} are well-defined. In the degenerate boundary case, the diffusion and drift coefficients of the candidate state equation are locally Lipschitz. As in the nondegenerate case, a truncation argument and the linear growth estimates yield a unique global strong solution.

\begin{theorem}\label{thm:strong:case2}
 Under Assumptions \ref{assump:parameter}, \ref{assmup:U12} and \ref{assump:holder:u2},  suppose further that $\alpha x_0\le \frac{b_0}{r}$, let $v$ be the function obtained in Theorem \ref{thm:vnew:case2}. Then, for all $x\geq \underline{x}$, SDE \eqref{eq:verify:wealth} admits a unique strong solution satisfying $X^{x;*}\geq \underline{x}$.
\end{theorem}
The proof of Theorem \ref{thm:strong:case2}  is provided in Appendix \ref{appendix:sec:case2:3}.

Similarly, combining Theorem \ref{thm:vnew:case2} with Theorem \ref{thm:strong:case2}, the verification argument applies, and $v$ is identified as the value function in the case $\alpha x_0\le \frac{b_0}{r}$. This case corresponds to Scenario~3 in the numerical experiments
presented in Section~\ref{sec:numerical}.

\section{Numerical Experiments}\label{sec:numerical}
In this section, we present numerical experiments to validate the theoretical results derived above and to analyze the impact of key parameters on the optimal strategy. The parameter assignments in the model are based on reasonable economic assumptions and financial practices.

First, for the financial-market parameters, we use average interest rates observed in July 2026. We set the standard risk-free interest rate to \(r_0=0.037\ (3.7\%)\), based on the average continuously compounded yield on 4-week Treasury bills, representing a short-term and liquid investment available to individual investors. For the fund's long-term riskless return, we use the Federal Reserve's estimated 30-year zero-coupon Treasury yield and set \(r_1=0.053\ (5.3\%)\), reflecting the fund's longer investment horizon and lower liquidity needs; see \cite*{us_treasury_daily_bill_rates_2026} and \cite*{fed_nominal_yield_curve_2026}. The running management fee is set to \(\kappa_1=0.004\ (0.4\%)\), based on the expense-ratio data reported by the Investment Company Institute \cite*{ici_expense_ratios}. Consequently, the retiree's net interest rate within the fund is \(r=r_1-\kappa_1=0.049\ (4.9\%)\), which remains above \(r_0\). We set the terminal management fee to \(\kappa_2=0\), consistent with the absence of a redemption fee for the Vanguard Target Retirement Income Fund \cite*{vanguard_vtinx_prospectus_2026}.
For the risky asset, we use unhedged global equities as a diversified investment benchmark. Rounding the figures reported in PGIM's Q1 2026 ten-year capital-market assumptions, we set the expected return and volatility to \(\mu= 0.08\ (8\%)\) and \(\sigma=0.2\ (20\%)\), respectively \cite*[p.~12]{pgim_cma_2026}.

Second, for retiree characteristics, we calibrate the constant mortality rate $\lambda$ by matching the expected remaining lifetime $1/\lambda$ to the U.S. life-table estimate of $23.6$ years at age 60 \cite*[Table~A]{nchs_life_expectancy}, giving $\lambda=1/23.6\approx0.042$; see also Figure~\ref{fig:makeham-time-constant-six} in Section~\ref{sec:conclusion} for a comparison with the Makeham law.
The common discount rate is set to \(\rho=0.04\), corresponding to an annual discount factor of approximately \(0.96\), consistent with structural estimates in the life-cycle literature \cite*{gourinchas_parker_2002,laibson_discount_functions_2026}.

Third, we assume that the retiree and the fund manager have the following utility functions:
\[
U_{c,1}(x) = U_{c,2}(x) = \frac{x^{\gamma_1}}{\gamma_1}, \quad 
U_{m,1}(x) = \frac{x^{\gamma_2}}{\gamma_2},
\]
where $\gamma_1 \le \gamma_2$, so that the retiree is weakly more risk-averse than the fund manager. Under this specification, $1-\gamma_i$ is the coefficient of relative risk aversion. Motivated by empirical evidence on moderate risk aversion \cite*{andersen_risk_time_2008,barseghyan_risk_preferences_2018}, we set $(\gamma_1,\gamma_2)=(0.2,0.2)$ as a baseline. For a wider robustness comparison, Figure~\ref{fig:gamma} considers $(\gamma_1,\gamma_2)=(0.1,0.3)$, $(0.1,0.1)$, and $(0.3,0.3)$. Moreover, since $\kappa_2=0$, the terminal-fee utility contributes only an additive constant to the objective, so we normalize $U_{m,2}\equiv0$ without affecting the optimal strategy.

We measure all monetary quantities in units of USD $100{,}000$, set $x_0=1$, and define the numerical utility functions on this scale. When $\gamma_1=\gamma_2$, this normalization is without loss of generality by homogeneity; see Appendix~\ref{sec:Homogeneity}. When the exponents differ, the monetary unit fixes the relative cardinal scaling, and hence the implicit weight, of the retiree's and the manager's utilities. We therefore interpret the heterogeneous-$\gamma$ comparisons conditional on this fixed normalization.

Finally, the key constraint parameters are set in three scenarios:
\begin{enumerate}
    \item Scenario 1: $(b_0, \alpha) = (0.035, 0.95)$, resulting in an effective lower bound on wealth $\underline{x} = \max\{\alpha x_0, b_0/r\}=0.95$,
    corresponding to Case I ($\alpha x_0 > \frac{b_0}{r}$).
    
    \item Scenario 2: $(b_0, \alpha) = (0.035, 0.8)$, yielding $\underline{x}=0.8$,
    also corresponding to Case I.
    \item Scenario 3: $(b_0, \alpha) = (0.04, 0.8)$, giving $\underline{x}=b_0/r\approx0.816$,
    corresponding to Case II ($\alpha x_0 < \frac{b_0}{r}$).
\end{enumerate}
The effective lower bound $\underline{x}$ combines the prescribed capital floor $\alpha x_0$ with the minimum capital $b_0/r$ required to sustain the benefit guarantee, thereby enforcing solvency and limiting excessive risk-taking.
 Based on these parameter assignments, we can solve for the optimal strategies, the manager's total utility, and obtain the 5th and 20th percentile paths of fund wealth across 10,000 simulated sample paths.

\begin{figure}[h]
    \centering
    \begin{subfigure}[b]{0.49\textwidth}
        \includegraphics[width=\textwidth]{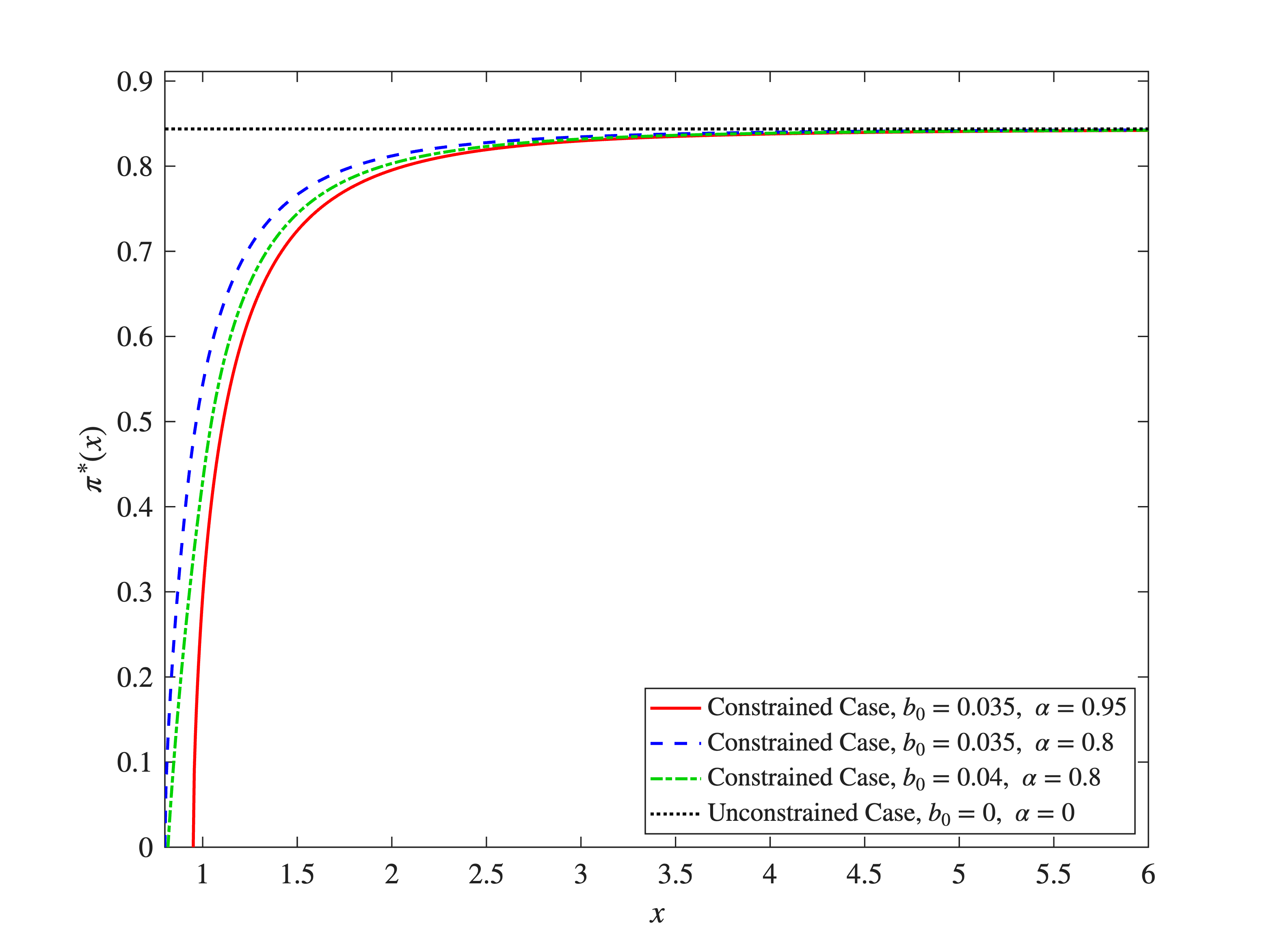}
        \caption{Optimal investment strategy}
    \end{subfigure}
    \hfill 
    \begin{subfigure}[b]{0.49\textwidth}
        \includegraphics[width=\textwidth]{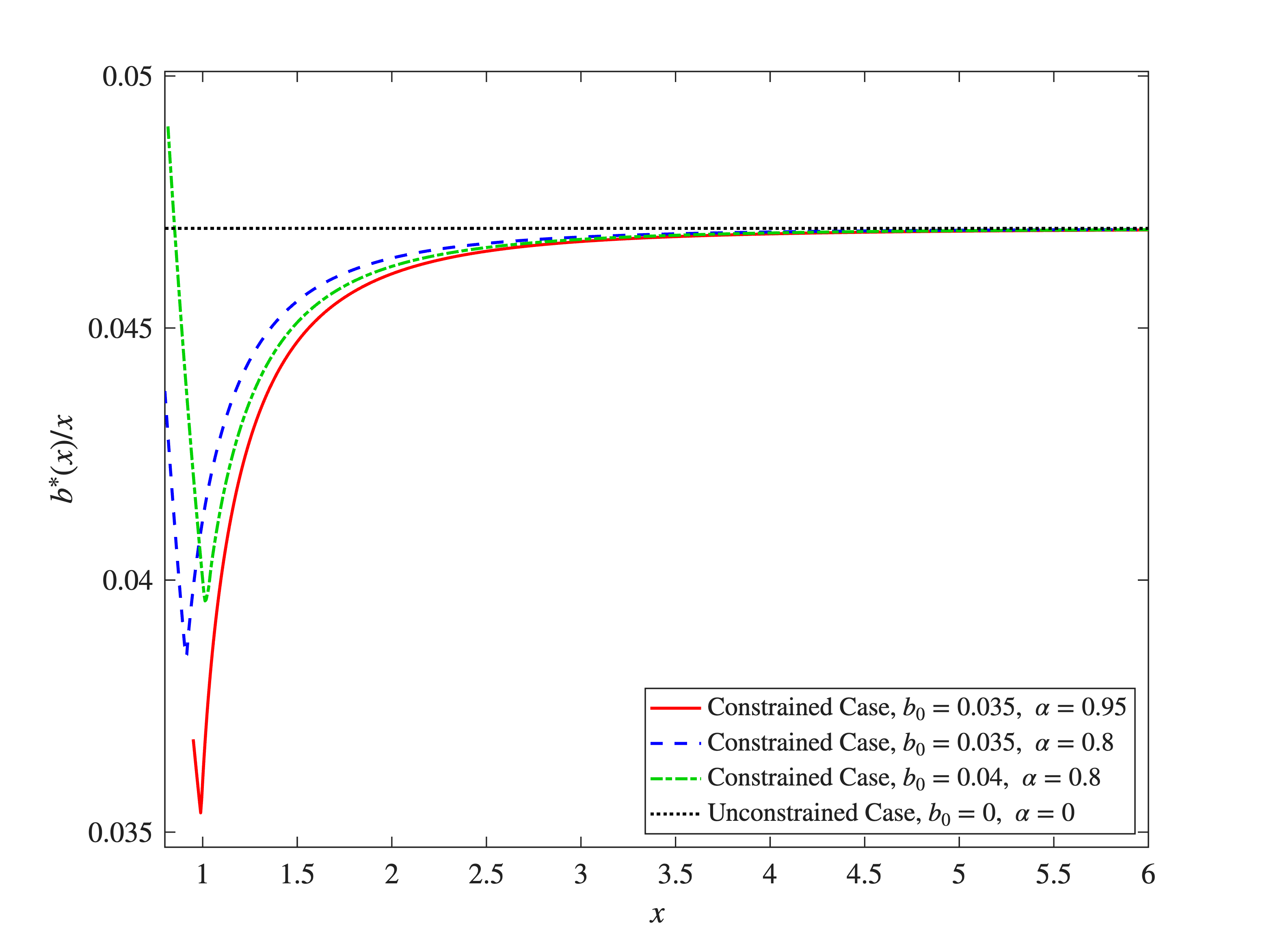}
          \caption{Optimal benefit-to-wealth strategy}
    \end{subfigure}
     \caption{ Optimal strategies. Panel (a) displays the optimal investment strategy, while panel (b) presents the optimal benefit-to-wealth ratio. In both panels, the solid, dashed, dash-dotted, and dotted lines correspond to the parameter sets $(b_0,\alpha)=(0.035,0.95)$, $(0.035,0.8)$, $(0.04,0.8)$, and $(0,0)$ (the unconstrained case, given in \eqref{eq:unconstraint:bpi} in Remark \ref{rmk:unconstraint}), respectively. The corresponding effective lower bounds on wealth, defined by $\underline{x}=\max\{\alpha x_0,b_0/r\}$, are $0.95$, $0.8$, and approximately $0.816$ for the first three cases, respectively, and the curves start from these values accordingly. The utility parameters are set to the base case $(\gamma_1,\gamma_2)=(0.2,0.2)$.}
       \label{fig:investment_benefit}
\end{figure}
Figure \ref{fig:investment_benefit} illustrates the optimal investment strategy and the optimal benefit-to-wealth ratio under different guarantee levels. For the optimal investment strategy in panel (a), all three constrained cases yield lower investment proportions than that under the unconstrained framework, reflecting the risk-limiting effect of the guarantee. Comparing $(b_0,\alpha)=(0.035,0.95)$ and $(0.035,0.8)$ (i.e., the solid and dashed lines), the solid line lies below the dashed line. This is because a larger $\alpha$ imposes a tighter lower bound on wealth, forcing the fund manager to adopt a more conservative investment strategy. Comparing $(b_0,\alpha)=(0.035,0.8)$ and $(0.04,0.8)$ (i.e., the dashed and dash-dotted lines), the dash-dotted line lies below the dashed line. This is due to the higher minimum benefit requirement $b_0$, which similarly tightens the constraint and prevents the manager from taking excessive risk. Moreover, in each constrained case, the optimal investment proportion is increasing in wealth, and as wealth becomes sufficiently large, it gradually converges to that under the unconstrained framework. This is because the influence of the guarantee weakens as wealth increases, making the constraint effectively irrelevant when wealth is high enough, so the problem becomes approximately unconstrained.

Panel (b) shows that the benefit-to-wealth ratio in each constrained case first decreases and then increases with wealth. The initial decline results from the minimum benefit requirement $b_0$. When wealth is near the lower bound $\underline{x}$, the optimal benefit rate $b^*(x)$ starts from $b_0$ and cannot increase too quickly, as the fund manager needs to prioritize wealth accumulation. In this region, $b^*(x)$ grows more slowly than $x$, causing the ratio $b^*(x)/x$ to exhibit a decreasing trend initially. As wealth continues to grow, the ratio starts to increase, and once wealth exceeds approximately $x=1.02$, the benefit-to-wealth ratios in all constrained cases are above the standard risk-free rate $r_0=3.7\%$, indicating that retirees can receive generous benefits when the pension fund wealth becomes relatively large.  Consistent with panel (a), the benefit-to-wealth ratio decreases with $\alpha$ when $b_0$ is fixed, and with $b_0$ when $\alpha$ is fixed and wealth exceeds a threshold near $x_0=1$. Below this threshold, however, the ordering with respect to $b_0$ is reversed: the direct effect of the higher minimum benefit requirement dominates, resulting in a higher ratio.

Figure \ref{fig:value_1} compares the value functions under different guarantee levels. The unconstrained value function dominates all constrained cases. 
Moreover, the value function is nonincreasing in either $\alpha$ or $b_0$ when the other is held fixed, as tighter guarantees shrink the admissible strategy set. As wealth increases, the differences among the value functions gradually diminish, indicating that the impact of the guarantee constraints becomes negligible when the fund wealth is sufficiently large. At the initial wealth level $x=x_0=1$, the relative differences are approximately $1.12\%$, $0.66\%$, and $0.89\%$ for the three constrained cases, respectively. These values can be interpreted as the utility cost of imposing the guarantee constraints.

\begin{figure}[h]
    \centering
    \begin{subfigure}[b]{0.49\textwidth}
        \includegraphics[width=\textwidth]{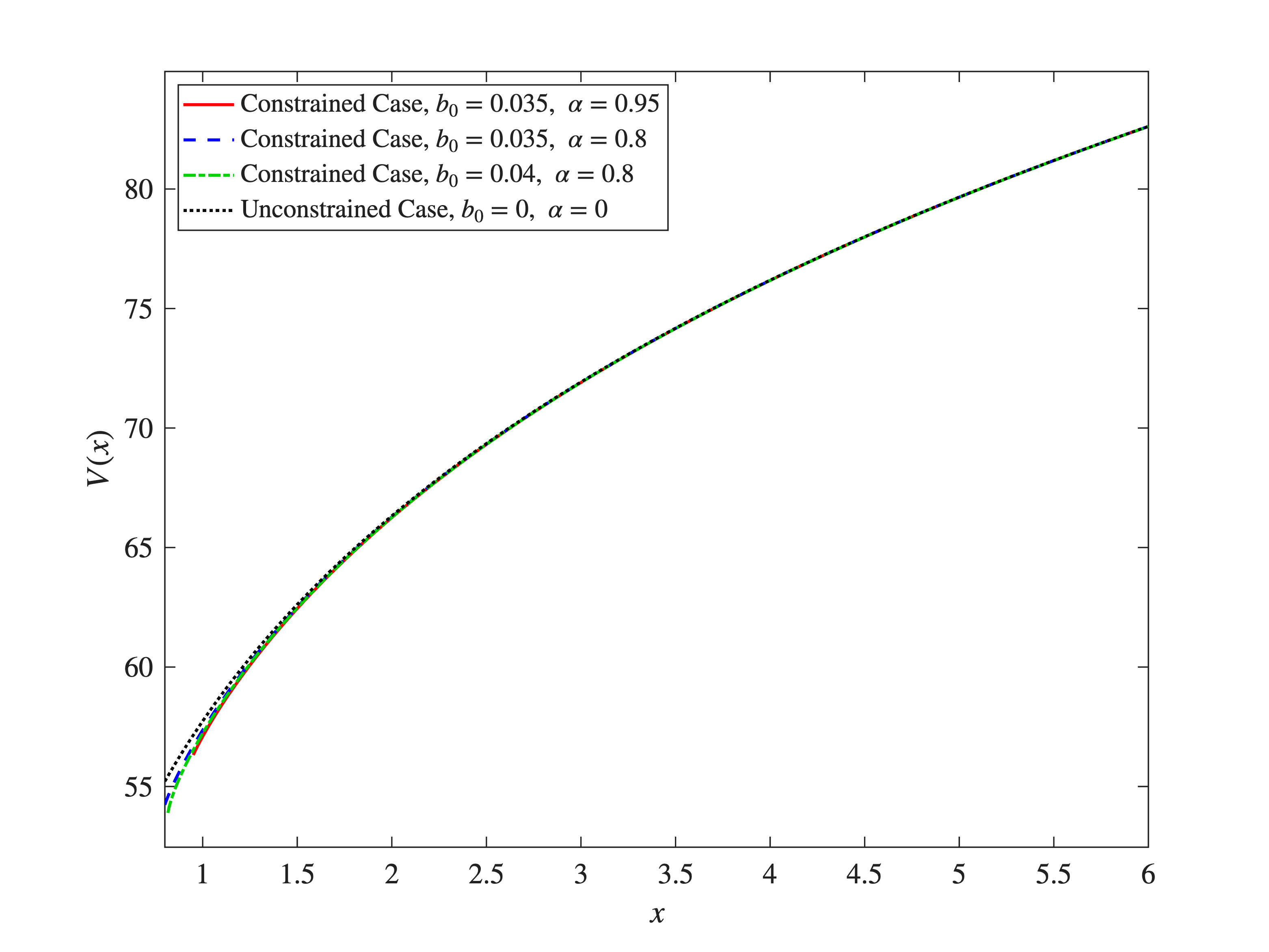}
        \caption{Value function}
    \end{subfigure}
    \hfill 
    \begin{subfigure}[b]{0.49\textwidth}
        \includegraphics[width=\textwidth]{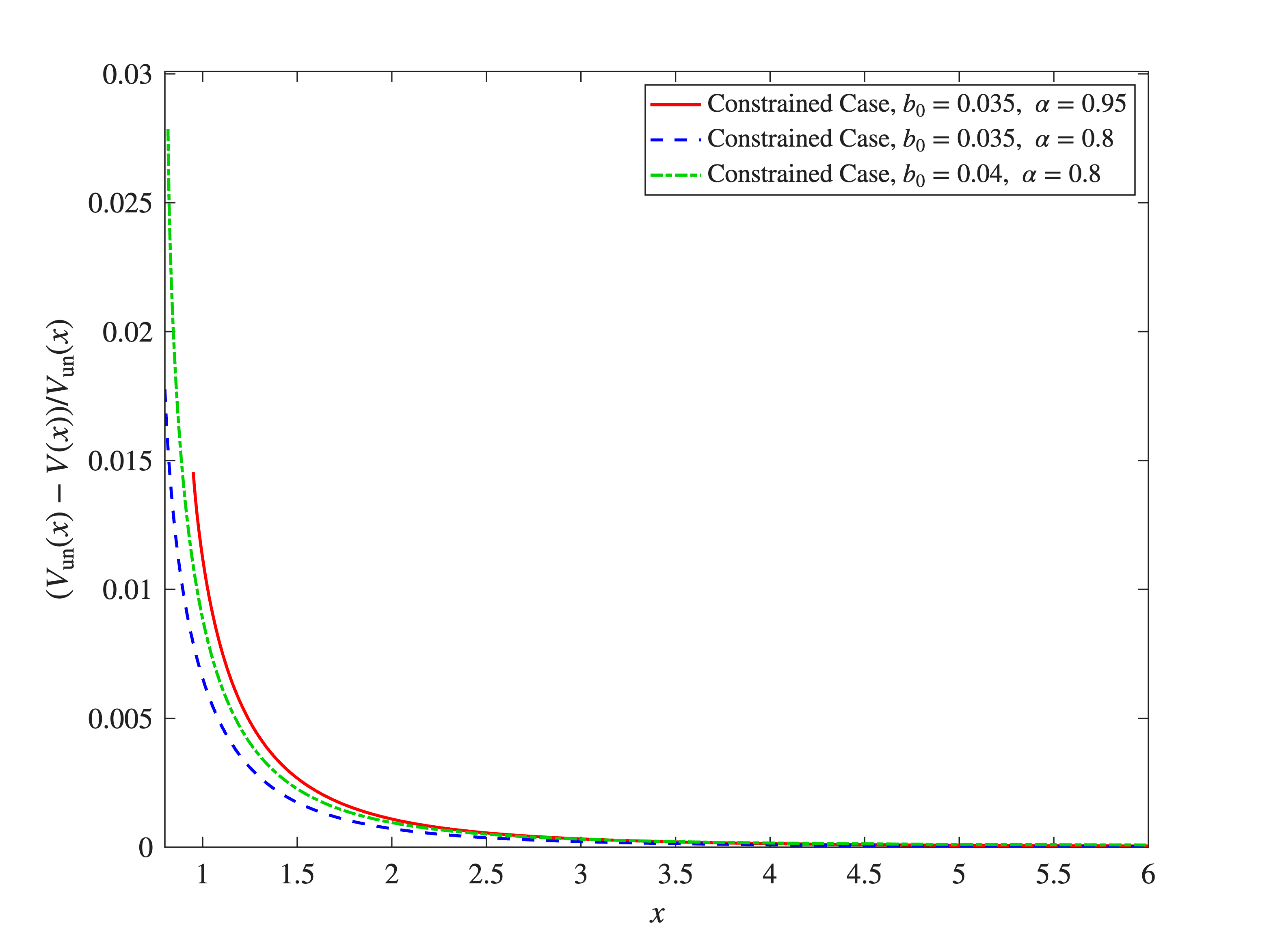}
         \caption{Relative difference}
    \end{subfigure}
        \caption{Comparison of value functions. Panel (a) displays the value functions under different parameter settings. The solid, dashed, dash-dotted, and dotted lines correspond to $(b_0,\alpha)=(0.035,0.95)$, $(0.035,0.8)$, $(0.04,0.8)$, and $(0,0)$ (the unconstrained case, denoted by $V_{\mathrm{un}}$ and given in Remark \ref{rmk:unconstraint}), respectively. The corresponding curves start from their effective lower bounds $\underline{x}=\max\{\alpha x_0,b_0/r\}$, which are $0.95$, $0.8$, and approximately $0.816$ for the first three cases. Panel (b) shows the relative difference $(V_{\mathrm{un}}-V)/V_{\mathrm{un}}$. The solid, dashed, and dash-dotted lines correspond to the three constrained cases $(0.035,0.95)$, $(0.035,0.8)$, and $(0.04,0.8)$, respectively, and start from the same lower bounds as in panel (a). The utility parameters are set to the base case $(\gamma_1,\gamma_2)=(0.2,0.2)$.}
            \label{fig:value_1}
\end{figure}

\begin{table}[!htbp]
    \caption{Relative utility losses (\%) for the retiree and the fund manager under the three baseline guarantee scenarios at $x=x_0=1$, with $(\gamma_1,\gamma_2)=(0.2,0.2)$.}
    \label{tab:utility-loss-decomposition}
    \centering
    \renewcommand{\arraystretch}{1.2}
    \begin{tabular*}{\textwidth}{@{\extracolsep{\fill}}lccc@{}}
        \hline
        Scenario & $(b_0,\alpha)$ & Retiree (\%) &  Manager (\%) \\
        \hline
        1 & $(0.035,0.95)$ & 2.22 & $-0.81$ \\
        2 & $(0.035,0.8)$  & 1.24 & $-0.38$ \\
        3 & $(0.04,0.8)$   & 1.32 & 0.10 \\
        \hline
    \end{tabular*}
    \smallskip
    \begin{minipage}{\textwidth}
        \small
        \textit{Note:} The utility costs in the third and fourth columns are computed as $(V_{un,1}-V_1)/V_{un,1}$ and $(V_{un,2}-V_2)/V_{un,2}$, respectively, and expressed as percentages. Here, $V_{un,1}$ and $V_1$ denote the retiree's expected discounted utility from benefits and bequest under the unrestricted and constrained optimal policies, respectively. Similarly, $V_{un,2}$ and $V_2$ denote the fund manager's expected discounted utility from management fees under these policies. Thus, $V_{un}=V_{un,1}+V_{un,2}$ and $V=V_1+V_2$. These utility components are evaluated under the respective optimal feedback policies by solving the corresponding linear Feynman--Kac equations. Negative values indicate utility gains relative to the unrestricted policy.
    \end{minipage}
\end{table}

To further understand how this cost is distributed, we decompose the value function into the retiree’s utility and the manager’s utility. Table~\ref{tab:utility-loss-decomposition} reports the corresponding relative utility losses in the three baseline scenarios. These results indicate that the utility loss from imposing the guarantee constraints mainly falls on the retiree side, whereas the manager’s utility decreases only marginally, if at all. Overall, even when the utilities are examined separately, the welfare impact remains relatively small, suggesting that the guarantee constraints can provide downside protection for retirees at only a modest efficiency cost.

\begin{table}[!htbp]
    \caption{Maximum relative utility costs under individual parameter variation.}
    \label{tab:utility-loss-sensitivity}
    \centering
    \renewcommand{\arraystretch}{1.2}
    \setlength{\tabcolsep}{4pt}
    \begin{tabular*}{\textwidth}{@{\extracolsep{\fill}}lcccc@{}}
        \hline
        Parameter varied & Range & Retiree ($\%$) & Manager ($\%$) & Combined ($\%$) \\
        \hline
        Wealth guarantee $\alpha$          & $[0.70,0.98]$   & 2.51 & 0.10 & 1.40 \\
        Minimum benefit $b_0$              & $[0.025,0.040]$ & 2.33 & 0.10 & 1.24 \\
        Common utility exponent $\gamma$   & $[0.10,0.30]$   & 3.87 & 0.20 & 2.20 \\
        Market price of risk $\theta$      & $[0.085,0.235]$ & 4.73 & 2.56 & 3.62 \\
        \hline
    \end{tabular*}
    \smallskip
    \begin{minipage}{\textwidth}
        \small
        \textit{Note:} The last three columns report maxima of $(V_{un,1}-V_1)/V_{un,1}$, $(V_{un,2}-V_2)/V_{un,2}$, and $(V_{un}-V)/V_{un}$, respectively, expressed as percentages, with the utility components defined in Table~\ref{tab:utility-loss-decomposition}. Each entry is the largest value found at $x=x_0=1$ in the corresponding numerical search. In each row, only the indicated parameter varies over the listed range, all others remain at their corresponding baseline values, and maxima are taken across the three baseline scenarios. Throughout, $\lambda=0.042$, $\rho=0.04$, $r_1=0.053$, and $\sigma=0.20$. Here, $\gamma_1=\gamma_2=\gamma$, and changes in $\theta$ are implemented through $\mu=r_1+\sigma\theta$. 
    \end{minipage}
\end{table}

Table~\ref{tab:utility-loss-sensitivity} examines the sensitivity of the utility costs to individual changes in the guarantee levels, the common utility exponent, and the market price of risk. The utility costs remain small under these individual parameter changes. Moreover, we examine joint variation in all four parameters over the same ranges, keeping the remaining parameters at their baseline values. The largest relative utility costs found in our numerical search are approximately $9.01\%$, $6.42\%$, and $8.21\%$ for the retiree, the manager, and their combined utility, respectively. These results suggest that the utility costs remain moderate even under joint parameter variation over the ranges considered.

\begin{figure}[h]
    \centering
    \begin{subfigure}[b]{0.49\textwidth}
        \includegraphics[width=\textwidth]{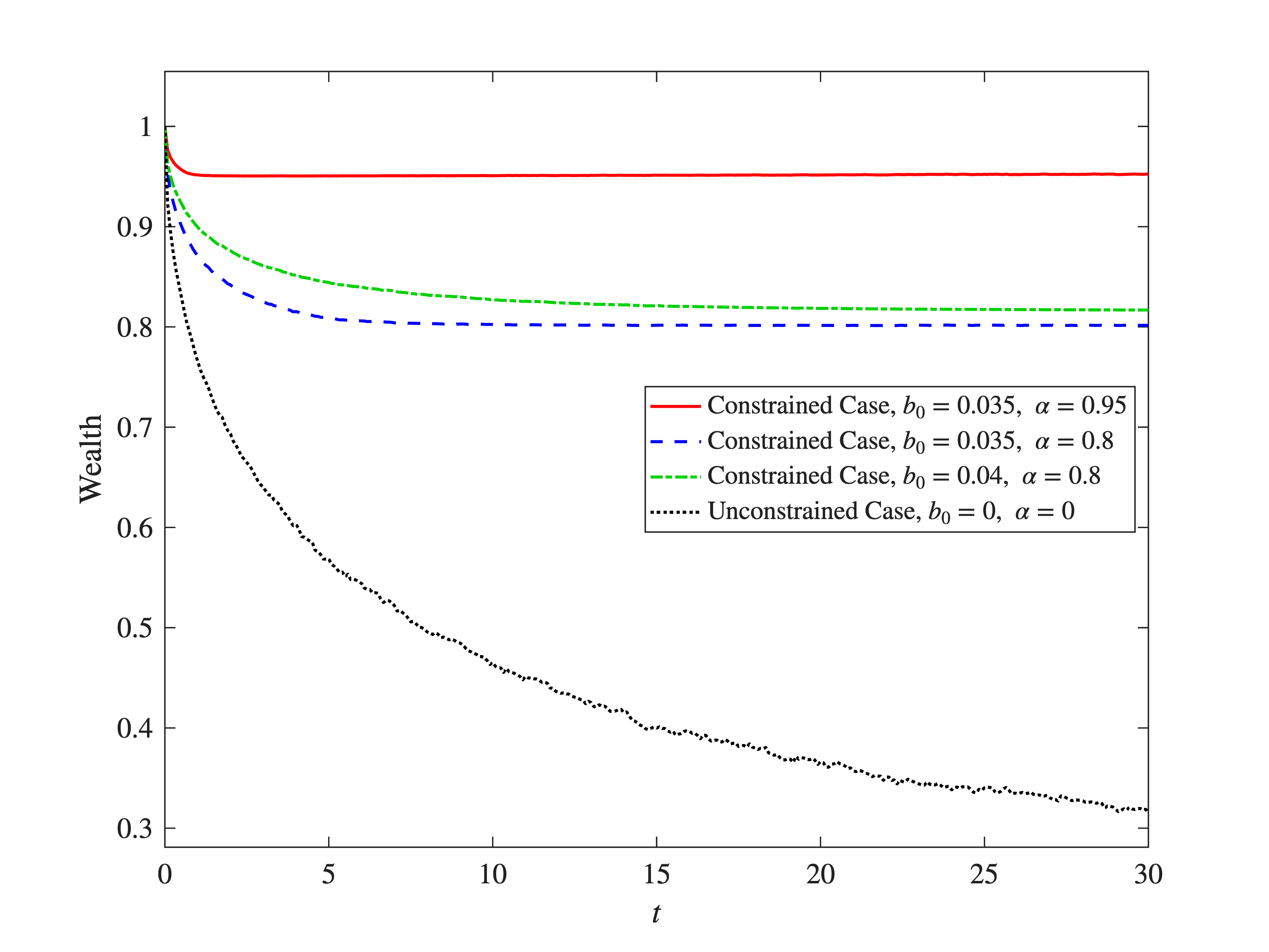}
        \caption{5th percentile of the optimal wealth }
    \end{subfigure}
    \hfill 
        \begin{subfigure}[b]{0.49\textwidth}
        \includegraphics[width=\textwidth]{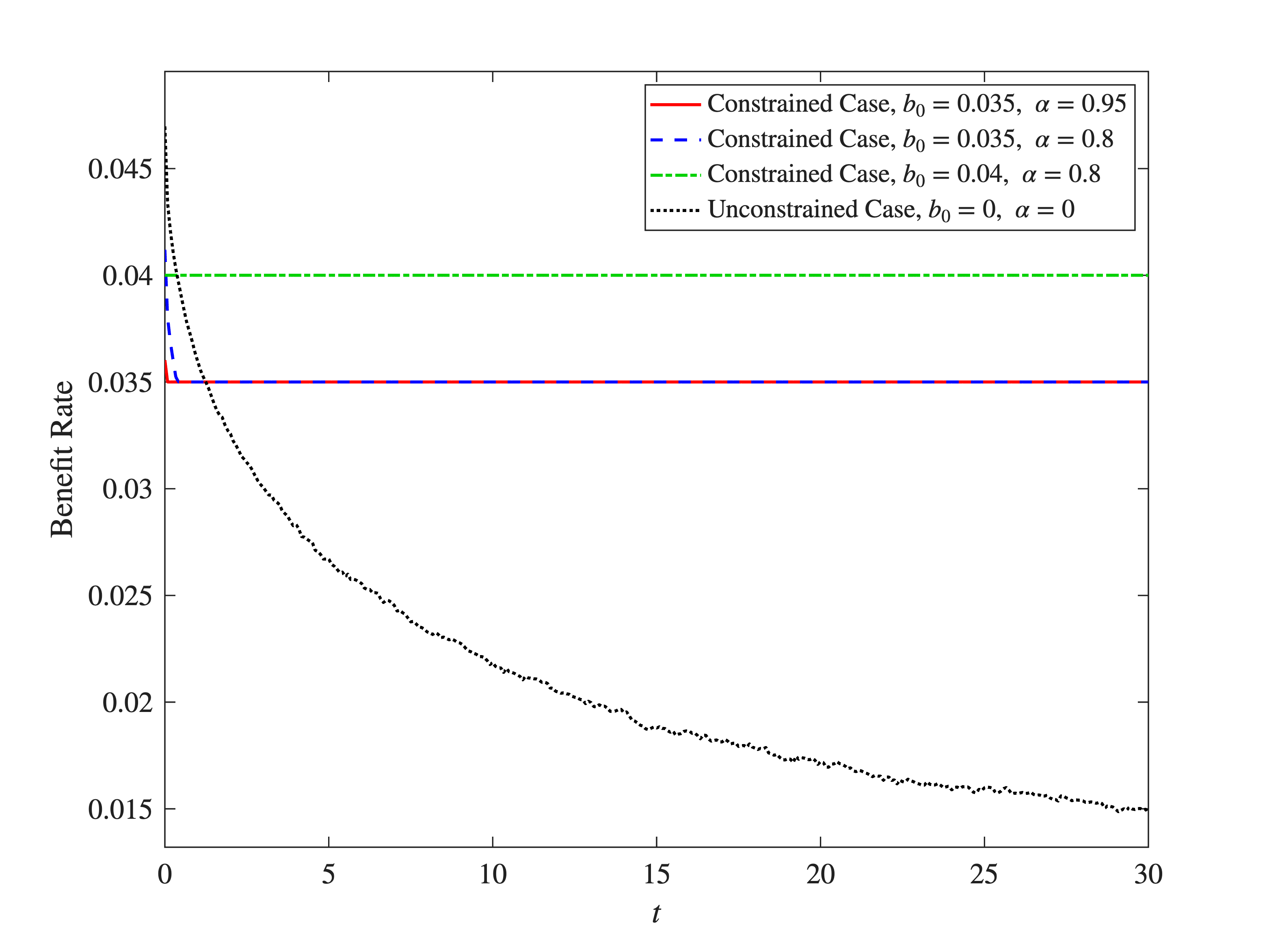}
        \caption{ 5th percentile of the optimal benefit rate }
    \end{subfigure}
    \begin{subfigure}[b]{0.49\textwidth}
        \includegraphics[width=\textwidth]{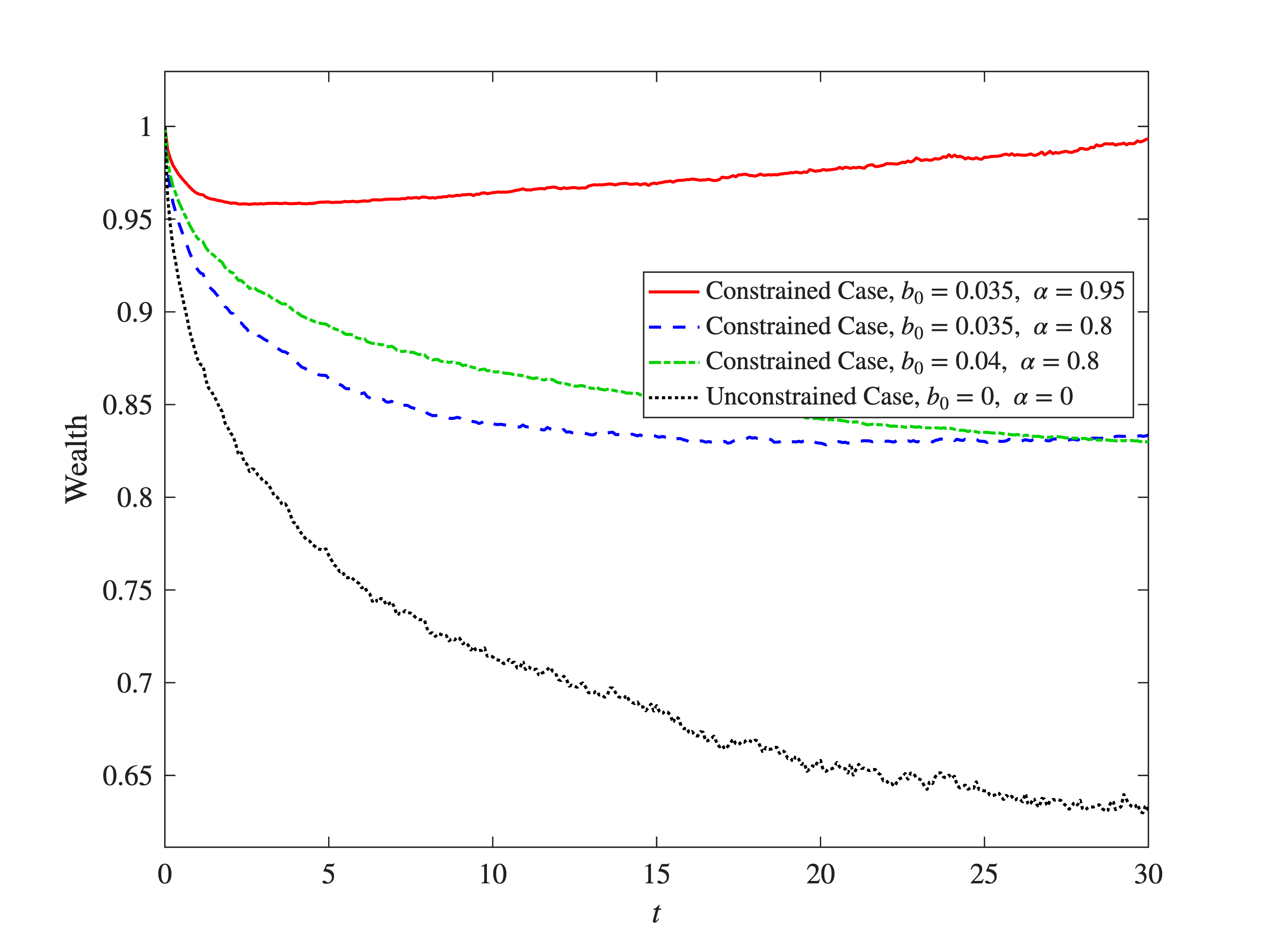}
        \caption{20th percentile of the optimal wealth }
    \end{subfigure}
    \hfill 
        \begin{subfigure}[b]{0.49\textwidth}
        \includegraphics[width=\textwidth]{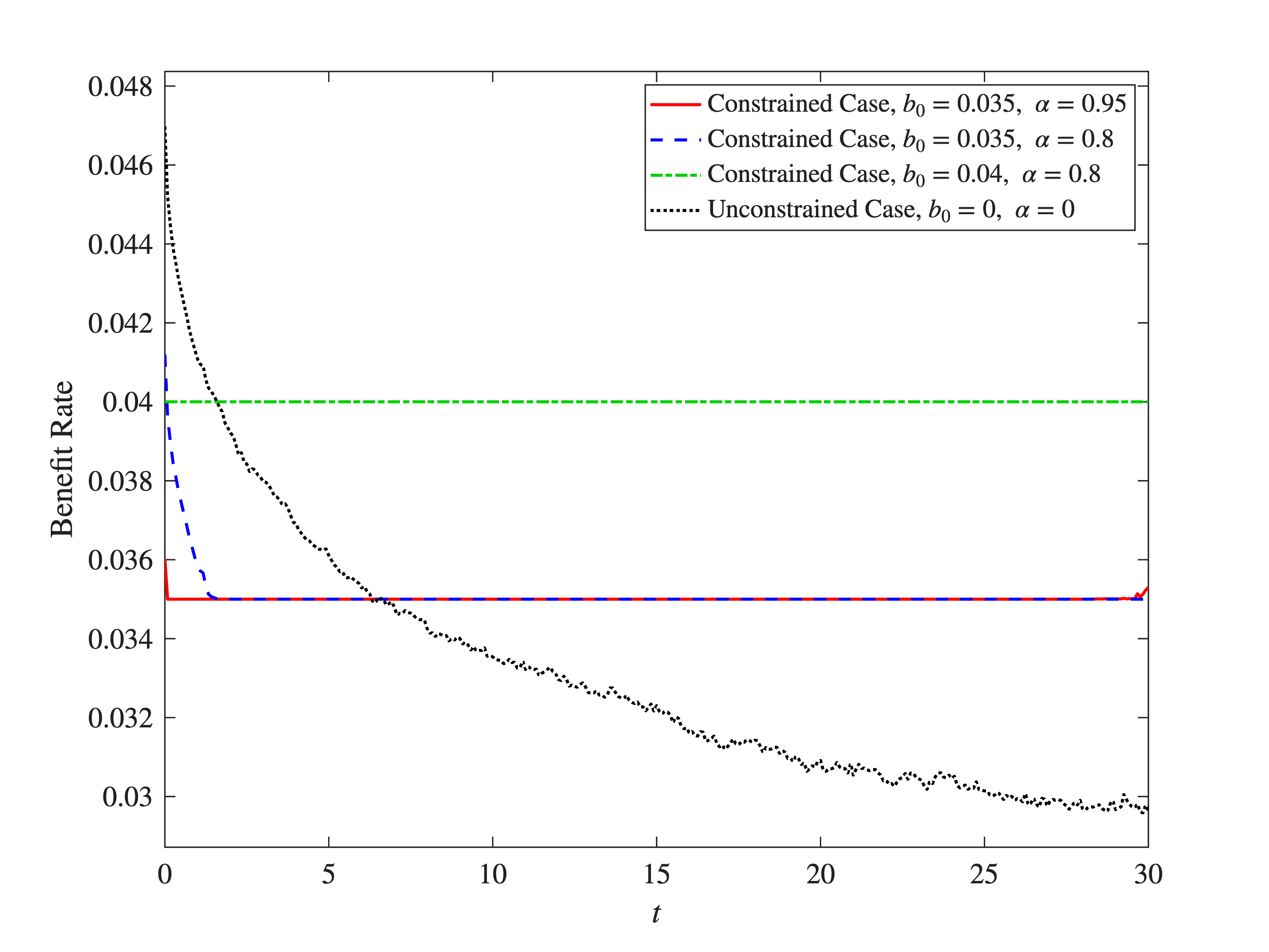}
        \caption{20th percentile of the optimal benefit rate }
    \end{subfigure}
   \caption{Percentiles of the optimal wealth and optimal benefit rate. The solid, dashed, dash-dotted, and dotted lines correspond to $(b_0,\alpha)=(0.035,0.95)$, $(0.035,0.8)$, $(0.04,0.8)$, and the unconstrained case, respectively. Panels (a) and (c) show the 5th and 20th percentiles of the optimal wealth, while panels (b) and (d) display the corresponding percentiles of the optimal benefit rate.}
\label{fig:percentile}
\end{figure}

Because the guarantee constraints are designed to provide downside protection, we focus on outcomes in the lower tail of the simulated distributions. Figure~\ref{fig:percentile} compares the constrained and unconstrained outcomes at the 5th and 20th percentiles across 10,000 simulated paths. The 5th percentile captures severe downside outcomes, while the 20th percentile provides a broader view of adverse performance. Panels (a) and (c) exhibit similar patterns. In both panels, the wealth
levels under all three constrained cases remain above those in the unconstrained
case. Specifically, at
the 5th percentile, wealth in the unconstrained case falls below $0.8$, the lowest effective wealth floor among the constrained cases, after approximately $0.71$ years, whereas
under the constrained cases, the corresponding wealth percentiles approach their respective lower bounds without falling below them. At the 20th percentile, wealth in the unconstrained case drops below
$0.8$ after approximately $3.4$ years, while the corresponding wealth percentiles in the constrained cases remain above their respective lower bounds.

Panels (b) and (d) show the corresponding patterns for benefit payments. In both cases, the unconstrained benefit rate declines over time: in panel (b), it falls below $0.035$ after approximately $1.25$ years, while in panel (d), it drops below $0.04$ after about $1.6$ years and subsequently falls below $0.035$ after approximately $6.3$ years. By contrast, under all three constrained policies, the benefit rates remain at or above their respective minimum levels $b_0$ throughout the 30-year horizon.

\begin{figure}[h]
    \centering
    \begin{subfigure}[b]{0.49\textwidth}
        \includegraphics[width=\textwidth]{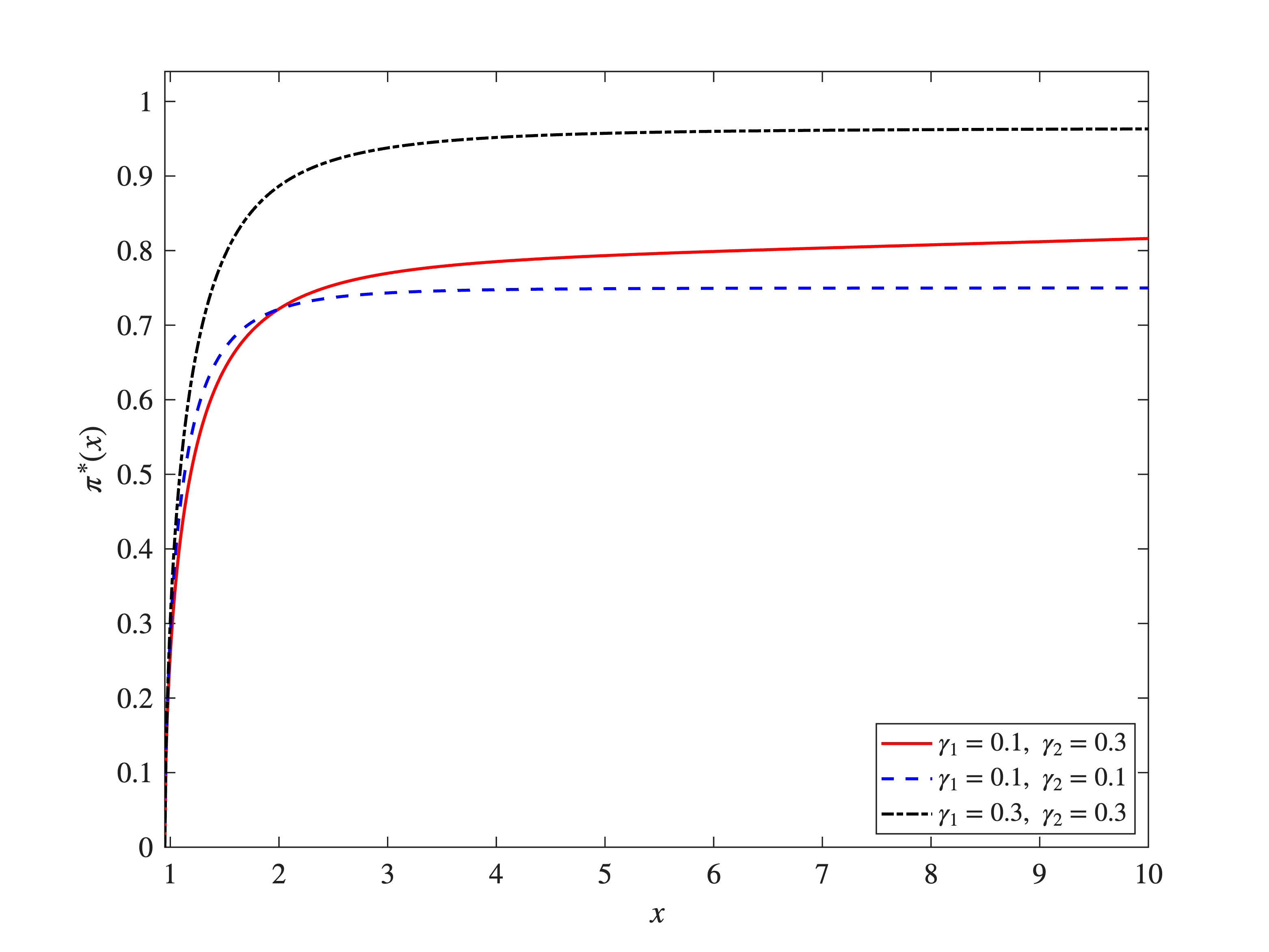}
        \caption{Optimal investment strategy}
    \end{subfigure}
    \hfill 
    \begin{subfigure}[b]{0.49\textwidth}
        \includegraphics[width=\textwidth]{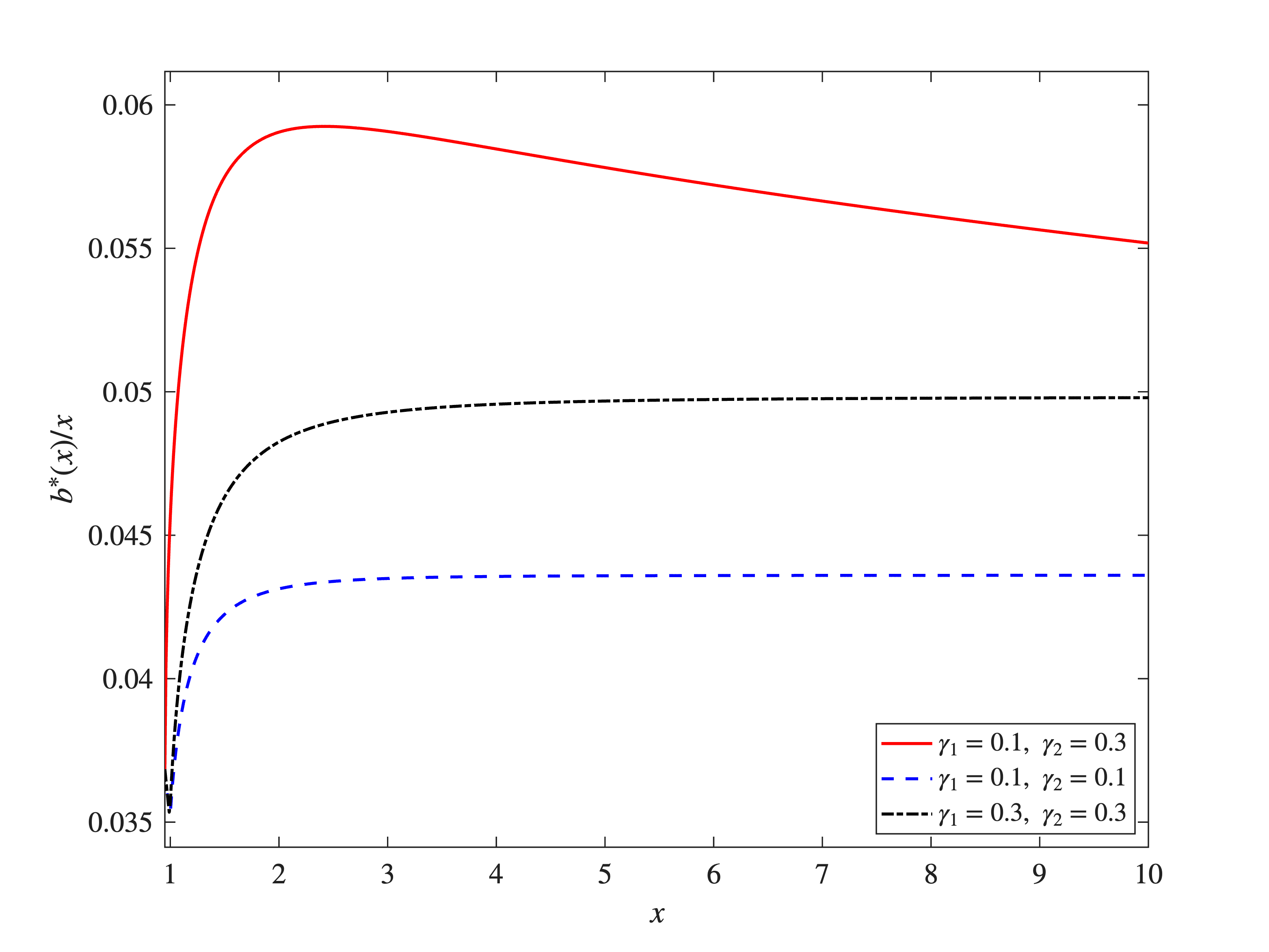}
         \caption{Optimal benefit-to-wealth strategy}
    \end{subfigure}
       \caption{Optimal strategies for different utility exponents with $(b_0,\alpha)=(0.035,0.95)$. Panel (a) shows the optimal investment strategy, while panel (b) displays the optimal benefit-to-wealth ratio. The solid, dashed, and dash-dotted lines correspond to $(\gamma_1,\gamma_2)=(0.1,0.3)$, $(0.1,0.1)$, and $(0.3,0.3)$, respectively.}
\label{fig:gamma}
\end{figure}
 Figure \ref{fig:gamma} illustrates the impact of the utility exponents on the optimal strategies under $(b_0,\alpha)=(0.035,0.95)$. In panel (a), the optimal investment generally decreases with
the overall level of risk aversion: the cases $(\gamma_1,\gamma_2)=(0.3,0.3)$, $(0.1,0.3)$, and $(0.1,0.1)$
correspond to progressively higher risk aversion and lower risky-asset allocation.  Economically, more risk-averse agents are more sensitive to downside fluctuations and therefore reduce their exposure to risky assets in favor of safer positions.

 In panel (b), comparing $(0.1,0.3)$ and $(0.1,0.1)$ (solid vs.\ dashed lines), the solid line lies above the dashed line. For a given level of wealth, a manager with lower risk
aversion is willing to allocate more to risky assets, anticipating higher future returns. As a result, the manager is more inclined to distribute a larger portion of benefits at the current stage. Comparing $(0.1,0.3)$ and $(0.3,0.3)$ (solid vs.\ dash-dotted lines), the solid line again lies above the dash-dotted line. This reflects the fact that the manager's decision incorporates the retiree's utility. When $\gamma_1$ is larger (i.e., the retiree is less risk-averse), the retiree is more willing to tolerate investment risk and therefore prefers to defer current benefits in order to retain more wealth for risky investment and achieve higher future returns. Consequently, the manager allocates more wealth to investment, accompanied by a lower benefit-to-wealth ratio.

\section{Conclusion}\label{sec:conclusion}
In this article, we study a continuous-time retirement planning problem in which a retiree delegates an initial lump sum to a professional fund manager, who determines the benefit rate and portfolio allocation under minimum wealth and benefit guarantees. The wealth floor \(X(t)\ge \alpha x_0\) ensures solvency and a minimum bequest, whereas the benefit floor \(b(t)\ge b_0\) guarantees a lifelong minimum income. The objective allows for general utility functions over the retiree's benefits and bequest and the manager's running and terminal management fees, thereby accommodating a broad range of preferences and trade-offs between retiree welfare and managerial incentives.

Applying dynamic programming, we derive the associated HJB equation for the value function and use a transform technique to reduce the fully nonlinear HJB equation to a semilinear auxiliary equation. When \(\alpha x_0>b_0/r\), monotonicity and asymptotic arguments allow us to select an appropriate solution from a parameterized family of global solutions. When \(\alpha x_0\le b_0/r\), we establish existence using the Schauder fixed point theorem and then verify the required monotonicity and growth properties through qualitative analysis. Transforming these solutions back yields classical solutions to the original HJB equation with the appropriate properties. Finally, verification theorems identify these solutions with the value function and establish the optimality of the associated feedback strategies.

Our numerical experiments reveal several insights. First, the guarantee constraints reduce the fund’s exposure to risky assets to protect solvency and, subject to the minimum-benefit requirement, limit discretionary benefit payments to preserve capital. As wealth increases, the influence of the constraints weakens, and both policies converge to their unconstrained benchmarks. Second, the welfare cost of providing these guarantees is modest and falls mainly on the retiree, while the manager’s utility changes only marginally. This limited welfare loss may therefore represent an acceptable price for the greater security afforded by the wealth and benefit guarantees. Third, lower-tail simulations show that the guarantees provide effective downside protection: constrained policies keep wealth and benefit payments at or above their prescribed floors, whereas both can fall substantially below these prescribed levels under the unconstrained policy. Overall, these findings suggest that minimum guarantees are a practically attractive feature of retirement decumulation products, offering meaningful protection against adverse outcomes at only a modest welfare cost.

Several extensions merit further investigation. First, the wealth and benefit floors could be allowed to evolve deterministically over time. Appendix~\ref{sec:extend} analyzes the special case of exponentially growing floors under CRRA utility, whereas the corresponding problem with general utility functions remains open. Second, the constant mortality intensity could be replaced by an age-dependent specification, such as the Makeham law. In their full generality, both extensions render the control problem time-inhomogeneous, leading to a PDE characterization of the value function and requiring analytical techniques beyond those developed in this paper.

\stepcounter{footnote}
\begin{figure}[H]
    \RawFloats
    \centering
    \begin{subfigure}[t]{0.49\textwidth}
        \vspace{0pt}
        \includegraphics[width=\textwidth]{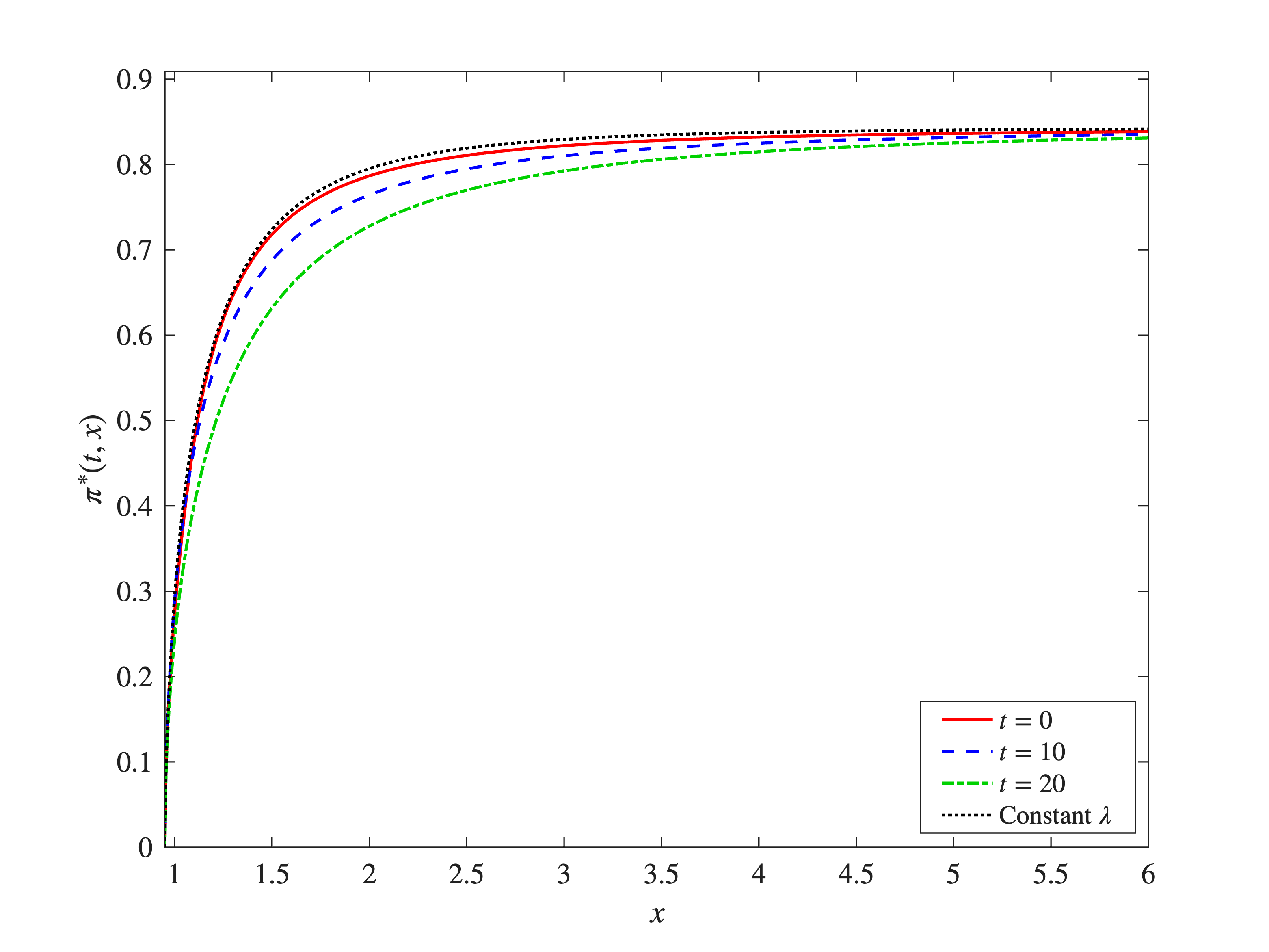}
        \caption{Optimal investment in Scenario~1.}
    \end{subfigure}
    \hfill
    \begin{subfigure}[t]{0.49\textwidth}
        \vspace{0pt}
        \includegraphics[width=\textwidth]{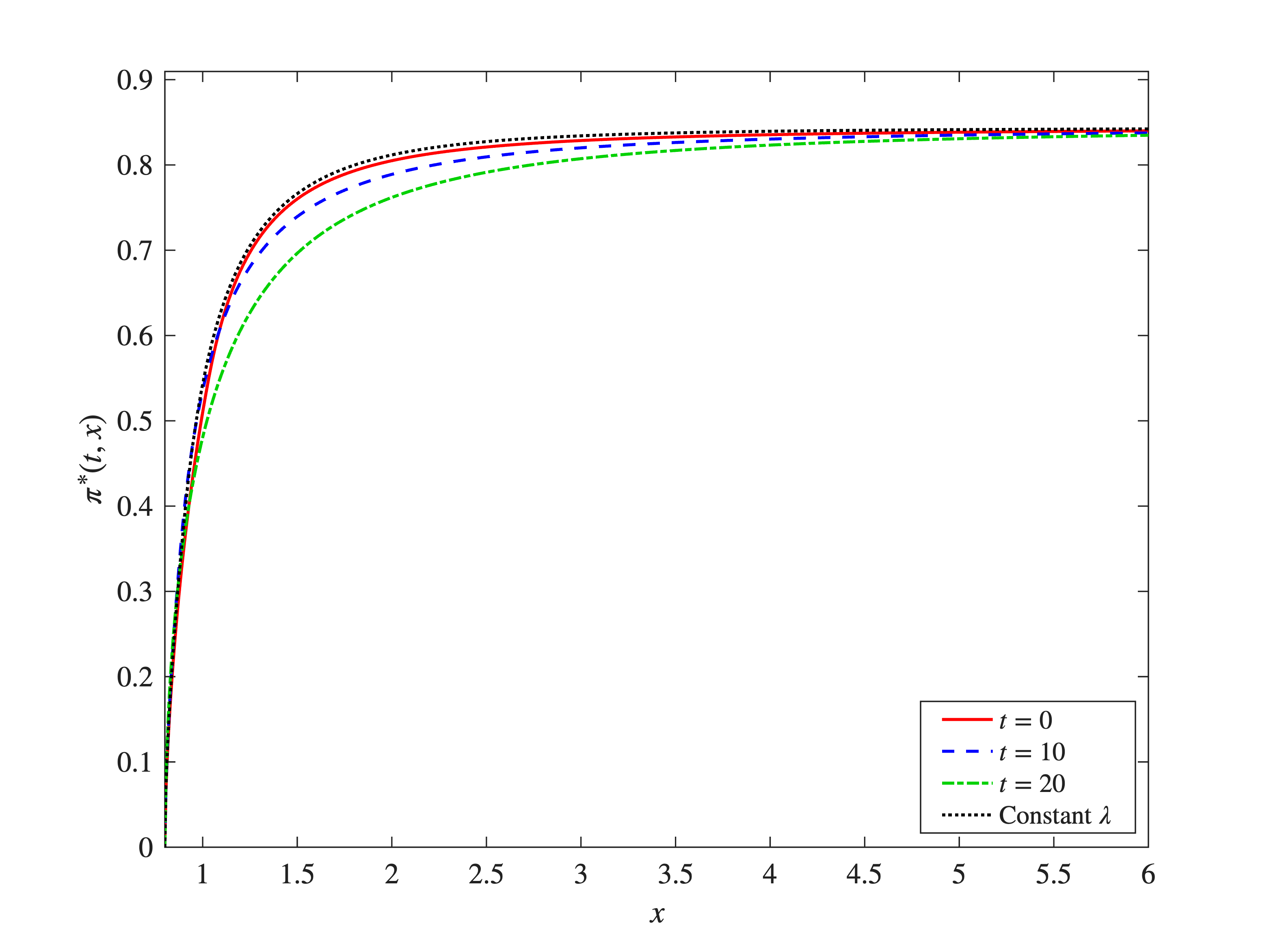}
        \caption{Optimal investment in Scenario~2.}
    \end{subfigure}
    \par\medskip
    \begin{subfigure}[t]{0.49\textwidth}
        \vspace{0pt}
        \includegraphics[width=\textwidth]{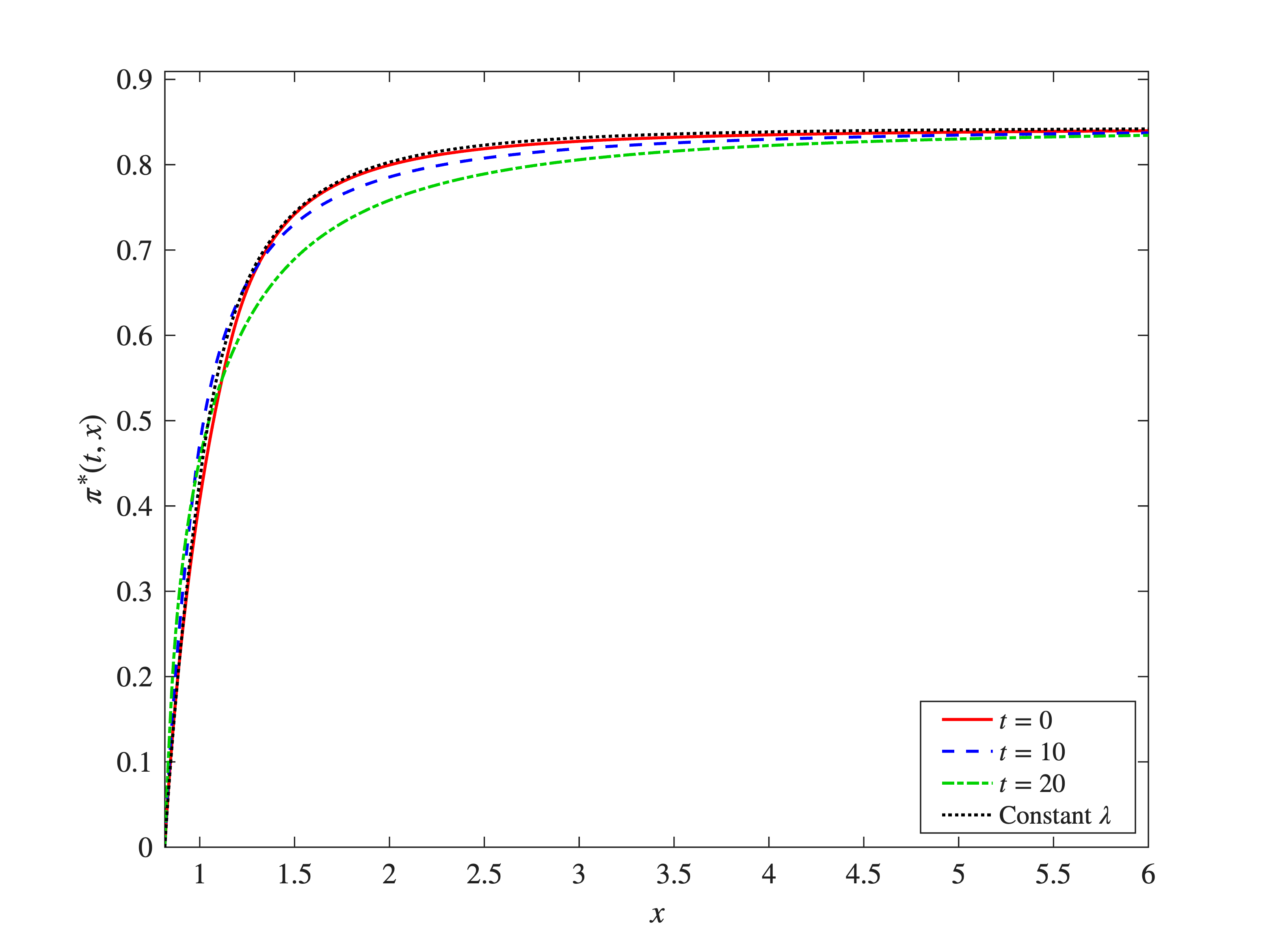}
        \caption{Optimal investment in Scenario~3.}
    \end{subfigure}
     \caption[Optimal investment proportions under Makeham and constant mortality.]{Optimal investment proportions under Makeham and constant mortality. Panels (a), (b), and (c) correspond to $(b_0,\alpha)=(0.035,0.95)$, $(0.035,0.8)$, and $(0.04,0.8)$, respectively. In each panel, the solid, dashed, and dash-dotted lines show the optimal investment proportion $\pi^*(t,x)$ under the recalibrated Makeham law\protect\footnotemark[\value{footnote}] $\lambda(t)=0.0056+0.0035e^{0.121t}$ at $t=0,10,20$, corresponding to ages 60, 70, and 80, respectively. The dotted line shows the optimal investment proportion under the constant mortality intensity $\lambda=0.042$. The respective mean remaining lifetimes at age 60 are approximately $23.54$ and $23.81$ years, close to the U.S. estimate of $23.6$ years \cite*[Table~A]{nchs_life_expectancy}. All other parameters are held at their baseline values.}
    \label{fig:makeham-time-constant-six}
\end{figure}
\footnotetext{The Makeham model $\mu(a)=0.00022+2.7\times10^{-6}(1.124)^a$ in \cite*[Section~4.3]{dickson2013actuarial} is widely used for actuarial illustrations. Here $a$ is actual age in years, whereas $t:=a-60$ is the number of years elapsed since age 60, with $t=0$ at age 60.  The book's model becomes $\lambda_{\mathrm{book}}(t)=\mu(60+t)=0.00022+2.7\times10^{-6}(1.124)^{60+t}$. It implies $27.21$ remaining years at age 60, exceeding the U.S. estimate of $23.6$ years;  see \cite*[Table~A]{nchs_life_expectancy}.
 We therefore recalibrate to $\lambda(t)=0.0056+0.0035e^{0.121t}$, yielding $23.54$  remaining years.}

For the Makeham specification, the numerical experiments in Figure~\ref{fig:makeham-time-constant-six} show that the optimal investment proportions remain close to those under constant mortality across the three scenarios, particularly at $t=0$ and $t=10$, corresponding to ages 60 and 70. These findings suggest that, under the parameter settings considered, the constant-mortality model can provide a useful practical benchmark for fund managers' investment decisions in early retirement. A systematic analysis of how age-dependent mortality affects optimal benefit payments remains an interesting direction for future research.

\section*{Acknowledgement}
 Zhaojie Ren acknowledges the financial support as a postdoctoral fellow from the Department of Statistics and Data Science at The Chinese University of Hong Kong. 
 Sheng Wang acknowledges the financial support as a postdoctoral fellow from the 
  Department of Statistics and Actuarial Science at The University of Hong Kong.  Tak Kwong Wong was partially supported by NSFC grant with project number
12571248.
Phillip Yam acknowledges the financial support from HKGRF-14300123 with the project title ``Well-posedness of Some Poisson-driven Mean Field Learning Models and their Applications'', and  HKSAR-GRF 14300025 with the project title ``A Generic Theory for Stochastic Control against Fractional Brownian Motions''. The work described in this article was also supported by a grant from the Germany/Hong Kong Joint Research Scheme sponsored by the Research Grants Council of Hong Kong and the German Academic Exchange Service of Germany (Reference No. G-CUHK411/23). He also thanks The University of Texas at Dallas for the kind invitation to be a Visiting Professor in Naveen Jindal School of Management.

\bibliographystyle{plain}
\bibliography{retirement}

\appendix
\makeatletter
\let\appendix@seccntformat\@seccntformat
\renewcommand{\@seccntformat}[1]{  \ifstrequal{#1}{section}{\thesection:\quad}{\appendix@seccntformat{#1}}}
\makeatother
\section*{Appendix}
\section{Proofs for the statements in Section \ref{sec:formu}}

For later use, we  introduce an equivalent probability measure. 
For any $T>0$, we  define the probability measure $\mathbb{Q}_{T}$ on $\mathcal{F}_T$ by
\begin{align}\label{eq:Q:measure}
    \frac{\md\mathbb{Q}_{T}}{\md \mathbb{P}}
    =\exp\left(-\theta W_{T}-\frac12 \theta^2T\right).
\end{align}
By Girsanov's theorem, the process 
$\{\widetilde{W}(t):=W(t)+\theta t:t\in[0,T]\}$ 
is a Brownian motion under $\mathbb{Q}_{T}$. 
This change of measure will be repeatedly used in the proofs below.
 
\paragraph{Proof of Lemma \ref{lma:b_0/r}}\label{appen:lma:b_0/r}
 We assume $\frac{b_0}{r}>\alpha x_0$. Suppose, by contradiction, that there exists some $t_0\geq0$ such that $\mathbb{P}( X(t_0)< \frac{b_0}{r})>0$.  For $0<\epsilon<\frac{b_0}r$, define the stopping time
\begin{align*}
    \tau_{\epsilon}:=\inf\left\{s:X(s)\leq \frac{b_0}{r}-\epsilon\right\}.
\end{align*}
Then there exist $\epsilon>0$ and $T>0$ such that $\mathbb{P}(\tau_{\epsilon}\leq T)>0$.
Let $M>0$ be sufficiently large such that
\begin{align}\label{eq:select:M}
    \int_0^M\mathrm{e}^{-rs}b_0\md s\geq \frac{b_0}{r}-\frac{\epsilon}{2}.
\end{align}
Set $\hat{\tau}_{\epsilon}:=\tau_{\epsilon}\wedge T$ and, for $s\in[0,M]$, define the shifted process
\begin{align*}
    Z_{\epsilon}(s):=\mathrm{e}^{-rs}X(\hat{\tau}_{\epsilon}+s)
    +\int_0^s\mathrm{e}^{-ru}b(\hat{\tau}_{\epsilon}+u)\md u,
\end{align*}
where $X$ is defined in \eqref{eq:sde:x:intro}.
Under $\mathbb{Q}_{T+M}$, with the shifted Brownian motion $\widetilde{W}_{\epsilon}(s):=\widetilde{W}(\hat{\tau}_{\epsilon}+s)-\widetilde{W}(\hat{\tau}_{\epsilon})$, $Z_{\epsilon}$ satisfies
\begin{align*}
    \md Z_{\epsilon}(s)
    =\mathrm{e}^{-rs}\sigma\pi(\hat{\tau}_{\epsilon}+s)
    X(\hat{\tau}_{\epsilon}+s)\md\widetilde{W}_{\epsilon}(s),
    \qquad Z_{\epsilon}(0)=X(\hat{\tau}_{\epsilon}).
\end{align*}
Thus, $Z_{\epsilon}$ is a nonnegative local martingale, and hence a supermartingale, with respect to the shifted filtration $\{\mathcal{F}_{\hat{\tau}_{\epsilon}+s}:s\in[0,M]\}$. Let $A_{\epsilon,T}:=\{\tau_{\epsilon}\leq T\}\in\mathcal{F}_{\hat{\tau}_{\epsilon}}$. On $A_{\epsilon,T}$, we have $X(\hat{\tau}_{\epsilon})=\frac{b_0}{r}-\epsilon$, and, by \eqref{eq:select:M},
\begin{align*}
    Z_{\epsilon}(M)\geq\int_0^M\mathrm{e}^{-ru}b_0\md u
    \geq\frac{b_0}{r}-\frac{\epsilon}{2}.
\end{align*}
Therefore,
\begin{align*}
    \left(\frac{b_0}{r}-\epsilon\right)\mathbbm{1}_{A_{\epsilon,T}}
    &\geq \mE^{\mathbb{Q}_{T+M}}
    \left[Z_{\epsilon}(M)\mathbbm{1}_{A_{\epsilon,T}}
    \,\middle|\,\mathcal{F}_{\hat{\tau}_{\epsilon}}\right]\geq\left(\frac{b_0}{r}-\frac{\epsilon}{2}\right)
    \mathbbm{1}_{A_{\epsilon,T}},
\end{align*}
which is a contradiction because $\mathbb{Q}_{T+M}$ is equivalent to $\mathbb{P}$ and $\mathbb{P}(A_{\epsilon,T})>0$.

\section{Proofs for the statements in Section \ref{sec:dpp:hjb}}
The following proof is analogous to that of Lemma \ref{lma:b_0/r}.
\subsection{Proof of Lemma \ref{lma:stop}}\label{appen:lma:stop}
It suffices to show that for any $T>0$ such that $\mathbb{P}(\{X^{x;u}(\tau)=\underline{x}\}\cap\{\tau\leq T\})>0$,
the conclusion holds on the event $\Omega_T:=\{X^{x;u}(\tau)=\underline{x}\}\cap\{\tau\leq T\}$.

Fix $T>0$ and let $\hat{\tau}:=\tau\wedge T$. For any $M>0$ and $s\in[0,M]$, define
\begin{align*}
    Z(s):=\mathrm{e}^{-rs}X^{x;u}(\hat{\tau}+s)
    +\int_0^s\mathrm{e}^{-ru}b(\hat{\tau}+u)\md u.
\end{align*}
As in the proof of Lemma \ref{lma:b_0/r}, under $\mathbb{Q}_{T+M}$, $Z$ is a nonnegative local martingale, and hence a supermartingale, with respect to the shifted filtration $\{\mathcal{F}_{\hat{\tau}+s}:s\in[0,M]\}$. Moreover, $\Omega_T\in\mathcal{F}_{\hat{\tau}}$, and on $\Omega_T$ we have $\hat{\tau}=\tau$, $Z(0)=\underline{x}$, and
\begin{align}\label{eq:Z:boundary:nonnegative}
    Z(s)-\underline{x}
    =\mathrm{e}^{-rs}\left(X^{x;u}(\tau+s)-\underline{x}\right)
    +\int_0^s\mathrm{e}^{-ru}\left(b(\tau+u)-b_0\right)\md u
    \geq0.
\end{align}
Consequently, for every $s\in[0,M]$,
\begin{align*}
    \underline{x}\mathbbm{1}_{\Omega_T}
    &\geq\mE^{\mathbb{Q}_{T+M}}
    \left[Z(s)\mathbbm{1}_{\Omega_T}\,\middle|\,\mathcal{F}_{\hat{\tau}}\right]
    \geq\underline{x}\mathbbm{1}_{\Omega_T}.
\end{align*}
It follows that $Z(s)=\underline{x}$ on $\Omega_T$ for every rational $s\in[0,M]$, and hence, by continuity, for every $s\in[0,M]$. Since both terms on the right hand side of \eqref{eq:Z:boundary:nonnegative} are nonnegative, we obtain $X^{x;u}(\tau+s)=\underline{x}$ for all $s\in[0,M]$ and $b(\tau+s)=b_0$ for a.e. $s\in[0,M]$, a.s. on $\Omega_T$. Taking $M,T\in\mathbb{N}$ and using countability, the conclusion follows on $\{X^{x;u}(\tau)=\underline{x}\}\cap \{\tau<\infty\}=\bigcup_{T\in\mathbb{N}}\Omega_T$.

\subsection{Proof of  Theorem \ref{thm:verification}}\label{appendix:sec:veri1}
Before giving the proof of Theorem \ref{thm:verification}, we first give two lemmas.
\begin{lemma}\label{lma:limit:v''}
 Under Assumptions \ref{assump:u:1} and \ref{assump:parameter}, and the conditions imposed on $v$ in Theorem \ref{thm:verification}, it holds that $r\underline{x}>I(v'(\underline{x}))$ and
    \begin{align}
 		\lim_{x\downarrow \underline{x}}(x-\underline{x})[v''(x)]^2=\frac14\theta^2(v'(\underline{x}))^2\frac{1}{r\underline{x}-I(v'(\underline{x}))}.\label{eq:verify:limit}
 	\end{align}
\end{lemma}
\begin{proof}
 We first prove $r\underline{x}>I(v'(\underline{x}))$. If $v'(\underline{x})\geq U_1'(b_0)$, then $I(v'(\underline{x}))=b_0<r\underline{x}$; If $U_1'(r\underline{x})<v'(\underline{x})< U_1'(b_0)$, then $I(v'(\underline{x}))<I(U_1'(r\underline{x}))=r\underline{x}$. In both cases, we have $r\underline{x}>I(v'(\underline{x}))$.

We next prove \eqref{eq:verify:limit}. Define 
 	\begin{align}	g(x):=D(v'(x))+U_2(x)+rxv'(x)-\delta v(x),\quad x>\underline{x}.\label{eq:defi_g}
 	\end{align}
    Because $v$ satisfies \eqref{eq:hjb:simplify} and $v''<0$ in $(\underline{x},\infty)$, we have $g<0$ in $(\underline{x},\infty)$ and 
 	\begin{align}\label{eq:v''}
 		v''(x)=\frac{1}{2}\theta^2 \frac{(v'(x))^2}{D(v'(x))+U_2(x)+rxv'(x)-\delta v(x)}=\frac{1}{2}\theta^2\frac{(v'(x))^2}{g(x)}, \quad x\in(\underline{x},\infty).
 	\end{align}
    Moreover,
 	\begin{align*}
 		g'(x)=U_2'(x)+\left[rx- I\left(v'(x)\right)\right]v''(x)+(r-\delta) v'(x), \quad x>\underline{x}.
 	\end{align*}
 	Using \eqref{eq:v''} and $\lim\limits_{x\downarrow\underline{x}}v''(x)=-\infty$ imposed in Theorem \ref{thm:verification}, we have $g(\underline{x})=0$. Applying l’Hôpital’s rule and using $\lim\limits_{x\downarrow\underline{x}}g'(x)=-\infty$, we obtain
 	\begin{align*}
 		\lim_{x\downarrow \underline{x}}(x-\underline{x})[v''(x)]^2&=\frac14\theta^4(v'(\underline{x}))^4\lim_{x\downarrow \underline{x}}\frac{(x-\underline{x})}{g^2(x)}=\frac14\theta^4(v'(\underline{x}))^4\lim_{x\downarrow \underline{x}}\frac{1}{2g(x)g'(x)}\nonumber\\
 		&=\frac14\theta^2(v'(\underline{x}))^2\lim_{x\downarrow \underline{x}}\frac{v''(x)}{g'(x)}=\frac14\theta^2(v'(\underline{x}))^2\frac{1}{r\underline{x}-I(v'(\underline{x}))}.
 	\end{align*}
\end{proof}
\begin{lemma}\label{lma:X^gamma}
Under Assumptions \ref{assump:u:1} and \ref{assump:parameter}, together with the conditions imposed on $v$ in Theorem \ref{thm:verification}, for any $T>0$ and $x\geq \underline{x}$, we have
    \begin{align*}
        \mE\left[|X^{x;*}(T)|^{\gamma}\right]\leq x^{\gamma}\exp\left(\left(\gamma r+\frac12\frac{\gamma\theta^2}{1-\gamma}\right)T\right).
    \end{align*}
\end{lemma}
\begin{proof}
    For any $T>0$ and $x\geq \underline{x}$, the optimal wealth process $X^{x;*}$ satisfies
\begin{align*}\begin{cases}
	\md X^{x;*}(s)=rX^{x;*}(s)\md s-b^*(X^{x;*}(s))\md s+\sigma\pi^*(X^{x;*}(s))X^{x;*}(s)\md \tilde{W}(s),\quad s\in[0,T],\\
	X^{x;*}(0)=x,
	\end{cases}
\end{align*}
where we recall $\widetilde{W}=\left\{\widetilde{W}_t=W_t+\theta t:t\in[0,T]\right\}$ is a Brownian Motion under the probability measure $\mathbb{Q}_T$,  defined  in \eqref{eq:Q:measure}.
Note that $\sigma \pi^*(x)x=-\theta\frac{v'(x)}{v''(x)}$ is of linear growth, as imposed in Theorem \ref{thm:verification}. Therefore, we can take expectations under $\mathbb{Q}_T$ and use the fact that $b^*(\cdot)\ge b_0>0$ to obtain
\begin{align*}
   \mE^{\mathbb{Q}_T}\left[X^{x;*}(t)\right]\leq x+r \int_0^t\mE^{\mathbb{Q}_T}\left[X^{x;*}(s)\right]\md s\quad \text{for all } t\in[0,T].
\end{align*}
An application of Gr\"onwall’s inequality yields
\begin{align*}
    \mE^{\mathbb{Q}_T}\left[X^{x;*}(T)\right]\leq x\exp(rT).
\end{align*}
By the change of measure and Hölder’s inequality, we obtain
\begin{align*}
     \mE\left[|X^{x;*}(T)|^{\gamma}\right]&= \mE^{\mathbb{Q}_T}\left[|X^{x;*}(T)|^{\gamma}\exp\left(\theta \tilde{W}_T-\frac12 \theta^2 T\right)\right]\\
     &\ \leq \left(\mE^{\mathbb{Q}_T}\left[X^{x;*}(T)\right]\right)^{\gamma}\left(\mE^{\mathbb{Q}_T}\left[\exp\left(\frac{1}{1-\gamma}\left(\theta \tilde{W}_T-\frac12 \theta^2 T\right)\right)\right]\right)^{1-\gamma}\\
     &\ \leq x^{\gamma}\exp\left(\left(\gamma r+\frac12\frac{\gamma\theta^2}{1-\gamma}\right)T\right).
\end{align*}
\end{proof}

\paragraph{Proof of Theorem \ref{thm:verification}:}

	Because $\lim\limits_{x\downarrow \underline{x}}v''(x)=-\infty$, we approximate $v$ by a sequence  $\{v_\epsilon\}_{\epsilon\in(0,1)} \subset C^2([\underline{x},\infty))$ in order to apply Itô's formula. For each $\epsilon\in(0,1)$, define    
	\begin{align*}
		v_{\epsilon}(x):=\left\{\begin{aligned}
			&v(x),&&x\in[\underline{x}+\epsilon,\infty),\\
			&a_{\epsilon}+c_{\epsilon}x+d_{\epsilon}x^2, &&x\in[\underline{x},\underline{x}+\epsilon),
		\end{aligned}\right.
	\end{align*} 
	where the constants $a_{\epsilon}, c_{\epsilon}$ and $ d_{\epsilon}$ are chosen so that $v_{\epsilon}\in C^2([\underline{x},\infty))$,  namely,
	\begin{align*}
	&	a_{\epsilon}=v(\underline{x}+\epsilon)-(\underline{x}+\epsilon)v'(\underline{x}+\epsilon)+\frac12 (\underline{x}+\epsilon)^2v''(\underline{x}+\epsilon),\\
		&c_{\epsilon}=v'(\underline{x}+\epsilon)-v''(\underline{x}+\epsilon)(\underline{x}+\epsilon),\quad d_{\epsilon}=\frac12 v''(\underline{x}+\epsilon). 
	\end{align*}
For $x\in[\underline{x},\underline{x}+\epsilon]$, straightforward calculations yield
	\begin{align}
		v_{\epsilon}(x)-v(x)&=v(\underline{x}+\epsilon)-v(x)-(\underline{x}+\epsilon)v'(\underline{x}+\epsilon)+\frac12 (\underline{x}+\epsilon)^2v''(\underline{x}+\epsilon)\nonumber\\
		&\quad\quad\quad+(v'(\underline{x}+\epsilon)-v''(\underline{x}+\epsilon)(\underline{x}+\epsilon) )x+\frac12 v''(\underline{x}+\epsilon) x^2\nonumber\\
		&=v(\underline{x}+\epsilon)-v(x)+(x-\underline{x}-\epsilon)v'(\underline{x}+\epsilon)+\frac12(x-\underline{x}-\epsilon)^2v''(\underline{x}+\epsilon),\label{eq:v:converge}\\
		v'_{\epsilon}(x)-v'(x)&=v'(\underline{x}+\epsilon)-v'(x)+(x-\underline{x}-\epsilon)v''(\underline{x}+\epsilon).\label{eq:v':converge}
	\end{align}
   Moreover, 
    \begin{align}\label{eq:conver:cross}
        \limsup_{\epsilon\downarrow 0} \sup_{x\in [\underline{x},\underline{x}+\epsilon]} \left|(\underline{x}+\epsilon-x)v''(\underline{x}+\epsilon)\right|^2\leq \limsup_{\epsilon\downarrow 0}  \left|\epsilon v''(\underline{x}+\epsilon)\right|^2=0,
    \end{align}
    where the  inequality uses $0\leq \underline{x}+\epsilon-x\leq \epsilon$ for any $x\in[\underline{x},\underline{x}+\epsilon]$, and the equality uses \eqref{eq:verify:limit}.
   Therefore, we can use \eqref{eq:conver:cross} to dominate the remainder of \eqref{eq:v:converge} and \eqref{eq:v':converge}, which ensures the uniform convergence. Together with the fact that $v$ and $v'$ are uniformly continuous and bounded on $[\underline{x},\underline{x}+1]$, we obtain
	\begin{align*}
		&v_{\epsilon}\to v \,\,\, \text{and}\,\,\,v_{\epsilon}'\to v' \quad \text{uniformly on }  [\underline{x},\infty)\quad \text{as } \epsilon\to 0.
	\end{align*}
    Moreover, because $D(\cdot)$ (defined in \eqref{eq:D}) is $C^1$, it follows that, as $\epsilon\to 0$,
    \begin{align*}
		&\sup _{b\geq b_0} \{U_1(b)-bv_\epsilon'(x)\}=D(v_{\epsilon}'(x)) \to D(v'(x))= \sup _{b\geq b_0} \{U_1(b)-bv'(x)\} \quad \text{uniformly on }  [\underline{x},\infty).
    \end{align*}
	Furthermore,  using $\lim\limits_{x\downarrow\underline{x}}v''(x)=-\infty$, as $\epsilon\to 0$,
    \begin{align*}
        \frac{1}{v_{\epsilon}''(x)}-\frac{1}{v''(x)}=\begin{cases}
           &\frac{1}{v''(\underline{x}+\epsilon)} - \frac{1}{v''(x)},x\in[\underline{x}, \underline{x}+\epsilon],\\
           &0,\quad x>\underline{x}+\epsilon,
        \end{cases} \quad \text{ uniformly converges to $0$ on }  [\underline{x},\infty).
    \end{align*}
Consequently, as $\epsilon\to 0$,
	\begin{align*}
	&\sup_{\pi\in\mR}\{\sigma \theta\pi x v'_\epsilon(x)+\frac{\sigma^2}{2}\pi^2x^2v_\epsilon''(x)\}=-\frac{1}{2}\theta^2\frac{\left(v_{\epsilon}'(x)\right)^2}{v_\epsilon''(x)} \\
     &\quad\quad\quad\to -\frac{1}{2}\theta^2\frac{\left(v'(x)\right)^2}{v''(x)}=	\sup_{\pi\in\mR}\{\sigma \theta\pi x v'(x)+\frac{\sigma^2}{2}\pi^2x^2v''(x)\}	 \quad \text{uniformly on } [\underline{x},\infty).
	\end{align*}
	Therefore, $v_\epsilon$ satisfies the following equation:
	\begin{align}
			-\delta v_\epsilon(x)&+\sup _{b\geq b_0} \{U_1(b)-bv_\epsilon'(x)\}+\sup_{\pi\in\mR}\{\sigma \theta\pi x v_\epsilon'(x)+\frac{\sigma^2}{2}\pi^2x^2v_\epsilon''(x)\}\nonumber\\&+U_2(x)+rxv_\epsilon'(x)=g_\epsilon(x),\quad x\geq \underline{x},\label{eq:verify:hjb:epsilon}
	\end{align}
	where $g_\epsilon$ is continuous and bounded, satisfies 
$g_\epsilon(x)=0$ for $x>\underline{x}+\epsilon$, and 
$g_\epsilon\to 0$ uniformly on $[\underline{x},\infty)$.
    
For any $x\geq \underline{x}$, $u\in \mathcal{A}(x)$ and $t>0$, 
applying Itô's formula to $\mathrm{e}^{-\delta t}v_\epsilon(X^{x;u}(t))$ gives
	\begin{align}
		\mathrm{e}^{-\delta t}v_\epsilon(X^{x;u}(t))-	v_\epsilon(x)&=-\int_0^t\delta \mathrm{e}^{-\delta s}v_\epsilon(X^{x;u}(s)) \md s+ \int_0^t\mathrm{e}^{-\delta s}v'_\epsilon(X^{x;u}(s))\mathrm{d} X^{x;u}(s)\nonumber\\
		&\quad\quad+ \frac12\int_0^t\mathrm{e}^{-\delta s}\sigma^2\pi^2(s)(X^{x;u}(s))^2v''_\epsilon(X^{x;u}(s))\md s.\label{eq:ito:epsilon}
	\end{align}
	After rearrangement, we obtain
	\begin{align*}
			&\mathrm{e}^{-\delta t}v_\epsilon(X^{x;u}(t))+ \int_0^t \mathrm{e}^{-\delta s}(U_1(b(s))+U_2(X^{x;u}(s)))\md s\\=&	v_\epsilon(x)+\int_0^t \mathrm{e}^{-\delta s}\bigg[-\delta v_\epsilon(X^{x;u}(s)) +\left(rX^{x;u}(s)+\sigma \theta\pi(s)X^{x;u}(s)-b(s)\right)v'_\epsilon(X^{x;u}(s))\\
			&\quad\quad\quad\quad\quad\quad\quad+\frac12\sigma^2\pi^2(s)(X^{x;u}(s))^2 v''_\epsilon(X^{x;u}(s))\bigg]\md s\\
			&\quad +\int_0^t \mathrm{e}^{-\delta s}(U_1(b(s))+U_2(X^{x;u}(s)))\md s+\int_0^t\mathrm{e}^{-\delta s}v'_\epsilon(X^{x;u}(s))\sigma\pi(s)X^{x;u}(s)\mathrm{d}W(s)\\
			\leq & v_\epsilon(x)+\int_0^t\mathrm{e}^{-\delta s}g_\epsilon(X^{x;u}(s))\md s+\int_0^t\mathrm{e}^{-\delta s}v'_\epsilon(X^{x;u}(s))\sigma\pi(s)X^{x;u}(s)\mathrm{d}W(s).
	\end{align*}
	Because $v$, $U_1$, $U_2$ are non-negative and $v_\epsilon\to v$ uniformly on $[\underline{x},\infty)$ as $\epsilon\to0$ (which implies that $v_\epsilon$ is bounded from below),  $\int_0^{\cdot}\mathrm{e}^{-\delta s}v'_\epsilon(X^{x;u}(s))\sigma\pi(s)X^{x;u}(s)\mathrm{d}W(s)$ is a supermartingale and thus
	\begin{align*}
		&\mE\left[\mathrm{e}^{-\delta t}v_\epsilon(X^{x;u}(t))+ \int_0^t \mathrm{e}^{-\delta s}(U_1(b(s))+U_2(X^{x;u}(s)))\md s\right]\leq v_{\epsilon}(x)+\mE\left[\int_0^t\mathrm{e}^{-\delta s}g_\epsilon(X^{x;u}(s))\md s\right].
	\end{align*}
    Because $v_\epsilon=v$ and $g_\epsilon=0$ on $(\underline{x}+\epsilon,\infty)$,
	letting $\epsilon\to0$ and using the dominated convergence theorem, we have
    	\begin{align*}
		&\mE\left[\mathrm{e}^{-\delta t}v(X^{x;u}(t))+ \int_0^t \mathrm{e}^{-\delta s}(U_1(b(s))+U_2(X^{x;u}(s)))\md s\right]\leq v(x).
	\end{align*}
  Because $v$, $U_1$ and $U_2$ are nonnegative, letting $t\to\infty$ and using  the monotone convergence theorem,  we get 
  \begin{align*}
		&\mE\left[\int_0^\infty \mathrm{e}^{-\delta s}(U_1(b(s))+U_2(X^{x;u}(s)))\md s\right]\leq v(x).
	\end{align*}
As $u=(b,\pi)$ is arbitrary, we conclude that  $V(x)\leq v(x)$. 
	
	Next, we prove that $(b^*(X^{x;*}),\pi^*(X^{x;*})) $ is optimal for any $x$. Note that   under the condition in Theorem \ref{thm:verification}(iii), for any $p>1$, \begin{align}\label{eq:moment:}
\mE\Big[\sup\limits_{s\in[0,t]}|X^{x;*}(s)|^p\Big]<\infty\quad \text{for all } t>0.
\end{align}
        Using  Assumption \ref{assump:u:1}(ii), together with our hypothesis that $I(v')$ and $\frac{v'}{v''}$ are of linear growth, we take expectations on both sides of \eqref{eq:ito:epsilon} (with $(b,\pi,X^{x;u})$ replaced by $(b^*(X^{x;*}),\pi^*(X^{x;*}),X^{x;*})$) to obtain
	\begin{align}\label{eq:verify:ito:*}
		v_\epsilon(x)\!=&\mE \left\{\int_0^t \mathrm{e}^{-\delta s}\left[ -\left(rX^{x;*}(s)+\sigma \theta\pi^*(X^{x;*}(s))X^{x;*}(s)-b^*(X^{x;*}(s))\right)v'_\epsilon(X^{x;*}(s))\right.\right.\nonumber\\
		&\left.\left.+\delta v_\epsilon(X^{x;*}(s))\!-\!\frac12\sigma^2(\pi^*(X^{x;*}(s)))^2(X^{x;*}(s))^2 v''_\epsilon(X^{x;*}(s))\!\right]\md s\!\right\}\!+\!	\mE\!\left[\mathrm{e}^{-\delta t}v_\epsilon(X^{x;*}(t))\!\right].
	\end{align}
Note that
\begin{align}\label{eq:dct:pi}
    (\pi^*(y))^2y^2v''_{\epsilon}(y)=\frac{\theta^2}{\sigma^2}v'(y)\frac{v'(y)}{v''(y)}\frac{v''_{\epsilon}(y)}{v''(y)}=\left\{\begin{aligned}
        &\frac{\theta^2}{\sigma^2}v'(y)\frac{v'(y)}{v''(y)},&&y\in(\underline{x}+\epsilon,\infty),\\
        &\frac{\theta^2}{\sigma^2}[v'(y)]^2\frac{v''(\underline{x}+\epsilon)}{[v''(y)]^2},&&y\in[\underline{x},\underline{x}+\epsilon].
    \end{aligned}\right.
\end{align}
Define
\begin{align*}
    C_{\underline{x}}:=\frac14\theta^2(v'(\underline{x}))^2\frac{1}{r\underline{x}-I(v'(\underline{x}))}>0.
\end{align*}
By \eqref{eq:verify:limit},  there exists $\epsilon_0>0$ such that
\begin{align*}
   C_{\underline{x}}/2\leq (y-\underline{x})[v''(y)]^2\leq 2C_{\underline{x}}\quad \text{for all } y\in (\underline{x},\underline{x}+\epsilon_0].
\end{align*}
Consequently,
\begin{align}    \label{eq:for:something}\sup_{\epsilon\in(0,\epsilon_0)}\sup_{y\in(\underline{x},\underline{x}+\epsilon]}\frac{|v''(\underline{x}+\epsilon)|}{[v''(y)]^2}\leq \sup_{\epsilon\in(0,\epsilon_0)}\sup_{y\in(\underline{x},\underline{x}+\epsilon]}\frac{\sqrt{\frac{2C_{\underline{x}}}{\epsilon}}}{\frac{C_{\underline{x}}/2}{(y-\underline{x})}}\leq \sqrt{\frac{8\epsilon_0}{C_{\underline{x}}}}.
\end{align}
 Combining \eqref{eq:for:something} with \eqref{eq:dct:pi} and the hypothesis that $\frac{v'}{v''}$ has linear growth, we then apply the dominated convergence theorem to  let $\epsilon\to0$, yielding
  \begin{align*}
   &\mE \left\{\int_0^t \mathrm{e}^{-\delta s}\bigg[\frac12\sigma^2(\pi^*(X^{x;*}(s)))^2(X^{x;*}(s))^2 v_{\epsilon}''(X^{x;*}(s))\bigg]\md s \right\}\\
   \to &\mE\left\{ \int_0^t \mathrm{e}^{-\delta s}\bigg[\frac12\sigma^2(\pi^*(X^{x;*}(s)))^2(X^{x;*}(s))^2 v''(X^{x;*}(s))\bigg]\md s\right\}.
  \end{align*}
 The remaining terms in \eqref{eq:verify:ito:*} can be treated in a similar manner. Applying the dominated convergence theorem and letting $\epsilon\to0$, we obtain
  	\begin{align}\label{eq:verify:ito:*:v}
		v(x)=&\mE \left\{\int_0^t \mathrm{e}^{-\delta s}\bigg[ -\left(rX^{x;*}(s)+\sigma \theta\pi^*(X^{x;*}(s))X^{x;*}(s)-b^*(X^{x;*}(s))\right)v'(X^{x;*}(s))\right.\nonumber\\
		&\left.+\delta v(X^{x;*}(s))-\frac12\sigma^2(\pi^*(X^{x;*}(s)))^2(X^{x;*}(s))^2 v''(X^{x;*}(s))\bigg]\md s\right\}+	\mE\left[\mathrm{e}^{-\delta t}v(X^{x;*}(t))\right]\nonumber\\
        =&\mE\left[\mathrm{e}^{-\delta t}v(X^{x;*}(t))+ \int_0^t \mathrm{e}^{-\delta s}\bigl(U_1(b^*(X^{x;*}(s)))+U_2(X^{x;*}(s))\bigr)\md s\right],
	\end{align}
 where the second equality follows from the fact that $v$ satisfies \eqref{eq:hjb}.

Moreover, by Lemma~\ref{lma:X^gamma}, we have
    \begin{align*}
        \lim_{t\to \infty} \mE\left[\mathrm{e}^{-\delta t}|X^{x;*}(t)|^{\gamma}\right]\leq   \lim_{t\to \infty}x^{\gamma}\exp\left(\left(-\delta+\gamma r+\frac12\frac{\gamma\theta^2}{1-\gamma}\right)t\right)=0.
    \end{align*}
Therefore, using the growth bound $v(x)\le C_{v,0} x^\gamma$, it follows that $\lim\limits_{t\to\infty}\mE\left[\mathrm{e}^{-\delta t}v(X^{x;*}(t))\right]=0$.
	Finally,  letting $t\to\infty$ in \eqref{eq:verify:ito:*:v}, the conclusion follows immediately.

\subsection{Proof of Theorem \ref{thm:verification2}}\label{appendix:sec:veri2}
   Fix $x>\underline{x}$ and let $u=(b,\pi)\in\mathcal{A}(x)$. 
For any $\epsilon\in\left(0,\frac12(x-\underline{x})\right)$, define
    \begin{align*}
       \tau:=\inf\{s>0:X^{x;u}(s)=\underline{x}\},\quad \tau_\epsilon:=\inf\{s>0:X^{x;u}(s)=\underline{x}+\epsilon\}.
    \end{align*}
   Then $\tau_\epsilon\uparrow \tau$ a.s..
Applying Itô's formula as in \eqref{eq:ito:epsilon} and using that $v\in C^2((\underline{x},\infty))$ satisfies \eqref{eq:hjb}, we obtain
	\begin{align*}
			&\mathrm{e}^{-\delta (t\wedge \tau_\epsilon)}v(X^{x;u}(t\wedge \tau_\epsilon))+ \int_0^{t\wedge \tau_\epsilon} \mathrm{e}^{-\delta s}(U_1(b(s))+U_2(X^{x;u}(s)))\md s\\
			\leq & v(x)+\int_0^{t\wedge \tau_\epsilon}\mathrm{e}^{-\delta s}v'(X^{x;u}(s))\sigma\pi(s)X^{x;u}(s)\mathrm{d}W(s),
	\end{align*}
    	Because $v$, $U_1$, and $U_2$ are non-negative, the stochastic integral is a supermartingale. 
Taking expectations gives\begin{align}\label{eq:case2:ito}
		&\mE\left[\mathrm{e}^{-\delta (t\wedge \tau_\epsilon)}v(X^{x;u}(t\wedge \tau_\epsilon))+ \int_0^{t\wedge \tau_\epsilon} \mathrm{e}^{-\delta s}(U_1(b(s))+U_2(X^{x;u}(s)))\md s\right]\leq v(x).
	\end{align}
    By Fatou's lemma,
    \begin{align*}
        \liminf_{\epsilon \to 0} \mE\left[\mathbb{I}_{\tau<\infty}\mathrm{e}^{-\delta (t\wedge \tau_\epsilon)}v(X^{x;u}(t\wedge \tau_\epsilon))\right]\geq \mE\left[\mathbb{I}_{\tau<\infty}\mathrm{e}^{-\delta (t\wedge\tau)}v(X^{x;u}(t\wedge\tau))\right].
    \end{align*}
    Since $v$, $U_1$, and $U_2$ are non-negative, letting $\epsilon\to0$ and applying the monotone convergence theorem to \eqref{eq:case2:ito}, we obtain
    \begin{align}
        \mE\left[\mathbb{I}_{\tau<\infty}\mathrm{e}^{-\delta (t\wedge \tau)}v(X^{x;u}(t\wedge \tau))+ \int_0^{t\wedge \tau} \mathrm{e}^{-\delta s}(U_1(b(s))+U_2(X^{x;u}(s)))\md s\right]\leq v(x).
        \label{eq:verify:2:lll}
    \end{align} 
    Applying Fatou's lemma again,
    \begin{align*}
      \liminf_{t\to \infty}\mE\left[\mathbb{I}_{\tau<\infty}\mathrm{e}^{-\delta (t\wedge\tau)}v(X^{x;u}(t\wedge\tau))\right]\geq \mE\left[\mathbb{I}_{\tau<\infty}\mathrm{e}^{-\delta \tau}v(X^{x;u}(\tau))\right].  
    \end{align*}
   Letting $t\to\infty$ in \eqref{eq:verify:2:lll} and using the monotone convergence theorem yields 
  \begin{align*}
		&\mE\left[\int_0^{\tau} \mathrm{e}^{-\delta s}(U_1(b(s))+U_2(X^{x;u}(s)))\md s+\mathbb{I}_{\tau<\infty}\mathrm{e}^{-\delta \tau}v(X^{x;u}(\tau))\right]\leq v(x).
	\end{align*}
  On the event $\{\tau<\infty\}$, we have $X^{x;u}(\tau)=\underline{x}$. 
By Lemma \ref{lma:stop}, $b(\tau+s)=b_0$ and $X^{x;u}(\tau+s)=\underline{x}$  for a.e. $s\geq 0$, a.s.. Therefore,
    \begin{align*}
        \mathbb{I}_{\tau<\infty}\mathrm{e}^{-\delta \tau}v(X^{x;u}(\tau))&\ =\mathbb{I}_{\tau<\infty}\mathrm{e}^{-\delta \tau}v(\underline{x})=\mathbb{I}_{\tau<\infty}\mathrm{e}^{-\delta \tau}V(\underline{x})\\
        &\ =\mathbb{I}_{\tau<\infty}\mathrm{e}^{-\delta \tau}\int_\tau^\infty \mathrm{e}^{-\delta (s-\tau)}\left(U_1(b_0)+U_2(\underline{x})\right)\md s\\
       &\  =\mathbb{I}_{\tau<\infty}\int_\tau^\infty \mathrm{e}^{-\delta s}\left(U_1(b(s))+U_2(X^{x;u}(s))\right)\md s.
    \end{align*}
Consequently,
    \begin{align*}
       &\mE\left[\int_0^{\infty} \mathrm{e}^{-\delta s}(U_1(b(s))+U_2(X^{x;u}(s)))\md s\right]\leq v(x). 
    \end{align*}
As $u=(b,\pi)$ is arbitrary, we conclude that  $V(x)\leq v(x)$. 

We next verify the optimality of $(b^*,\pi^*)$. 

Repeating the above argument and using that $(b^*,\pi^*)$ attains the supremum in the HJB equation, we obtain
\begin{align*}
    &\mathrm{e}^{-\delta (t\wedge \tau_\epsilon)}v(X^{x;*}(t\wedge \tau_\epsilon))+ \int_0^{t\wedge \tau_\epsilon} \mathrm{e}^{-\delta s}(U_1(b^*(X^{x;*}(s)))+U_2(X^{x;*}(s)))\md s\\
    =&v(x)+\int_0^{t\wedge \tau_\epsilon}\mathrm{e}^{-\delta s}v'(X^{x;*}(s))\sigma\pi^*(X^{x;*}(s))X^{x;*}(s)\mathrm{d}W(s).
\end{align*}
Under the linear growth condition in (iii) of Theorem \ref{thm:verification2}, \eqref{eq:moment:} holds for any $p \ge 1$.
Thus, $\eqref{eq:case2:ito}$ holds with equality when $u$ is replaced by $(b^*,\pi^*)$, i.e.
\begin{align*}
		&\mE\left[\mathrm{e}^{-\delta (t\wedge \tau_\epsilon)}v(X^{x;*}(t\wedge \tau_\epsilon))+ \int_0^{t\wedge \tau_\epsilon} \mathrm{e}^{-\delta s}(U_1(b^*(X^{x;*}(s)))+U_2(X^{x;*}(s)))\md s\right]= v(x).
	\end{align*}
Letting $\epsilon\to0$ and applying the dominated convergence theorem, we obtain
\begin{align}\label{eq:veri1}
		&\mE\left[\mathrm{e}^{-\delta (t\wedge \tau)}v(X^{x;*}(t\wedge \tau))+ \int_0^{t\wedge \tau} \mathrm{e}^{-\delta s}(U_1(b^*(X^{x;*}(s)))+U_2(X^{x;*}(s)))\md s\right]= v(x).
	\end{align}
    Note that
    \begin{align}\label{eq:left}
       \mE\left[\mathrm{e}^{-\delta (t\wedge \tau)}v(X^{x;*}(t\wedge \tau))\right]= \mE\left[\mathrm{e}^{-\delta t}v(X^{x;*}(t))\mathbbm{1}_{\tau>t}+\mathrm{e}^{-\delta \tau}v(X^{x;*}(\tau))\mathbbm{1}_{\tau\leq t}\right].
    \end{align}
  The first term on the right-hand side of \eqref{eq:left} converges to $0$ as $t \to \infty$, for the same reason as in the proof of Theorem~\ref{thm:verification}. The second term on the right-hand side of \eqref{eq:left} converges, as $t \to \infty$, to $\mE\left[\mathrm{e}^{-\delta \tau}v(X^{x;*}(\tau))\mathbbm{1}_{\tau<\infty}\right]$, by the monotone convergence theorem .

Letting $t\to\infty$ in \eqref{eq:veri1}, we get
\begin{align*}
    v(x)=&\mE\left[\mathrm{e}^{-\delta \tau}v(X^{x;*}( \tau))\mathbbm{1}_{\tau<\infty}+ \int_0^{ \tau} \mathrm{e}^{-\delta s}(U_1(b^*(X^{x;*}(s)))+U_2(X^{x;*}(s)))\md s\right]\\
    =&\mE\left[\int_0^{ \infty} \mathrm{e}^{-\delta s}(U_1(b^*(X^{x;*}(s)))+U_2(X^{x;*}(s)))\md s\right].
\end{align*}
Thus, $v=V$ and $(b^*,\pi^*)$ is optimal.

\section{Proofs for the Statements in Section~\ref{sec:well:hjb}}
\subsection{Proof of Lemma \ref{lma:local:tildeh}}
\label{sec:proof:local:tildeh}
Fix $y\in(0,z_1)$.
Equation \eqref{eq:extension:u_22} can be rewritten as
\begin{equation}\label{eq:H:system}
\left\{
\begin{aligned}
    &H''(h)=\Phi(h,H(h),H'(h)),\quad h\in(y,\infty),\\
    &H(y)=\underline{x}, \quad H'(y)=0,
\end{aligned}
\right.
\end{equation}
    where $\Phi(\cdot):(0,\infty)\times \mR^2\to \mR$ is defined by
    \begin{align}\label{eq:PHI}
        \Phi(h,p,q):=-\frac{2(r-\delta)}{\theta^2h}{q}-\frac{2}{\theta^2}\tilde{U}_2'({p}){q}+\frac{2r}{\theta^2h^2}{p}-\frac{2}{\theta^2h^2}\tilde{I}(h).
    \end{align}
    Introducing $P(h)=H(h)$ and ${Q}(h)=H'(h)$, the equation  \eqref{eq:H:system} is equivalent to the first-order system
    \begin{align}\label{eq:ode:pq}
        \begin{cases}
            &{P}'(h)={Q}(h) \,\, \,\text{and}\,\,\, {Q}'(h)=\Phi(h,{P}(h),{Q}(h)), \quad h\in(y,\infty),\\
            &{P}(y)=\underline{x},\quad {Q}(y)=0.
        \end{cases}
    \end{align}
   The function $\Phi$ is  locally Lipschitz continuous in $(h,p,q)$. Hence, $\Phi$ is  Lipschitz continuous on $[y,M]\times K$ for any $M>y$ and any compact set $K\subset\mR^2$ containing $(\underline{x},0)$.
 Therefore, system \eqref{eq:ode:pq} admits a unique local solution $({P},{Q})$ on $[y,y+\epsilon(y))$ for some $\epsilon(y)>0$ (see, e.g., \cite*[Theorem 2.2]{teschl2012ordinary}). Consequently, $H(\cdot;y)$ is uniquely determined and belongs to $C^2([y,y+\epsilon(y)))$.

 \subsection{Proof of Lemma \ref{lma:represent}}
 \label{sec:proof:lma:represent}
 \textit{Proof of \textup{(i)}.}
Fix $y\in(0,z_1)$. For notational simplicity, we suppress the
 dependence on $y$ and write $H$ for $H(\cdot;y)$. Rearranging
\eqref{eq:extension:u_22} by moving
$\tilde{I}(h)+\tilde{U}_2'(H(h))h^2H'(h)$ to the right-hand side,
we obtain
 \begin{align}\label{eq:ode:tilde:y:kernel}
\left\{
\begin{aligned}
		&-\frac{1}{2} \theta^2 h^2H''(h) -\left(r-\delta\right) hH'(h) +r H(h) =\tilde{I}(h)+\tilde{U}_2'(H(h))h^2H'(h),\quad h\in (y,M_y),\\
        &H\left(y\right)=\underline{x}, \quad H'\left(y\right)=0.
        \end{aligned}
        \right.
\end{align}
To derive a representation for $H$, we regard the right-hand
side of \eqref{eq:ode:tilde:y:kernel}, evaluated along the fixed
solution, as a prescribed forcing term. The left-hand side of \eqref{eq:ode:tilde:y:kernel} is then a
second-order linear Cauchy--Euler operator. The resulting linear initial-value problem admits a unique solution (see, e.g., \cite*[Theorem 2.2]{teschl2012ordinary}), which can be obtained by variation of parameters using the two fundamental solutions of the associated homogeneous equation.
Omitting the routine calculations,
we obtain
\begin{align}\label{eq:H:old}
&H(h;y)\nonumber
       \\ =~&\left\{\frac{p_2\underline{x}}{p_1+p_2}y^{-p_1}-\frac{2}{\theta^2(p_1+p_2)}\int_{y}^{h}\left[\tilde{I}(z)+\tilde{U}_2'(H(z;y))z^2H'(z;y)\right]z^{-p_1-1}\md z \right\} h^{p_1}\nonumber\\
     &+\left\{\frac{p_1\underline{x}}{p_1+p_2}y^{p_2}+\frac{2}{\theta^2(p_1+p_2)}\int_{y}^{h}\left[\tilde{I}(z)+\tilde{U}_2'(H(z;y))z^2H'(z;y)\right]z^{p_2-1}\md z\right\} h^{-p_2},
\end{align}
where $p_1>0$ and $q_1=-p_2<0$ are the two roots of the
characteristic quadratic equation \eqref{eq:p}. These roots satisfy
$p_1p_2=\frac{2r}{\theta^2}$; hence, \eqref{eq:H:old} can be
equivalently rewritten as \eqref{eq:tildeH}.
Formula \eqref{eq:H:old} can also be verified directly. First, evaluating at $h=y$ yields
$H(y)=\underline{x}$ and $H'(y)=0$, so the initial conditions are satisfied.
Second, differentiating \eqref{eq:H:old} twice with respect to $h$ and substituting into the
left-hand side of \eqref{eq:ode:tilde:y:kernel} shows that the ODE is
identically satisfied for all $h\in(y,M_y)$. By the uniqueness of solutions
to linear Cauchy--Euler equations with prescribed initial conditions,
\eqref{eq:H:old} is the desired formula.

\medskip
\noindent{\textit{Proof of \textup{(ii)}.}}
For notational simplicity, we also write $\bar H$ for
$\bar H(\cdot;y)$.

For any $h\in[y,\phi(y))$, by \eqref{eq:H:old}, we have
    \begin{align*}
    H(h)=~&\frac{p_2\underline{x}}{p_1+p_2}y^{-p_1}h^{p_1}+\frac{p_1\underline{x}}{p_1+p_2}y^{p_2}h^{-p_2}\\
    &-\frac{2}{\theta^2(p_1+p_2)}\int_{y}^{h}\left[\tilde{I}(z)+\tilde{U}_2'(H(z))z^2H'(z)\right]\left[\left(\frac{h}{z}\right)^{p_1}-\left(\frac{h}{z}\right)^{-p_2}\right]z^{-1}\md z \\
    \le~&\frac{p_2\underline{x}}{p_1+p_2}y^{-p_1}h^{p_1}+\frac{p_1\underline{x}}{p_1+p_2}y^{p_2}h^{-p_2}
    \\
    =~& \frac{\underline{x}}{p_1+p_2}h^{p_1}\left(p_2y^{-p_1}+p_1y^{p_2}h^{-p_1-p_2}\right)\\
    \le~& \underline{x} h^{p_1}y^{-p_1}:=\bar H(h), 
    \end{align*}
    where the first inequality uses that $\tilde{U}_2'(x)\ge 0$ for $x\ge \underline{x}$ and $H'(h)\ge 0$ on $[y,\phi(y))$, and the second inequality follows from $h^{-p_1-p_2} \le y^{-p_1-p_2}$ for all $h\ge y$.

For $h\in[\phi(y),M_y)$, we have $\bar H'(h)>0$ while $H'(h)\le 0$.
It then follows that $H(h)\le  \bar H(h)$ for all $h\in[y,M_y)$.

\subsection{Proof of Proposition \ref{lma:global}}
\label{sec:proof:lma:global}
Fix $y \in (0,z_1)$. It suffices to show that for any $M>y$, \eqref{eq:extension:u_22} admits a unique solution $H(\cdot;y)\in C^2([y,M])$. 
By Lemma~\ref{lma:represent}(ii), $H(\cdot;y)$ is bounded above, and thus $\tilde{U}_2'(H)$ is also bounded on $[y,M]$. Since $y>0$, there exists a constant $C_M>0$
such that
\begin{align*}
&\left|-\frac{2(r-\delta)}{\theta^2h}
-\frac{2}{\theta^2}\tilde{U}_2'(H(h))\right|,
\qquad
\left|\frac{2r}{\theta^2h^2}\right|,
\qquad
\left|-\frac{2}{\theta^2h^2}\tilde{I}(h)\right|
\leq C_M
\end{align*}
for all $h\in[y,M]$.
The ODE system \eqref{eq:ode:pq} can be rewritten as
\begin{equation}\label{eq:system}
\left\{
\begin{aligned}
    P'(h)&=Q(h),\\
    Q'(h)&=\left[-\frac{2(r-\delta)}{\theta^2h}-\frac{2}{\theta^2}\tilde{U}_2'(H(h))\right]Q(h)
          +\frac{2r}{\theta^2h^2}P(h)-\frac{2}{\theta^2h^2}\tilde{I}(h), \quad h\in[y,M],\\
    P(y)&=\underline{x}, \quad Q(y)=0.
\end{aligned}
\right.
\end{equation}
From ${P}'(h) = {Q}(h)$, we have
\begin{align*}
    |{P}(h)|\leq \underline{x}+\int_y^h|{Q}(z)|\md z\quad  \text{for all } h\in[y,M].
\end{align*}
Similarly, from $Q'(h)=[-\frac{2(r-\delta)}{\theta^2h}-\frac{2}{\theta^2}\tilde{U}_2'(H(h))]Q(h)
          +\frac{2r}{\theta^2h^2}P(h)-\frac{2}{\theta^2h^2}\tilde{I}(h)$, it follows that
\begin{align*}
    |{Q}(h)|\leq C_M\int_y^h (|P(z)|+|Q(z)|)\md z+C_M(M-y)\quad  \text{for all } h\in[y,M]. 
\end{align*}
Combining these estimates, we obtain
\begin{align*}
    |{P}(h)|+|{Q}(h)|\leq \underline{x}+C_M(M-y) +(C_M+1)\int_y^h (|P(z)|+|Q(z)|)\md z\quad  \text{for all } h\in[y,M].
\end{align*}
By Grönwall’s inequality, $|{P}(h)|+|{Q}(h)|$ is bounded on $[y,M]$, which implies that neither ${P}(h)$ nor ${Q}(h)$ can blow up on this interval. Therefore, the solution exists throughout $[y,M]$; see, e.g., \cite*[Corollary 2.15]{teschl2012ordinary}.

Since $M>y$ is arbitrary, it follows that the local solution can be extended to $[y,\infty)$. Uniqueness then follows directly from the locally Lipschitz property of the ODE system.

\subsection{Proof of Lemma~\ref{lma:G}}
\label{sec:proof:lma:G}
\textit{Proof of part~\textup{(i)}.}
By \eqref{eq:tildeI:growth} and
$p_1>\frac{1}{1-\gamma}$, for every $y\in(0,z_0)$,
\begin{align*}
		\int_{y}^{\infty}\tilde{I}(z)z^{-p_1-1}\md z \le  \ &M_{1,1}\int_{z_0}^{\infty}z^{\frac{1}{1-\gamma}-p_1-1}\md z+b_0\int_y^{z_0}z^{-p_1-1}\md z\\
		=\ &\frac{ M_{1,1}}{p_1-\frac{1}{1-\gamma}}z_0^{\frac{1}{1-\gamma}-p_1}+b_0\frac{y^{-p_1}-z_0^{-p_1}}{p_1}.
\end{align*}
Consequently, by the definition of $\bar{G}$,
	\begin{align}\label{eq:barG:lowerbound}
		\bar{G}(y)\ge r\underline{x}-\frac{ M_{1,1}p_1}{p_1-\frac{1}{1-\gamma}}z_0^{\frac{1}{1-\gamma}-p_1}y^{p_1}-b_0+b_0z_0^{-p_1}y^{p_1}.
	\end{align}
It follows that $\liminf\limits_{y\downarrow0}\bar{G}(y)\geq r\underline{x}-b_0>0$.
 
Moreover, since $\tilde{I}(z) > \tilde{I}(z_1) = r\underline{x}$ 
for $z>z_1$, we have
	\begin{align*}
		\bar{G}\left(z_1\right)
		&< r\underline{x}\left[1-p_1\int_{z_1}^{\infty}z^{-p_1-1}\md z\,z_1^{p_1}\right]=0.
	\end{align*} 
	
By the dominated convergence theorem, $\bar{G}(\cdot)$ is continuous on $(0,z_1]$.
Since $\liminf\limits_{y\downarrow0}\bar{G}(y)>0$ and $\bar{G}(z_1)<0$, the intermediate value theorem yields some
$\bar{y}\in(0,z_1)$ such that $\bar{G}(\bar{y})=0$.

Uniqueness follows by the change of variables $z=yt$:
\[
\bar G(y)=r\underline{x}
-p_1y^{p_1}\int_y^\infty \tilde I(z)z^{-p_1-1}\md z
=r\underline{x}-p_1\int_1^\infty \tilde I(yt)t^{-p_1-1}\md t.
\]
Since $\tilde I$ is nondecreasing on $(0,\infty)$ and strictly increasing on
$(z_0,\infty)$, $\bar G$ is strictly decreasing, and its root is unique.

\medskip
\noindent\textit{Proof of part~\textup{(ii)}.}
We first show that $G$ is well-defined. To this end, it suffices to verify that
\begin{align*}  \int_{y}^{\infty}\left|\tilde{U}_2'(H(z;y))z^2H'(z;y)\right|z^{-p_1-1}\md z<\infty\quad \text{for all } y\in(0,\bar{y}].
\end{align*}
We begin with the interval $(y,\phi(y))$. For any $h\in(y,\phi(y))$, since $\tilde{U}_2'(x)\ge0$ and $H'(z;y)>0$ for all $z\in(y,h)$, we have
\begin{align*}
  &\int_{y}^{h}\left|\tilde{U}_2'(H(z;y))z^2H'(z;y)\right|z^{-p_1-1}\md z=\int_{y}^{h}\tilde{U}_2'(H(z;y))H'(z;y)z^{-p_1+1}\md z\\
=~&\tilde{U}_2(H(h;y))h^{-p_1+1}-\tilde{U}_2(H(y;y))y^{-p_1+1}+(p_1-1)\int_{y}^{h}\tilde{U}_2(H(z;y))z^{-p_1}\md z.
\end{align*}
Using the growth and monotonicity conditions in
Assumption~\ref{assmup:U12}(ii), the extension properties in
\eqref{eq:tildeU2:properties}, and the upper bound in
Lemma~\ref{lma:represent}(ii), we have
\begin{align*}
    0\leq \tilde{U}_2(H(z;y))\leq \tilde{U}_2(\bar H(z;y)) \leq C_{2,1}(\bar H(z;y))^\gamma\leq C_{2,1}\left(\underline{x} y^{-p_1}\right)^\gamma z^{p_1\gamma}\quad \text{for all }z\in[y,\infty).
\end{align*}
Since $p_1>\frac{1}{1-\gamma}$, it follows that
\begin{align*}
   \lim_{h\uparrow \phi(y)}\tilde{U}_2(H(h;y))h^{-p_1+1} =\begin{cases}
       \tilde{U}_2(H(\phi(y);y))(\phi(y))^{-p_1+1},\quad &\phi(y)<\infty,\\
       0,\quad &\phi(y)=\infty.
   \end{cases}
\end{align*}
Moreover,
\begin{align}\label{eq:bound:u2h}
    \int_{y}^{h}\tilde{U}_2(H(z;y))z^{-p_1}\md z\leq  \int_{y}^{\infty}\tilde{U}_2(H(z;y))z^{-p_1}\md z\leq \frac{C_{2,1}}{(1-\gamma)p_1-1}\underline{x}^\gamma y^{1-p_1}.
\end{align}
By the monotone convergence theorem,
\begin{equation}\label{eq:G:ibp:increasing}
\begin{aligned}
    \ & \int_{y}^{\phi(y)}\left|\tilde{U}_2'(H(z;y))z^2H'(z;y)\right|z^{-p_1-1}\md z= \int_{y}^{\phi(y)}\tilde{U}_2'(H(z;y))z^2H'(z;y)z^{-p_1-1}\md z\\
  =\ &\lim_{h\uparrow \phi(y)}\tilde{U}_2(H(h;y))h^{-p_1+1}-\tilde{U}_2(H(y;y))y^{-p_1+1}+(p_1-1)\int_{y}^{\phi(y)}\tilde{U}_2(H(z;y))z^{-p_1}\md z <\infty.
\end{aligned}
\end{equation}
If $\phi(y)=\infty$, the preceding estimate already establishes the desired integrability. Next, if $\phi(y)<\infty$, then $H'(z;y)<0$ for all $z\in(\phi(y),\infty)$, and a similar argument yields
\begin{equation}\label{eq:G:ibp:decreasing}
\begin{aligned}
  \ &\int_{\phi(y)}^{\infty}\left|\tilde{U}_2'(H(z;y))z^2H'(z;y)\right|z^{-p_1-1}\md z=-\int_{\phi(y)}^{\infty}\tilde{U}_2'(H(z;y))z^2H'(z;y)z^{-p_1-1}\md z\\
  =\ &\tilde{U}_2(H(\phi(y);y))(\phi(y))^{-p_1+1}-(p_1-1)\int_{\phi(y)}^{\infty}\tilde{U}_2(H(z;y))z^{-p_1}\md z<\infty.
\end{aligned}
\end{equation}
Therefore, $G$ is well-defined. Moreover, using \eqref{eq:G:ibp:increasing} and \eqref{eq:G:ibp:decreasing}, it can be rewritten as
\begin{align*}
    G(y)=r\underline{x}-p_1\int_{y}^{\infty}\tilde{I}(z)z^{-p_1-1}\md z\,y^{p_1}+p_1\tilde{U}_2(\underline{x})y-p_1(p_1-1)\int_y^{\infty}\tilde{U}_2(H(z;y))z^{-p_1}\md z\,y^{p_1}.
\end{align*}
Combining the estimate \eqref{eq:barG:lowerbound} with \eqref{eq:bound:u2h}, we obtain
\begin{align*}
    G(y)\ge r\underline{x}-\frac{ M_{1,1}p_1}{p_1-\frac{1}{1-\gamma}}z_0^{\frac{1}{1-\gamma}-p_1}y^{p_1}-b_0+b_0z_0^{-p_1}y^{p_1}+p_1\tilde{U}_2(\underline{x})y-\frac{C_{2,1}p_1(p_1-1)}{(1-\gamma)p_1-1}\underline{x}^\gamma y.
\end{align*}
Hence,
$\liminf\limits_{y\downarrow0}G(y)\geq r\underline{x}-b_0>0$. Finally, since $\Phi$ defined in \eqref{eq:PHI}  is locally Lipschitz, the solution of \eqref{eq:ode:pq} (equivalently, \eqref{eq:extension:u_22}) depends continuously on the initial condition and parameters (see \cite*[Theorem 2.9]{teschl2012ordinary}). Hence, both $H(h;y)$ and $H'(h;y)$ are continuous on $\{(h,y):0<y\leq\bar y,\ h\geq y\}$. Therefore, by the dominated convergence theorem, $G$ is continuous on $(0,\bar{y}]$.

Since $\liminf\limits_{y\downarrow0}G(y)>0$, the set
\[
\mathcal{Y}:=\{0<y\leq\bar y:G(y)>0\}
\]
is nonempty. Hence $y_0:=\sup\mathcal{Y}$ is well-defined and belongs
to $(0,\bar y]$. Choose a sequence $\{y_n\}\subset\mathcal{Y}$ such that
$y_n\to y_0$. Since $G(y_n)>0$ for every $n$, the continuity of $G$
yields $G(y_0)=\lim\limits_{n\to\infty}G(y_n)\geq0$.

\subsection{Proof of Theorem \ref{thm:v}}\label{appendix:sec:case1:2}
Let $y_0$ and $H$ be as in Theorem~\ref{thm:H}, and set
$A:=\frac{1}{y_0}\in\left(U_1'(r\underline{x}),\infty\right)$.
Define 
\begin{align*}
    F(z):=H\left(\frac{1}{z}\right), \quad z\in \left(0,A\right].
\end{align*}
Then $F\in C^2(\left(0,A\right])$ satisfying \eqref{eq:F} and $F'(z)<0$ for any $z\in \left(0,A\right)$.
Next, let $f:[\underline{x},\infty)\to \left(0,A\right]$ be the inverse function of $F$ satisfying $F(f(x))=x$ for any  $x\in [\underline{x},\infty)$.  Then, $f\in C([\underline{x},\infty))\cap C^2((\underline{x},\infty))$ satisfying \eqref{eq:f}. Define
\begin{align}\label{eq:v:integral}
    v(x):=a+\int_{\underline{x}}^xf(z)\md z, \quad x\in[\underline{x},\infty),
\end{align}
where $a$ is a constant to be determined. Then, $v\in C^1([\underline{x},\infty))\cap C^3((\underline{x},\infty))$.
Moreover, since $f>0$ and $F'<0$, \eqref{eq:v:integral} gives
$v'(x)=f(x)>0$ for $x\in[\underline{x},\infty)$ and
$v''(x)=f'(x)=1/F'(f(x))<0$ for $x\in(\underline{x},\infty)$.

For any $x>\underline{x}$ and $\epsilon\in(0,x-\underline{x})$, 
integrating  both sides of the ODE in \eqref{eq:f} from $\underline{x}+\epsilon$ to $x$ yields
\begin{align}
    (\delta-r+\theta^2) (v(x)-v(\underline{x}+\epsilon))&= U_2(x)- U_2(\underline{x}+\epsilon)+\frac{1}{2} \theta^2 \int_{\underline{x}+\epsilon}^x\frac{\left(f(z)\right)^2 f''(z)}{\left(f'(z)\right)^2}\md z\nonumber \\
		&\quad \quad\quad +\int_{\underline{x}+\epsilon}^x\left[r z-I\left(f(z)\right)\right] f'(z)\md z.\label{eq:inte:v}
\end{align}
Using $v'=f$ from \eqref{eq:v:integral}, $D'=-I$ from
\eqref{eq:I} and \eqref{eq:D}, and integration by parts, we obtain the
following identities:
\begin{align*}  &(i)\quad \int_{\underline{x}+\epsilon}^x\frac{\left(f(z)\right)^2 f''(z)}{\left(f'(z)\right)^2}\md z=\int_{\underline{x}+\epsilon}^x\left(2f(z)-\frac{\md}{\md z}\frac{f^2(z)}{f'(z)}\right)\md z\\&\quad\quad\quad\quad\quad\quad\quad\quad\quad\quad\quad\,\,\,=2(v(x)-v(\underline{x}+\epsilon))-\frac{[v'(x)]^2}{v''(x)}+\frac{[v'(\underline{x}+\epsilon)]^2}{v''(\underline{x}+\epsilon)};\\
     &(ii)
     \quad \int_{\underline{x}+\epsilon}^x zf'(z)\md z=xv'(x)-(\underline{x}+\epsilon)v'(\underline{x}+\epsilon)-(v(x)-v(\underline{x}+\epsilon));\\  &(iii)\quad\int_{\underline{x}+\epsilon}^xI\left(f(z)\right) f'(z)\md z= -\int_{\underline{x}+\epsilon}^x \md D(f(z))=-D(v'(x))+D(v'(\underline{x}+\epsilon)).
\end{align*}
Substituting (i), (ii) and (iii) above into \eqref{eq:inte:v} yields
\begin{align}
     (\delta-r+\theta^2) (v(x)-v(\underline{x}&+\epsilon))= U_2(x)- U_2(\underline{x}+\epsilon) +rxv'(x)-r(\underline{x}+\epsilon)v'(\underline{x}+\epsilon) \nonumber\\ &\quad -r(v(x)-v(\underline{x}+\epsilon))+D(v'(x))-D(v'(\underline{x}+\epsilon))\nonumber\\
     &\quad +\frac{1}{2} \theta^2\left(2(v(x)-v(\underline{x}+\epsilon))-\frac{[v'(x)]^2}{v''(x)}+\frac{[v'(\underline{x}+\epsilon)]^2}{v''(\underline{x}+\epsilon)}\right).\label{eq:v:case2:epsilon}
\end{align}
By $F'(A)=0$ and the definitions of $f$ and $v$, as
$\epsilon\downarrow0$, we have $v(\underline{x}+\epsilon)\to a$,
$v'(\underline{x}+\epsilon)\to A$, and
$v''(\underline{x}+\epsilon)\to-\infty$. Hence, by the continuity of
$U_2$ and $D$, letting $\epsilon\downarrow0$ in
\eqref{eq:v:case2:epsilon} yields
\begin{align*}
     (\delta-r+\theta^2) (v(x)-a)&= U_2(x)- U_2(\underline{x})+\frac{1}{2} \theta^2\left(2(v(x)-a)-\frac{[v'(x)]^2}{v''(x)}\right) +rxv'(x)-r\underline{x}A \\ &\quad -r(v(x)-a)+D(v'(x))-D(A).
\end{align*}
Let $a:=\frac{1}{\delta}\left(U_2(\underline{x})+r\underline{x}A+D(A)\right)$, then $v$ satisfies \eqref{eq:ode:final}. 

It remains to establish the bounds in \eqref{eq:property:v}.
By the upper bound in \eqref{eq:H:property}, we have
\begin{align*}
    \frac{v'(x)}{x^{\gamma-1}}=\frac{f(x)}{x^{\gamma-1}}\overset{z=f(x)}{=}\frac{z}{[F(z)]^{\gamma-1}}=\left(\frac{z^{\frac{1}{\gamma-1}}}{F(z)}\right)^{\gamma-1}\overset{h=1/z}=\left(\frac{H(h)}{h^{\frac{1}{1-\gamma}}}\right)^{1-\gamma}\leq C_{H,0}^{1-\gamma}.
\end{align*}
Thus, we can take $C_{v,3}:=C_{H,0}^{1-\gamma}$. The existence of $C_{v,2}$ is a direct consequence of $C_{v,3}$ and we omit the proof.
 Next, the lower bound in \eqref{eq:H:property} gives
\begin{align*}
    \frac{I(v'(x))}{x}=\frac{I(f(x))}{x}\overset{z=f(x)}{=}\frac{I(z)}{F(z)}\overset{h=1/z}=\frac{I(\frac1h)}{H(h)}\le \frac{\theta^2p_1(p_1+p_2)}{2}.
\end{align*}
Thus, we can take $C_{v,4}:=\frac{\theta^2p_1(p_1+p_2)}{2}$.
Finally, \eqref{eq:H':property} gives
\begin{align*}
    -\frac{v'(x)}{xv''(x)}=-\frac{f(x)}{xf'(x)}\overset{z=f(x)}{=}-\frac{zF'(z)}{F(z)}\overset{h=1/z}{=}\frac{hH'(h)}{H(h)}\leq p_1.
\end{align*}
Therefore, we can take $C_{v,5}=p_1$.

\subsection{Proof of Theorem \ref{thm:strong:case1}}\label{appendix:sec:case1:3}
\textbf{Step 1: Construction of a locally Lipschitz drift.}
To prepare for the change of measure argument used in Step~3, for
each $c\in\mR$, define
\begin{align*}
    \Psi_c(x):=I(v'(x))+c\frac{v'(x)}{v''(x)}, \quad x\in[\underline{x},\infty).
\end{align*}
By $\lim_{x\downarrow\underline{x}}v'(x)/v''(x)=0$, we set $\Psi_c(\underline{x}):=I(v'(\underline{x}))$. Since $I$ and $v'$ are
continuous, $\Psi_c\in C([\underline{x},\infty))$.
We aim to prove that there exists $c^*\in\mR$ such that $\Psi_{c^*}$ is
locally Lipschitz continuous on $[\underline{x},\infty)$. By Assumption~\ref{assmup:U12}(i) and the definition of $I$ in \eqref{eq:I},
$I$ is locally Lipschitz continuous on $(0,\infty)$. Moreover,
Theorem~\ref{thm:v} gives $v\in C^3((\underline{x},\infty))$ and
$v''<0$.  It therefore remains only to
establish local Lipschitz continuity near $\underline{x}$.

We distinguish two cases. 

Case (i): $v'(\underline{x})>U_1'(b_0)$. The
continuity of $v'$ yields some $\epsilon>0$ such that
$v'(x)>U_1'(b_0)$ for every
$x\in[\underline{x},\underline{x}+\epsilon]$. It then follows that
$\Psi_0(x)=I(v'(x))=b_0$ for all $x\in[\underline{x},\underline{x}+\epsilon]$. Hence,
we may take $c^*=0$.

Case (ii):  $v'(\underline{x})\leq U_1'(b_0)$. Since $v''<0$, we
have $v'(x)\leq v'(\underline{x})$ for $x\geq\underline{x}$.
Consequently, 
$I(v'(x))=(U_1')^{-1}(v'(x))$ for every
$x\in[\underline{x},\infty)$.
By Assumption~\ref{assmup:U12}(i),
$(U_1')^{-1}\in C^2((0,U_1'(b_0)])$, as needed below. The mean value theorem ensures
that it  suffices to prove that $\Psi_c'$ is bounded near
$\underline{x}$. Using $v'=f$ from \eqref{eq:v:integral},
differentiating the inverse relation $F(f(x))=x$, and recalling that
$F(z)=H(1/z)$, we obtain
\begin{align}\label{eq:psi'}
    \Psi_c'(x)&=I'(v'(x))v''(x)-c\frac{v'(x)v'''(x)}{[v''(x)]^2}+c=I'(f(x))f'(x)-c\frac{f(x)f''(x)}{[f'(x)]^2}+c\nonumber\\
    &=\frac{I'(f(x))}{F'(f(x))}+c\frac{f(x)F''(f(x))}{F'(f(x))}+c\overset{z=f(x)}{=}\frac{I'(z)}{F'(z)}+c\frac{zF''(z)}{F'(z)}+c\nonumber\\
    &= -\frac{I'(z)z^2}{H'(1/z)}-\frac{c}{z}\frac{H''(1/z)}{H'(1/z)}-c\overset{h=1/z}{=}-\frac{I'(1/h)}{h^2H'(h)}-ch\frac{H''(h)}{H'(h)}-c.
\end{align}
    Since $h=1/f(x)\downarrow y_0$ as
    $x\downarrow\underline{x}$, the last expression in
    \eqref{eq:psi'} shows that, to prove $\Psi_c'$ is bounded near
    $\underline{x}$, it suffices to prove that
    \begin{align}\label{eq:Psi:quotient}
        \frac{I'(1/h)+ch^3H''(h)}{h^2H'(h)},\quad h\in(y_0,\infty),
    \end{align}
    is bounded near $y_0$.
    Since $H$ satisfies \eqref{eq:extension:u_222}, we have
     \begin{align}\label{eq:H''}
        H''(h)=-\frac{2}{\theta^2h^2}\left[\left(r-\delta\right) hH'(h)+ U_2'(H(h))h^2H'(h) -r H(h) +I\left( \frac1h\right)\right],\quad h\in[y_0,\infty).
    \end{align}
By Theorem~\ref{thm:H}, $H\in C^2([y_0,\infty))$,
$H(y_0)=\underline{x}$, $H'(y_0)=0$, and
$y_0<1/U_1'(r\underline{x})$. Hence,
$I(1/y_0)<r\underline{x}$. Evaluating \eqref{eq:H''} at $y_0$ and
using $H'(y_0)=0$, we obtain
\begin{align*}
    H''(y_0)=\frac{2}{\theta^2y_0^2}\left(r\underline{x}-I(1/y_0)\right)>0.
\end{align*}
Since $H'(y_0)=0$, boundedness of the quotient in
\eqref{eq:Psi:quotient} near $y_0$ requires the numerator
$I'(1/h)+ch^3H''(h)$ to vanish at $h=y_0$. This motivates us to choose 
\begin{align*}
    c^*:=-\frac{I'(1/y_0)}{y_0^3H''(y_0)}.
\end{align*}
It remains to verify that this choice is sufficient.
Telescoping with the term
$\frac{2c^*h^3}{\theta^2}U_2'(H(h))H'(h)$ in the numerator and using
$H(h)\to H(y_0)=\underline{x}$ as $h\downarrow y_0$, we obtain
\begin{align*}
    &\lim_{h\downarrow y_0}
    \frac{I'(1/h)+c^*h^3H''(h)}{h^2H'(h)}\!=\!-\frac{2c^*y_0}{\theta^2}U_2'(\underline{x})
    \!+\!\lim_{h\downarrow y_0}
    \frac{I'(1/h)+c^*h^3
    \left[H''(h)+\dfrac{2}{\theta^2}U_2'(H(h))H'(h)\right]}
    {h^2H'(h)}.
\end{align*}
In the second limit, both the numerator and the denominator vanish at
$h=y_0$ respectively by the definition of $c^*$ and $H'(y_0)=0$ .
Its numerator is $C^1$ by \eqref{eq:H''} and
Assumption~\ref{assmup:U12}(i), while
\begin{align*}
    \left.\frac{\md}{\md h}\left(h^2H'(h)\right)\right|_{h=y_0}
    =y_0^2H''(y_0)>0.
\end{align*}
Thus, l'H\^opital's rule shows that the second limit is finite.
Therefore, the quotient in \eqref{eq:Psi:quotient} with $c=c^*$ is
bounded near $y_0$, and thus $\Psi_{c^*}'$ is bounded near
$\underline{x}$ by the last equality of \eqref{eq:psi'}.

\textbf{Step 2:  H\"older continuity of the diffusion coefficient.}
By Theorem \ref{thm:v}, $\frac{\theta^2v'}{2v''}\in
C^1((\underline{x},\infty))$, thus it suffices to prove that
$\frac{\theta^2v'}{2v''}$ is H\"older continuous of order $\frac12$ near $\underline{x}$.
By \eqref{eq:ode:final}, we have
\begin{align}\label{eq:diffusion:g}  \frac{\theta^2v'(x)}{2v''(x)}=\frac{D(v'(x))+U_2(x)
			+rx v'(x)-\delta v(x)}{v'(x)}=\frac{g(x)}{v'(x)},
\end{align}
where we recall that $g$ is defined in \eqref{eq:defi_g}, and 
\begin{align}\label{eq:g:prime}
    g'(x)=U_2'(x)+(r-\delta)v'(x)+\left[rx- I\left(v'(x)\right)\right]v''(x),\quad x\in(\underline{x},\infty).
\end{align}
For any $\epsilon>0$, there exist constants $C_{\epsilon,1},C_{\epsilon,2}>0$ such that, for $\underline{x}\leq x<y<\underline{x}+\epsilon$,
\begin{align*}
     \frac{\theta^2}{2}\left|\frac{v'(x)}{v''(x)}-\frac{v'(y)}{v''(y)}\right|
    &\leq \left|\frac{1}{v'(x)}-\frac{1}{v'(y)}\right|\left|g(x)\right|+ \left|\frac{1}{v'(y)}\right|\left|g(x)-g(y)\right|\\
    &\leq C_{\epsilon,1}\left\{\int_x^y|v''(z)|\md z+\left(|x-y|+\int_x^y|v''(z)|\md z\right)\right\}\\
    &\leq C_{\epsilon,2}\sqrt{|x-y|}.
\end{align*}
Here the first inequality follows by telescoping with the term $g(x)/v'(y)$; the second follows from $(1/v')'=-v''/(v')^2$, and the fact that all terms in \eqref{eq:g:prime} other than $v''$ are continuous and thus bounded on $[\underline{x},\underline{x}+\epsilon]$; and the last follows from Lemma \ref{lma:limit:v''}.

\textbf{Step 3:  Global well-posedness of the wealth SDE.}     It suffices to show that, for any $T>0$, the following SDE admits a unique strong solution taking values in $[\underline{x},\infty)$:
\begin{align}\begin{cases}\label{eq:wealth:ws}
	\md X^{x;*}(s)=\tilde{b}(X^{x;*}(s))\md s+\tilde{\sigma}(X^{x;*}(s))\md \widehat{W}(s), \quad s\in[0,T],\\
	X^{x;*}(0)=x,
	\end{cases}
\end{align}
where
\begin{align*}
    &\tilde{b}(x):=rx-b^*(x)+c^* \frac{\sigma x}{\theta}\pi^*(x)=\left\{\begin{aligned}
        &rx-\Psi_{c^*}(x),&&x\in[\underline{x},\infty),\\
        &rx-I(v'(\underline{x})), &&x\in(-\infty,\underline{x}),
    \end{aligned}\right. \\
    &\tilde{\sigma}(x):=\sigma\pi^*(x)x
    =\begin{cases}
        -\theta\dfrac{v'(x)}{v''(x)},&x\in(\underline{x},\infty),\\
        0,&x\in(-\infty,\underline{x}].
    \end{cases}
\end{align*}
and $\{\widehat{W}(t)=W(t)+\left(\theta-\frac{c^*}{\theta}\right)t:t\in[0,T]\}$ is a Brownian Motion under $\mathbb{Q}^*_T$ defined by
\begin{align*}
    \frac{\md \mathbb{Q}^*_T}{\md \mathbb{P}}=\exp\left(-\left(\theta-\frac{c^*}{\theta}\right) W_T-\frac12 \left(\theta-\frac{c^*}{\theta}\right)^2T\right).
\end{align*}
Choose a nondecreasing function $\varphi\in C^1_b(\mR)$\footnote{$ C^1_b(\mR)$ denotes the set of all functions in $C^1(\mR)$ such that both the functions themselves and their first derivatives are bounded.} such that
\begin{align*}
    \varphi(x)=x \quad \text{for } |x|\leq 1,\qquad
    \varphi(x)=2 \quad \text{for } x\geq 2,\qquad
    \varphi(x)=-2 \quad \text{for } x\leq -2.
\end{align*}
For $n\in\mathbb{N}$, set $\varphi_n(x):=n\varphi(x/n)$,  $b_n(x):=\tilde{b}(\varphi_n(x))$ and $\sigma_n(x):=\tilde{\sigma}(\varphi_n(x))$.
By Steps~1 and~2, the truncated coefficient $b_n$ is bounded and
Lipschitz continuous, while $\sigma_n$ is bounded and satisfies, for
some $C_n>0$,
\begin{align*}
    |\sigma_n(x)-\sigma_n(y)|^2\leq C_n|x-y|,
    \qquad x,y\in[\underline{x},\infty).
\end{align*}
Hence the following SDE admits a unique strong solution as shown in, for example,  \cite*[Theorem IX.1.7 and IX.3.5]{RevuzYor1999}:
\begin{align*}\begin{cases}
	\md X^{x;*}_n(s)=b_n(X^{x;*}_n(s))\md s+\sigma_n(X^{x;*}_n(s))\md \hat{W}(s),\\
	X^{x;*}_n(0)=x.
	\end{cases}
\end{align*}
Moreover, for $n>x$, since $b_n(\underline{x})=\tilde{b}(\underline{x})>0$ and $\sigma_n(\underline{x})=\tilde{\sigma}(\underline{x})=0$, we have $X^{x;*}_n\geq \underline{x}$.
Define
\begin{align*}
    \tau_n :=\inf\left\{s\geq 0:X^{x;*}_n(s)\geq n\right\}\quad \text{and}\quad
    X^{x;*}(s):=X^{x;*}_n(s), \quad s\in[0,\tau_n].
\end{align*}
Then $X^{x;*}$ satisfies \eqref{eq:wealth:ws} on $[0,\tau_n]$.
Moreover, as $I(v')$ and $\frac{v'}{v''}$ are of linear growth in $x$ by \eqref{eq:property:v}, $X^{x;*}$ will not explode in any finite time; see e.g. \cite*[Theorem IV.2.4]{ikeda1981stochastic}. Therefore, SDE \eqref{eq:wealth:ws} admits a unique global strong solution taking values in $[\underline{x},\infty)$.

\section{Proofs for the Statements in Section~\ref{sec:degenerate}}

\subsection{Proof of Proposition \ref{prop:Hnew:existence}}
\label{appendix:case2:schauder}

\medskip
\noindent
 We first use the Schauder fixed point theorem to construct a solution to ODE~\eqref{eq:Hnew:case2:extended}.
 Define the Banach space
	\begin{align}
	\mathcal{B} := \bigg\{ H \in C((0,\infty)) : \|H\|_\mathcal{B} := \sup_{{h}\in (0,\infty)} \frac{|H(h)|}{h^{-\beta_1} + h^{\frac{1}{1-\gamma}+\beta_2}} < \infty \bigg\},\label{eq:norm:weight}
	\end{align}
     where $\beta_1>0$ and $\beta_2>0$ are two constants satisfying $p_2+1-\beta_1\gamma>0$ and $p_1-\frac{1}{1-\gamma}-\beta_2\gamma >0$.  Let $\mathcal{D}$ be a  closed and convex subset of $\mathcal{B}$ defined as follows:
	\begin{align*}
	\mathcal{D} :=\left\{ H \in \mathcal{B} : |H(h)| \leq  C_{\mathcal D}\left(1+h^{\frac{1}{1-\gamma}}\right) \text{ for all } h\in(0,\infty) \right\},
	\end{align*}
	where $C_{\mathcal D}\ge\underline{x}$ is a constant to be determined later.
	As in Section~\ref{sec:well:hjb}, applying variation of parameters to  \eqref{eq:Hnew:case2:extended}, using the asymptotic condition at $0$ and the growth condition at $\infty$ 
     to eliminate the homogeneous terms, and then integrating by parts motivates the following operator.
	For \(H \in \mathcal{D}\), define
	\begin{align}
	(\Gamma H)(h) :=~& \frac{2}{\theta^2(p_1+p_2)}\bigg\{\int_h^\infty\left[\tilde{I}(z)+(p_1-1)\tilde{U}_2(H(z))z\right]z^{-p_1-1}\mathrm{d}z \, h^{p_1}\nonumber\\
    &+\int_0^h\left[\tilde{I}(z)-(p_2+1)\tilde{U}_2(H(z))z\right]z^{p_2-1}\mathrm{d}z \, h^{-p_2}\bigg\}, \quad h\in (0,\infty).\label{eq:GammaH}
	\end{align}
	    We next verify the conditions for Schauder fixed point theorem:

	    \medskip
\noindent\textbf{1.} \(\Gamma(\mathcal{D}) \subset \mathcal{D}\):
 For any $H\in\mathcal{D}$, the properties of
$\tilde{U}_2$ in \eqref{eq:tildeU2:properties:case2} and
$C_{\mathcal D}\ge\underline{x}$ imply that
\begin{align}\label{eq:tildeu h}
       0\le \tilde{U}_2(H(z))\le C_{2,1}C_{\mathcal D}^{\gamma}\left(1+z^{\frac{1}{1-\gamma}}\right)^\gamma\le C_{2,1}C_{\mathcal D}^{\gamma}\left(1+z^{\frac{\gamma}{1-\gamma}}\right)\quad \text{for all } z\in(0,\infty),
\end{align}
where the last inequality follows from the subadditivity of \(t\mapsto t^\gamma\) for \(0<\gamma<1\). Setting $\widetilde{M}_{1,1}:=\max\{C_{1,1}^{\frac{1}{1-\gamma}},b_0\}$, we obtain from
\eqref{eq:tildeI:growth} that
\begin{align}\label{eq:tildeu i}
       0\le \tilde{I}(z)\le  \widetilde{M}_{1,1}\left(1+z^{\frac{1}{1-\gamma}}\right)\quad \text{for all } z\in(0,\infty).
\end{align}
By \eqref{eq:tildeu h}--\eqref{eq:tildeu i}, $p_2>0$, and $p_1>\frac{1}{1-\gamma}$, $\Gamma H$ is well-defined. By the dominated convergence theorem, $\Gamma H\in C((0,\infty))$.
Moreover, for any $h\in(0,\infty)$,
\begin{align}
    |(\Gamma H)(h)|\le\ & \frac{2}{\theta^2(p_1+p_2)}\Bigg\{ \widetilde{M}_{1,1}\left(\frac{1}{p_1}+\frac{1}{p_1-\frac{1}{1-\gamma}}h^{\frac{1}{1-\gamma}}+\frac{1}{p_2}+\frac{1}{p_2+\frac{1}{1-\gamma}}h^{\frac{1}{1-\gamma}}\right)\nonumber\\
    &\quad+C_{2,1} C_{\mathcal D}^{\gamma}\left[h+\frac{p_1-1}{p_1-\frac{1}{1-\gamma}}h^{\frac{1}{1-\gamma}}+h+\frac{p_2+1}{p_2+\frac{1}{1-\gamma}}h^{\frac{1}{1-\gamma}}\right]\Bigg\}\nonumber\\
    &\leq C_{\Gamma,0}\left(1+C_{\mathcal D}^{\gamma}\right)
    \left(1+h^{\frac{1}{1-\gamma}}\right),\quad h\in(0,\infty).\label{eq:Gamma H}
\end{align}
Here, $C_{\Gamma,0}>0$ is independent of $H$ and $C_{\mathcal D}$.
 Since $0<\gamma<1$, we may choose $C_{\mathcal D}$ sufficiently large so that
\(
 C_{\Gamma,0}\left(1+C_{\mathcal D}^{\gamma}\right)\leq C_{\mathcal D}.
\)
Consequently,
\[
 |(\Gamma H)(h)|\le  C_{\mathcal D}\left(1+h^{\frac{1}{1-\gamma}}\right)\quad \text{for all }h\in(0,\infty).
\]
Thus, $\Gamma H\in \mathcal{D}$.

    \medskip
\noindent
\textbf{2.} \(\Gamma(\mathcal{D})\) is precompact in \(\mathcal{B}\):
For any $H\in \mathcal{D}$, we have
\begin{align*}
    (\Gamma H)'(h)=\ & \frac{2}{\theta^2(p_1+p_2)}\bigg\{p_1\int_h^\infty\left[\tilde{I}(z)+(p_1-1)\tilde{U}_2(H(z))z\right]z^{-p_1-1}\mathrm{d}z \, h^{p_1-1}\\
    \ &\quad -p_2\int_0^h\left[\tilde{I}(z)-(p_2+1)\tilde{U}_2(H(z))z\right]z^{p_2-1}\mathrm{d}z \, h^{-p_2-1}-(p_1+p_2)\tilde{U}_2(H(h))\bigg\}.
\end{align*}
By \eqref{eq:tildeu h} and \eqref{eq:tildeu i}, we obtain
\begin{align}
|(\Gamma H)'(h)|
\leq{}&\frac{2}{\theta^2(p_1+p_2)}\Bigg\{
\widetilde{M}_{1,1}\left[2h^{-1}
+\left(\frac{p_1}{p_1-\frac{1}{1-\gamma}}
+\frac{p_2}{p_2+\frac{1}{1-\gamma}}\right)
h^{\frac{\gamma}{1-\gamma}}\right]\nonumber\\
&\quad+C_{2,1}C_{\mathcal D}^{\gamma}\Bigg[2(p_1+p_2)
+\Bigg(\frac{p_1(p_1-1)}{p_1-\frac{1}{1-\gamma}}
+\frac{p_2(p_2+1)}{p_2+\frac{1}{1-\gamma}}+p_1+p_2\Bigg)
h^{\frac{\gamma}{1-\gamma}}\Bigg]\Bigg\}\nonumber\\
\leq{}&C_{\Gamma,1}\left(h^{-1}+h^{\frac{\gamma}{1-\gamma}}\right),
\quad h\in(0,\infty).\label{eq:GAMMA'}
\end{align}
where $C_{\Gamma,1}>0$ is independent of $H$.

For any sequence $\{Y_n\}\subset\Gamma (\mathcal{D})$ and any $R>0$,  it follows from \eqref{eq:Gamma H} and \eqref{eq:GAMMA'} that $\{Y_n\}$ is uniformly bounded and equicontinuous on $\left[\frac{1}{R},R\right]$. By the Arzelà–Ascoli theorem, there exists a subsequence that converges in \(C\left(\left[\frac{1}{R},R\right]\right)\). A diagonal argument over \(R \in \mathbb{N}\) yields a subsequence (still denoted by $\{Y_n\}$ for convenience) and a function \(Y \in C((0,\infty))\) such that, subsequentially, \(Y_n \to Y\) locally uniformly on $(0,\infty)$. By pointwise convergence and \eqref{eq:Gamma H}, we have \( |Y(h)| \leq C_{\mathcal D}\left(1 + h^{\frac{1}{1-\gamma}}\right)\) for any $h\in(0,\infty)$, which yields $Y \in \mathcal{D}\subset \mathcal{B}$.

Since \(Y_n,Y\in\mathcal D\), for any $\epsilon > 0$, choose $R=R(\epsilon) > 0$ sufficiently large such that, for all $n \in \mathbb{N}$,
	\begin{align}
	\sup_{h\in \left(0,\frac 1R\right)\cup (R,\infty)} \frac{ |Y_{n}(h) - Y(h)| }{h^{-\beta_1} + h^{\frac{1}{1-\gamma}+\beta_2}} \leq 2  C_{\mathcal D}\cdot\!\sup_{h\in \left(0,\frac 1R\right)\cup (R,\infty)} \frac{1+h^{\frac{1}{1-\gamma}}}{h^{-\beta_1} + h^{\frac{1}{1-\gamma}+\beta_2}} < \epsilon,\label{eq:0infty}
	\end{align} 
    where the last inequality follows from \(\beta_1,\beta_2>0\) by choosing \(R\) sufficiently large.
	On \([\frac1R,R]\), the local uniform convergence implies that there exists $N=N(\epsilon)\in \mathbb{N}$ such that, for all $n>N$,
	\begin{align}
	\sup_{{h}\in \left[\frac1R,R\right]} \frac{|Y_{n}(h) - Y(h)|}{h^{-\beta_1} + h^{\frac{1}{1-\gamma}+\beta_2}}  < \epsilon.\label{eq:1RR}
	\end{align}
    Thus, combining \eqref{eq:0infty} and \eqref{eq:1RR}, we obtain
	 $\|Y_n-Y\|_{\mathcal{B}}<\epsilon$ for any $n>N$. Consequently,
    $\{Y_{n}\}$  converges to  $Y$  in $\mathcal{B}$.
    
    \medskip
\noindent
	\textbf{3.} \(\Gamma\) is continuous:
   By the global H\"older property of $\tilde U_2$, for \(H_n,H\in\mathcal D\) with
    \(\|H_n-H\|_{\mathcal B}\to0\) and $h\in(0,\infty)$,
    \begin{align*}
        &\frac{\theta^2(p_1+p_2)}{2}|(\Gamma H_n)(h)-(\Gamma H)(h)|\\
        \leq&\!\int_h^\infty\!(p_1\!-\!1)\!\left|\tilde{U}_2(H_n(z))\!-\!\tilde{U}_2(H(z))\!\right|z^{-p_1}\mathrm{d}z \, h^{p_1}
    \!+\!\int_0^h\!(p_2\!+\!1)\!\left|\tilde{U}_2(H_n(z))\!-\!\tilde{U}_2(H(z))\!\right|z^{p_2}\mathrm{d}z \, h^{-p_2}\\
    \le\ & C_{2,2}\Bigg\{\int_h^\infty(p_1-1)\left|H_n(z)-H(z)\right|^{\gamma_0} z^{-p_1}\mathrm{d}z \, h^{p_1}
    +\int_0^h(p_2+1)\left|H_n(z)-H(z)\right|^{\gamma_0} z^{p_2}\mathrm{d}z \, h^{-p_2}\Bigg\}\\
    \le\ & C_{2,2}\|H_n - H\|_\mathcal{B}^{\gamma_0}\Bigg\{\int_h^\infty(p_1-1)\left(z^{-\beta_1}+z^{\frac{1}{1-\gamma}+\beta_2}\right)^{\gamma_0} z^{-p_1}\mathrm{d}z \, h^{p_1}
    \\&\quad\quad\quad\quad\quad\quad\quad\quad+\int_0^h(p_2+1)\left(z^{-\beta_1}+z^{\frac{1}{1-\gamma}+\beta_2}\right)^{\gamma_0} z^{p_2}\mathrm{d}z \, h^{-p_2}\Bigg\}\\
    \le\ & C_{2,2}\|H_n - H\|_\mathcal{B}^{\gamma_0}\Bigg\{\int_h^\infty(p_1-1)\left(z^{-\beta_1\gamma_0}+z^{\frac{\gamma_0}{1-\gamma}+\beta_2\gamma_0}\right) z^{-p_1}\mathrm{d}z \, h^{p_1}
    \\&\quad\quad\quad\quad\quad\quad\quad\quad+\int_0^h(p_2+1)\left(z^{-\beta_1\gamma_0}+z^{\frac{\gamma_0}{1-\gamma}+\beta_2\gamma_0}\right)z^{p_2}\mathrm{d}z \, h^{-p_2}\Bigg\}\\
    =\ & C_{2,2}\|H_n - H\|_\mathcal{B}^{\gamma_0}\Bigg\{\frac{p_1-1}{p_1+\beta_1\gamma_0-1}h^{-\beta_1\gamma_0+1}+\frac{p_1-1}{p_1-1-\frac{\gamma_0}{1-\gamma}-\beta_2\gamma_0}h^{\frac{\gamma_0}{1-\gamma}+1+\beta_2\gamma_0}
    \\&\quad\quad\quad\quad\quad\quad\quad\quad+\frac{p_2+1}{p_2+1-\beta_1\gamma_0}h^{-\beta_1\gamma_0+1}+\frac{p_2+1}{p_2+1+\frac{\gamma_0}{1-\gamma}+\beta_2\gamma_0}h^{\frac{\gamma_0}{1-\gamma}+1+\beta_2\gamma_0}\Bigg\}\\
    \le\ &  C_{\Gamma,2}\|H_n - H\|_\mathcal{B}^{\gamma_0}\left(h^{-\beta_1}+h^{\frac{1}{1-\gamma}+\beta_2}\right),
    \end{align*}
    where the last inequality uses $0<\gamma_0\le\gamma$, $p_1-\frac{1}{1-\gamma}-\beta_2\gamma >0$ and $p_2+1-\beta_1\gamma>0$, and $C_{\Gamma,2}>0$ is a constant independent of $H_n$ and $H$. Hence, the continuity of $\Gamma$ follows.
        
By the Schauder fixed point theorem, there exists \(H \in \mathcal{D}\) such that \(\Gamma H = H\); its fixed-point identity is precisely \eqref{eq:Hdef}.
The fixed-point identity also improves the interior regularity. Indeed, by the integral representation of \(\Gamma H\) in \eqref{eq:GammaH} together with the continuity of \(\tilde I\) and \(\tilde U_2\circ H\), we obtain \(\Gamma H\in C^1((0,\infty))\). Since $H=\Gamma H$, we have $H\in C^1((0,\infty))$ and hence $\tilde U_2\circ H\in C^1((0,\infty))$. Consequently, the integral representation of \(\Gamma H\) now gives \(\Gamma H=H\in C^2((0,\infty))\). 

Differentiating~\eqref{eq:Hdef} once and twice yields, for any $h\in(0,\infty)$,
    \begin{align}
        H'(h)&=\frac{2}{\theta^2(p_1+p_2)}\bigg\{p_1\int_h^\infty\left[\tilde{I}(z)+(p_1-1)\tilde{U}_2(H(z))z\right]z^{-p_1-1}\mathrm{d}z \, h^{p_1-1}\nonumber\\
   &\quad -p_2\int_0^h\left[\tilde{I}(z)-(p_2+1)\tilde{U}_2(H(z))z\right]z^{p_2-1}\mathrm{d}z \, h^{-p_2-1}-(p_1+p_2)\tilde{U}_2(H(h))\bigg\},\label{eq:Hpdef}\\
        H''(h)&=\frac{2}{\theta^2(p_1+p_2)}\Bigg\{p_1(p_1-1)\int_h^\infty\left[\tilde{I}(z)+(p_1-1)\tilde{U}_2(H(z))z\right]z^{-p_1-1}\mathrm{d}z \, h^{p_1-2}\nonumber\\
        &\quad-p_2(-p_2-1)\int_0^h\left[\tilde{I}(z)-(p_2+1)\tilde{U}_2(H(z))z\right]z^{p_2-1}\mathrm{d}z \, h^{-p_2-2}-(p_1+p_2)\tilde{I}(h)h^{-2}\nonumber\\
       &\quad+\left[ -p_1(p_1-1)+p_2(p_2+1)\right]\tilde{U}_2(H(h))h^{-1}-(p_1+p_2)\tilde{U}_2'(H(h))H'(h)\Bigg\}.\nonumber
    \end{align}
Consequently, for any $h\in (0,\infty)$,
\begin{align*}
    &-\frac{1}{2} \theta^2 h^2H''(h) -\left(r-\delta\right) hH'(h)- \tilde{U}_2'(H(h))h^2H'(h) +r H(h) -\tilde{I}\left( h\right)\\
    =\ &\frac{2}{\theta^2(p_1+p_2)}\left\{-\frac{1}{2} \theta^2\left[ -p_1(p_1-1)+p_2(p_2+1)\right]+(r-\delta)(p_1+p_2)\right\}\tilde{U}_2(H(h))h=0.
\end{align*}
Therefore, the fixed point $H$ solves ODE~\eqref{eq:Hnew:case2:extended} on $(0,\infty)$.

The remaining task is to show that $H$ extends to $h=0$ with $H(0)=b_0/r$ and $H'(0)=0$.
   Using $H\in C((0,\infty))$ and $|H(h)|\le C_{\mathcal D}(1+h^{\frac{1}{1-\gamma}})$, and applying l'H\^opital's rule separately to the two terms on the right-hand side of~\eqref{eq:Hdef}, we obtain
\begin{align*} &\lim_{h\downarrow0}\frac{\int_h^\infty\left[\tilde{I}(z)\!+\!(p_1\!-\!1)\tilde{U}_2(H(z))z\right]z^{-p_1-1}\mathrm{d}z}{h^{-p_1}}
    \!=\!\lim_{h\downarrow0}\frac{-\!\left[\tilde{I}(h)\!+\!(p_1\!-\!1)\tilde{U}_2(H(h))h\right]h^{-p_1-1}}{-p_1h^{-p_1-1}}
    \!=\!\frac{b_0}{p_1},\\
	    &\lim_{h\downarrow0}\frac{\int_0^h\left[\tilde{I}(z)\!-\!(p_2\!+\!1)\tilde{U}_2(H(z))z\right]z^{p_2-1}\mathrm{d}z}{h^{p_2}}
	    \!=\!\lim_{h\downarrow0}\frac{\left[\tilde{I}(h)\!-\!(p_2\!+\!1)\tilde{U}_2(H(h))h\right]h^{p_2-1}}{p_2h^{p_2-1}}
	    \!=\!\frac{b_0}{p_2}.
\end{align*}
Substituting these limits into~\eqref{eq:Hdef} and using $p_1p_2=2r/\theta^2$ yield
$\lim\limits_{h\downarrow0}H(h)=b_0/r$. 

Furthermore, since $\tilde I(z)=b_0$ for $z\in(0,z_0)$, for any $h\in(0,z_0]$, the $\tilde I$-terms in~\eqref{eq:Hpdef} satisfy
\begin{align}
  &p_1\int_h^\infty\tilde{I}(z)z^{-p_1-1}\mathrm{d}z \, h^{p_1-1}
  -p_2\int_0^h\tilde{I}(z)z^{p_2-1}\mathrm{d}z \, h^{-p_2-1}\nonumber\\
  =\ &\left[p_1\int_{z_0}^\infty \tilde{I}(z)z^{-p_1-1}\mathrm{d}z-b_0z_0^{-p_1}\right]h^{p_1-1}.\label{eq:h':simplify}
\end{align}
 For the remaining terms in~\eqref{eq:Hpdef}, l'H\^opital's rule gives
\begin{align*}
    &\lim_{h\downarrow 0}\left(p_1(p_1-1)\frac{\int_h^\infty \tilde{U}_2(H(z))z^{-p_1}\mathrm{d}z}{h^{1-p_1}}+p_2(p_2+1)\frac{\int_0^h\tilde{U}_2(H(z))z^{p_2}\mathrm{d}z}{h^{p_2+1}}-(p_1+p_2)\tilde{U}_2(H(h))\right)\\
    =&p_1\tilde{U}_2(H(0))+p_2\tilde{U}_2(H(0))-(p_1+p_2)\tilde{U}_2(H(0))=0,
\end{align*}
Combining this limit with~\eqref{eq:h':simplify}, we obtain $\lim\limits_{h\downarrow 0}H'(h)=0$. Hence, $H\in C^1([0,\infty))\cap C^2((0,\infty))$.

\subsection{Proof of  Theorem \ref{thm:vnew:case2}}\label{appendix:sec:case2:2}
Let $H$ be the function obtained in Theorem \ref{thm:Hnew:case2}. 
As in the proof of Theorem \ref{thm:v}, define
\begin{align*}
    F(z):=H\left(\frac{1}{z}\right), \quad z\in \left(0,\infty\right).
\end{align*}
Then $F\in C^2(\left(0,\infty\right))$ satisfies \eqref{eq:F} (with $A=\infty$) and $F'(z)<0$ for any $z\in \left(0,\infty\right)$.
Next, let $f:(\underline{x},\infty)\to \left(0,\infty\right)$ be the inverse function of $F$ satisfying $F(f(x))=x$ for any  $x\in (\underline{x},\infty)$.  Then, $f\in  C^2((\underline{x},\infty))$ satisfies \eqref{eq:f}.  Moreover,
\begin{align}
    \lim_{x\downarrow \underline{x}}\frac{f(x)}{(x-\underline{x})^{-\frac{1}{p_1}}}&=\lim_{x\downarrow \underline{x}}\frac{f(x)}{(F(f(x))-\underline{x})^{-\frac{1}{p_1}}}\overset{z=f(x)}{=}\lim_{z\uparrow \infty}\frac{z}{(F(z)-\underline{x})^{-\frac{1}{p_1}}}\nonumber\\&\overset{h=1/z}{=}\left(\lim_{h\downarrow 0}\frac{H(h)-\underline{x}}{h^{p_1}}\right)^{\frac{1}{p_1}}
    =\left(\lim_{h\downarrow 0}\frac{H' 
    (h)}{p_1h^{p_1-1}}\right)^{\frac{1}{p_1}}=\left(\frac{l_H}{p_1}\right)^{\frac{1}{p_1}}.\label{eq:limit:v'}
\end{align}
Combining this with the fact that $p_1>1$, $f$ is integrable at $\underline{x}$. 
Define
\begin{align*}
    v(x):=V(\underline{x})+\int_{\underline{x}}^xf(z)\md z, \quad x\in[\underline{x},\infty).
\end{align*}
Then $v\in C([\underline{x},\infty))\cap C^3((\underline{x},\infty))$ because $f\in  C^2((\underline{x},\infty))$, and  the first limit in \eqref{eq:property:v:case2} holds. We then show that $v$ satisfies \eqref{eq:ode:final:case2}.  As in the proof of Theorem \ref{thm:v}, for any $x>\underline{x}$ and $\epsilon\in(0,x-\underline{x})$, integrating both sides of the ODE in \eqref{eq:f} from $\underline{x}+\epsilon$ to $x$, \eqref{eq:v:case2:epsilon} remains valid.
Note that
\begin{align*}
    &\lim_{\epsilon\downarrow 0}\frac{[v'(\underline{x}+\epsilon)]^2}{v''(\underline{x}+\epsilon)}=\lim_{\epsilon\downarrow 0}\frac{[f(\underline{x}+\epsilon)]^2}{f'(\underline{x}+\epsilon)}\overset{z=f(\underline{x}+\epsilon)}{=}\lim_{z\to \infty}z^2F'(z)\overset{1/h=z}{=}-\lim_{h\downarrow0}H'(h)=0,\\
    &\lim_{\epsilon \downarrow 0}\epsilon v'(\underline{x}+\epsilon)= \lim_{\epsilon \downarrow 0}\epsilon f(\underline{x}+\epsilon)\overset{z=f(\underline{x
    }+\epsilon)}{=}\lim_{z\to \infty}z(F(z)-\underline{x})\overset{1/h=z}{=}\lim_{h\downarrow0}\frac{H(h)-\underline{x}}{h}=0.
\end{align*}
Using the fact that, by \eqref{eq:I} and \eqref{eq:D}, $D(v'(\underline{x}+\epsilon))=U_1(b_0)-b_0v'(\underline{x}+\epsilon)=U_1(b_0)-r\underline{x}v'(\underline{x}+\epsilon)$ for $\epsilon$ small enough, and letting $\epsilon\to0$ in \eqref{eq:v:case2:epsilon}, we obtain
\begin{align*}
     (\delta-r+\theta^2) (v(x)-V(\underline{x}))&= U_2(x)- U_2(\underline{x}) +rxv'(x) -r(v(x)-V(\underline{x}))+D(v'(x))-U_1(b_0)\nonumber\\
     &\quad +\frac{1}{2} \theta^2\left(2(v(x)-V(\underline{x}))-\frac{[v'(x)]^2}{v''(x)}\right).
\end{align*}
Combining this with
$V(\underline{x})
=\frac{1}{\delta}U_1(b_0)
+\frac{1}{\delta}U_2(\underline{x})$
in \eqref{eq:V:underline:x}, we conclude that
$v$ satisfies \eqref{eq:ode:final:case2}.

It remains to prove the remaining assertions in \eqref{eq:property:v:case2}. Following the same arguments as in the proof of
Theorem~\ref{thm:v}, and using
\eqref{eq:Hnew:lower-bound} and
\eqref{eq:Hnew:derivative-bound}, we obtain
\[
v'(x)>0,\qquad v''(x)<0,\qquad
I(v'(x))\leq
\frac{\theta^2p_1(p_1+p_2)}{2}x,\qquad
-\frac{v'(x)}{xv''(x)}\leq p_1,
\quad x>\underline{x}.
\]
Thus, one may take
$C_{v,7}=\frac{\theta^2p_1(p_1+p_2)}{2}$ and $C_{v,8}=p_1$. Moreover,
\begin{align*}
    &\limsup_{x\to\infty}\frac{v'(x)}{x^{\gamma-1}}= \limsup_{x\to\infty}\frac{f(x)}{x^{\gamma-1}}\overset{z=f(x)}{=} \limsup_{z\to0}\frac{z}{[F(z)]^{\gamma-1}}\\=\ &\limsup_{z\to0}\left(\frac{z^{\frac{1}{\gamma-1}}}{F(z)}\right)^{\gamma-1}\overset{h=1/z}=\limsup_{h\to\infty}\left(\frac{H(h)}{h^{\frac{1}{1-\gamma}}}\right)^{1-\gamma}<\infty.
\end{align*}
Here, the second equality is justified by
$\lim\limits_{x\to\infty}f(x)=0$, which follows from $f=F^{-1}$ and
$
\lim\limits_{z\downarrow0}F(z)
=\lim\limits_{z\downarrow0}H(1/z)=\infty$. The preceding estimate for $v'$ directly yields
$v(x)\leq C_{v,6}x^\gamma$ for some $C_{v,6}>0$ and all
$x\geq\underline{x}$.

\subsection{Proof of Theorem \ref{thm:strong:case2}}\label{appendix:sec:case2:3}

By \eqref{eq:property:v:case2}, the drift coefficient
$x\mapsto rx+\sigma\theta\pi^*(x)x-b^*(x)$ and the diffusion
coefficient $x\mapsto\sigma\pi^*(x)x$ of
SDE~\eqref{eq:verify:wealth} have linear growth and vanish at
$\underline{x}$. Thus, as in the proof of
Theorem~\ref{thm:strong:case1}, it remains only to verify their local
Lipschitz continuity on $[\underline{x},\infty)$; the same truncation argument then yields a unique global strong
solution taking values in $[\underline{x},\infty)$.

First, since
$\lim\limits_{x\downarrow\underline{x}}v'(x)=\infty$, the first-case
argument in Step~1 of the proof of Theorem~\ref{thm:strong:case1}
shows that $x\mapsto rx-b^*(x)$ is locally Lipschitz continuous on
$[\underline{x},\infty)$.

Second, 
$\sigma\pi^*(x)x=-\theta\frac{v'(x)}{v''(x)}$ is locally Lipschitz continuous
on $(\underline{x},\infty)$ by the regularity of $v$. It therefore
suffices to establish its local Lipschitz continuity at
$\underline{x}$. With
\eqref{eq:ode:final:case2} in place of \eqref{eq:ode:final},
Equation~\eqref{eq:diffusion:g} remains valid.
Since $\lim\limits_{x\downarrow\underline{x}}v'(x)=\infty$,
\eqref{eq:I} yields some $\epsilon>0$ such that
$I(v'(x))=b_0=r\underline{x}$ for
$x\in(\underline{x},\underline{x}+\epsilon]$.
Differentiating \eqref{eq:diffusion:g} and using \eqref{eq:g:prime}, we obtain
\begin{align}
\frac{\md}{\md x}\left(\frac{\theta^2v'(x)}{2v''(x)}\right)
&=\frac{\md}{\md x}\left(\frac{g(x)}{v'(x)}\right)
=\frac{g'(x)}{v'(x)}-\frac{v''(x)g(x)}{[v'(x)]^2}=\frac{g'(x)}{v'(x)}-\frac{\theta^2}{2}\nonumber\\
&=r-\delta-\frac{\theta^2}{2}
+\frac{U_2'(x)}{v'(x)}
+\frac{r(x-\underline{x})v''(x)}{v'(x)},
\quad x\in(\underline{x},\underline{x}+\epsilon].
\label{eq:diffusion:derivative:case2}
\end{align}
Since \eqref{eq:limit:v'} implies that
$\lim\limits_{x\downarrow\underline{x}}f(x)=\infty$, we have
\begin{align*}
    &\lim_{x\downarrow \underline{x}}\frac{f'(x)}{(x-\underline{x})^{-\frac{1}{p_1}-1}}=\lim_{x\downarrow \underline{x}}\frac{f'(x)}{(F(f(x))-\underline{x})^{-\frac{1}{p_1}-1}}\overset{z=f(x)}{=}\lim_{z\uparrow \infty}\frac{(F(z)-\underline{x})^{\frac{1}{p_1}+1}}{F'(z)}\\\overset{h=1/z}{=}\ &-\lim_{h\downarrow0}\frac{(H(h)-\underline{x})^{\frac{1}{p_1}+1}}{h^2H'(h)}=-\left(\lim_{h\downarrow0}\frac{H(h)-\underline{x}}{h^{p_1}}\right)^{\frac{1}{p_1}+1}\lim_{h\downarrow0}\frac{h^{p_1+1}}{h^2H'(h)}=-\frac{l_H^{\frac{1}{p_1}}}{p_1^{\frac{1}{p_1}+1}}.
\end{align*}
Since $v'=f$ and $v''=f'$, combining the preceding limit with
\eqref{eq:limit:v'} yields
\begin{align*}
    \lim_{x\downarrow\underline{x}}\frac{v'(x)}{(x-\underline{x})v''(x)}=-p_1.
\end{align*}
By the preceding limit and the continuity of $U_2'$, the right-hand
side of \eqref{eq:diffusion:derivative:case2} is bounded near
$\underline{x}$. Hence, $x\mapsto\sigma\pi^*(x)x$ is locally
Lipschitz continuous at $\underline{x}$. This completes the proof.

\section{Difficulties of Alternative Methods} \label{sec:discussion}

In this section, we discuss whether classical methods commonly used for constrained problems can be applied to our model. In particular, we consider the martingale  method, and the viscosity solution method, and show that these methods are not directly applicable in our setting. The underlying difficulties are discussed in detail below.

\subsection{The Martingale  Method}
For simplicity, in this section, we assume that $b_0=0$ and $\alpha=1$. Moreover, we consider the investment strategy in terms of the  amount rather than the proportion of wealth. As can be seen from the main analysis, our method applies to both formulations and leads to identical results. For any $x\geq\underline{x}$, set $y:=x-\underline{x}\geq0$ and define the shifted wealth process
	$Y^{y;u}:=X^{x;u}-\underline{x}$. Then $Y^{y;u}$ satisfies
	\begin{align*}
		\begin{cases}
			&\mathrm{d} Y^{y;u}(s)=rY^{y;u}(s)\md s+r\underline{x}\md s +\sigma \theta\pi(s)\mathrm{d}s-b(s)\mathrm{d}s+\sigma\pi(s)\mathrm{d}W(s),\\
			&Y^{y;u}(0)=y.
		\end{cases}
	\end{align*}
	Then $X^{x;u}\geq \underline{x}$ is equivalent to the  borrowing constraint $Y^{y;u}\geq 0$. Define the admissible  set
 \begin{align*}
		\tilde{\mathcal{A}}(y):&=\bigg\{u=(b,\pi):\int_0^T(|b(s)|+|\pi(s)|^2)\md s<\infty, \quad \text{for all } T\geq0, \text{ a.s.}; \\ &\quad \quad\quad\quad\quad\quad \quad b(s)\geq 0\,\, \text{and}\,\,{Y^{y;u}(s)}\geq 0, \quad \text{for a.e. } s\geq0, \text{ a.s.} \bigg\}.
	\end{align*}
 Accordingly, the shifted problem is a state-constrained control problem with the standard borrowing constraint $Y^{y;u}\geq0$:
	\begin{align}\label{eq:value:shifted}
		\tilde{V}(y):=V(y+\underline{x})=\sup_{(b,\pi)\in\tilde{\mathcal{A}}(y)}\mathbb{E}\bigg[\int_0^{\infty}\mathrm{e}^{-\delta s}U_1(b(s))\md s+\int_0^{\infty}\mathrm{e}^{-\delta s}U_2(Y^{y;u}(s)+\underline{x})\md s\bigg].
	\end{align}
 Under this borrowing constraint, a benefit process $b$ can be supported by some portfolio strategy while maintaining $Y^{y;u}\geq0$ if and only if the following inequality holds; see \cite*{he1993labor} and \cite*{el1998optimization}:
	$$
	\sup _{\mathcal{D} \in \mathcal{N} \mathcal{I}} \mathbb{E}\left[\int_0^\infty H_t \mathcal{D}_t\left(b_t-r\underline{x}\right) \md t\right] \leq y,
	$$
where $\mathcal{N} \mathcal{I}$ denotes the set of all nonnegative, non-increasing, right-continuous processes with left limits and starting at $1$, and $H$ is the pricing kernel satisfying:
\begin{align*}
    \md H_t=-H_t\left(r\md t+\theta \md W(t)\right).
\end{align*}
 In the case $U_2 = 0$, the objective in \eqref{eq:value:shifted} depends solely on the benefit process $b(\cdot)$, and the problem admits a classical martingale dual formulation. The dual control variable $\mathcal{D} \in \mathcal{NI}$ acts as a state-price density distortion, and the primal problem can be recovered from the solution to the dual HJB equation via standard convex duality arguments.

However, when $U_2\not\equiv0$, the above benefit-only dual construction is no longer sufficient because the objective in \eqref{eq:value:shifted} also depends explicitly on the wealth process. Under an appropriate transversality condition, eliminating the wealth state through its conditional-valuation representation would make the resulting benefit-only objective depend nonlocally on the future path of $b(\cdot)$. Consequently, the classical benefit-only martingale reduction does not apply immediately.

\subsection{The Viscosity Solution Method}
We consider the case $\alpha x_0>b_0/r$. A viscosity characterization of the value function provides a natural starting point, but two further steps are needed for the conclusions of this paper: establishing classical regularity and obtaining the estimates required by the verification theorem.

First, even supposing that $V$ has been shown to be a strictly increasing, concave $C^1$ viscosity solution, the improvement to $C^2$ requires an additional argument for the present HJB equation. Define
\[
K(x):=\delta V(x)-U_2(x)-D(V'(x))-rxV'(x).
\]
For a classical solution with $V'>0$ and $V''<0$, the HJB equation gives
\[
K(x)=-\frac{\theta^2(V'(x))^2}{2V''(x)}>0,
\qquad
V''(x)=-\frac{\theta^2(V'(x))^2}{2K(x)}.
\]
To obtain classical regularity from this rearrangement, one must first establish that $K$ is strictly positive throughout $(\underline{x},\infty)$. Continuity of $V$ and $V'$ alone does not exclude zeros of $K$, at which the displayed expression for $V''$ would become singular. Moreover, an identity valid at points of twice differentiability does not by itself imply $C^2$ regularity. A separate argument is therefore needed to exclude such zeros and to justify the passage from the viscosity equation to the classical equation.

Second, even if
\[
V\in C^2((\underline{x},\infty)),
\qquad V'>0,\qquad V''<0,
\]
has been established, this interior regularity alone does not provide the boundary properties required by Theorem~\ref{thm:verification}. Thus, obtaining interior $C^2$-regularity through a viscosity solution method is far from sufficient for verification; precise boundary and growth estimates are also
needed, and our approach works directly with the HJB equation, obtaining these estimates in
the course of constructing a classical solution. In particular, one must also prove that $V\in C^1([\underline{x},\infty))$ and that
\[
U_1'(r\underline{x})<V'(\underline{x})<\infty,
\qquad
\lim_{x\downarrow\underline{x}}V''(x)=-\infty.
\]
Together with the HJB equation, these properties yield the more precise estimate in Lemma~\ref{lma:limit:v''}:
\[
\lim_{x\downarrow\underline{x}}
(x-\underline{x})[V''(x)]^2
=
\frac{\theta^2[V'(\underline{x})]^2}
{4\bigl[r\underline{x}-I(V'(\underline{x}))\bigr]}.
\]
This estimate describes the behavior of the diffusion coefficient
\[
\sigma\pi^*(x)x=-\theta\frac{V'(x)}{V''(x)}
\]
near the boundary and enters the analysis of the wealth process under the proposed optimal strategy. Moreover, the verification theorem requires global growth estimates, including
\[
\sup_{x>\underline{x}}\frac{I(V'(x))}{x}<\infty,
\qquad
\sup_{x>\underline{x}}
\left(-\frac{V'(x)}{xV''(x)}\right)<\infty,
\]
as well as the prescribed growth bound on $V$. One must further establish that the SDE associated with the feedback controls admits a unique strong solution satisfying the state constraint. These conclusions require additional analysis beyond interior $C^2$-regularity.

Our construction obtains the required derivative estimates through the auxiliary function $H$. The required properties of $H$ and the resulting regularity, boundary behavior, and growth estimates are established in Theorems~\ref{thm:H} and~\ref{thm:v}, respectively.

\section{Further Results under CRRA Utility}\label{sec:crra-results}
Throughout this appendix, assume that the utility functions in \eqref{eq:valuefunction:def} satisfy $U_1(b)=b^\gamma/\gamma$ and $U_2(x)=\eta x^\gamma/\gamma$ for some $\gamma\in(0,1)$ and $\eta>0$.

\subsection{Homogeneity and Scale Normalization}\label{sec:Homogeneity}
We first show that the value function is jointly homogeneous in current wealth, initial wealth, and the minimum benefit. For any $x_0>0$, $b_0\in(0,rx_0]$, $\alpha\in(0,1]$, and $x\geq \max\left\{\alpha x_0,\frac{b_0}{r}\right\}$, define
\begin{align*}
V(x;x_0,b_0):=\sup_{(b,\pi)\in\mathcal{A}(x;x_0,b_0)}\mathbb{E}\bigg[\int_0^{\infty}\mathrm{e}^{-\delta s}U_1(b(s))\md s+\int_0^{\infty}\mathrm{e}^{-\delta s}U_2(X^{x;u}(s))\md s\bigg],
\end{align*}
where
\begin{align*}
\mathcal{A}(x;x_0,b_0):=\bigg\{u=(b,\pi):&\int_0^T(|b(s)|+|\pi(s)|^2)\md s<\infty\quad \text{for all } T\geq0, \text{ a.s.}; \\
&b(s)\geq b_0\,\, \text{and}\,\, X^{x;u}(s)\geq \max\left\{\alpha x_0,\frac{b_0}{r}\right\}
\quad \text{for a.e. } s\geq0, \text{ a.s.}\bigg\}.
\end{align*}
The optimal benefit rate and investment strategy are given by
\begin{equation}\label{eq:homogeneity:strategies}
b^*(x;x_0,b_0)=
\max\left\{b_0,\left[V_x(x;x_0,b_0)\right]^{\frac{1}{\gamma-1}}\right\},\quad 
\pi^*(x;x_0,b_0)
=
-\dfrac{(\mu-r_1)V_x(x;x_0,b_0)}{\sigma^2xV_{xx}(x;x_0,b_0)},
\end{equation}

For any $k>0$, the linearity of the state dynamics yields $X^{kx;(kb,\pi)}=kX^{x;(b,\pi)}$. Hence, the map $(b,\pi)\mapsto(kb,\pi)$ is a bijection from $\mathcal{A}(x;x_0,b_0)$ onto $\mathcal{A}(kx;kx_0,kb_0)$. The homogeneity of $U_1$ and $U_2$ then gives
\begin{align*}
V(kx;kx_0,kb_0)=k^\gamma V(x;x_0,b_0).
\end{align*}
Differentiating this identity yields
\begin{equation}\label{eq:homogeneity:derivative-scaling}
V_x(kx;kx_0,kb_0)=k^{\gamma-1}V_x(x;x_0,b_0),\quad
V_{xx}(kx;kx_0,kb_0)=k^{\gamma-2}V_{xx}(x;x_0,b_0).
\end{equation}
It follows directly from \eqref{eq:homogeneity:strategies} and \eqref{eq:homogeneity:derivative-scaling} that
\begin{align*}
b^*(kx;kx_0,kb_0)
&=\max\left\{kb_0,k\left[V_x(x;x_0,b_0)\right]^{\frac{1}{\gamma-1}}\right\}=k b^*(x;x_0,b_0),\\
\pi^*(kx;kx_0,kb_0)
&= -\frac{(\mu-r_1)k^{\gamma-1}V_x(x;x_0,b_0)}{\sigma^2kx\,k^{\gamma-2}V_{xx}(x;x_0,b_0)}
=\pi^*(x;x_0,b_0).
\end{align*}
Taking $k=1/x_0$ yields
\begin{align*}
 \frac{V(x;x_0,b_0)}{x^\gamma}
=\frac{V\left(\frac{x}{x_0};1,\frac{b_0}{x_0}\right)}{(x/x_0)^\gamma},\quad
\frac{b^*(x;x_0,b_0)}{x}
=\frac{b^*\left(\frac{x}{x_0};1,\frac{b_0}{x_0}\right)}{x/x_0},\quad
 \pi^*(x;x_0,b_0)
=\pi^*\left(\frac{x}{x_0};1,\frac{b_0}{x_0}\right).
\end{align*}
Thus, both $\pi^*$ and $b^*/x$ are homogeneous of degree zero under the joint scaling of all monetary quantities. Holding $b_0/x_0$ fixed, the numerical curves of $b^*/x$ as functions of $x/x_0$ can therefore be computed under the normalization $x_0=1$ without loss of generality.

\subsection{Exponentially Growing Guarantees}\label{sec:extend}
\begingroup
\setlength{\abovedisplayskip}{5pt plus 1pt minus 1pt}
\setlength{\belowdisplayskip}{5pt plus 1pt minus 1pt}
\setlength{\abovedisplayshortskip}{3pt plus 1pt}
\setlength{\belowdisplayshortskip}{5pt plus 1pt minus 1pt}
We next allow the minimum benefit and wealth floor to grow at a common rate $0<\nu<r$. The constant constraints are then replaced by
\begin{align*}
b(t)\geq b_0\mathrm{e}^{\nu t}\quad\text{and}\quad
X(t)\geq \alpha x_0\mathrm{e}^{\nu t}\quad\text{for a.e. }t\geq0,\ \text{a.s.},
\end{align*}
where $b_0\in(0,(r-\nu)x_0]$ and $\alpha\in(0,1]$.
Introduce $\tilde{X}(t):=\mathrm{e}^{-\nu t}X(t)$ and $\tilde{b}(t):=\mathrm{e}^{-\nu t}b(t)$. Applying this transform to the wealth dynamics \eqref{eq:wealth:statics}, we obtain
\begin{align*}
\left\{
\begin{aligned}
    &\mathrm{d} \tilde{X}(t)=[r-\nu+\sigma \theta\pi(t)]\tilde{X}(t)\mathrm{d}t-\tilde{b}(t)\mathrm{d}t+\sigma\pi(t)\tilde{X}(t)\mathrm{d}W(t),\\
    &\tilde{X}(0)=x_0.
    \end{aligned}
    \right.
\end{align*}
 Applying the same transform to the objective \eqref{eq:valuefunction:def} and using CRRA homogeneity gives the equivalent problem
\begin{align*}
& \sup_{(\tilde{b},\pi)}\mathbb{E}\bigg[
\int_0^\infty \mathrm{e}^{-(\delta-\nu\gamma)t}
\left(U_1(\tilde{b}(t))+U_2(\tilde{X}(t))\right)\md t\bigg],\\
&\text{subject to}\quad \tilde{b}(t)\geq b_0\quad\text{and}\quad
\tilde{X}(t)\geq \alpha x_0\quad\text{for a.e. }t\geq0,\ \text{a.s.}
\end{align*}
 Thus, the problem with growing guarantees reduces to the baseline problem with $r$ and $\delta$ replaced by $r-\nu$ and $\delta-\nu\gamma$, respectively.
\endgroup

\end{document}